\documentclass[12pt, english]{article}

\usepackage[utf8]{inputenc}
\usepackage{authblk}
\usepackage{mathptmx}

\usepackage{rotating}
\usepackage{amsmath}
\usepackage{amsfonts, amssymb}
\usepackage{latexsym}
\usepackage{amsthm}
\usepackage[usenames]{color}
\usepackage{epsfig}
\usepackage{xcolor}
\usepackage[citebordercolor=red,urlbordercolor=red,linkbordercolor=blue]{hyperref}
\usepackage{bbm}
\usepackage{enumerate}
\usepackage{xy}
\xyoption{all}

\usepackage{dynkin-diagrams}

\newcommand{\rank}{{\rm rank}}

\providecommand{\fp}{{\mathfrak p}}
\providecommand{\fo}{{\mathfrak o}}

\providecommand{\g}{{\mathfrak g}}
\providecommand{\fg}{{\mathfrak g}}

\providecommand{\fk}{{\mathfrak k}}

\providecommand{\fa}{{\mathfrak a}}

\providecommand{\fh}{{\mathfrak h}}

\providecommand{\fz}{{\mathfrak z}}

\providecommand{\fsu}{{\mathfrak{su}}}
\providecommand{\fsl}{{\mathfrak{sl}}}
\providecommand{\fso}{\mathfrak{so}}
\providecommand{\fsp}{\mathfrak{sp}}

\providecommand{\comment}[1]{}

\renewcommand{\emptyset}{\varnothing}

\providecommand{\PP}{\mathbb{P}}

\providecommand{\face}{\mbox{face}}

\providecommand{\im}{\mbox{im}}

\newtheorem{theorem}{Theorem}[section]
\newtheorem{lemma}[theorem]{Lemma}
\newtheorem{claim}{Claim}
\newtheorem{proposition}[theorem]{Proposition}
\newtheorem{corollary}[theorem]{Corollary}

\newtheorem{ex-conjecture}[theorem]{Ex-Conjecture}

\theoremstyle{definition}
\newtheorem{definition}[theorem]{Definition}

\theoremstyle{remark}
\newtheorem{remark}[theorem]{Remark}

\providecommand{\x}{{\bf x}}
\providecommand{\y}{{\bf y}}

\providecommand{\tr}{\text{tr\,}}

\providecommand{\aff}{\text{Aff}~}

\renewcommand{\phi}{\varphi}

\providecommand{\beq}{\begin{equation}}
\providecommand{\eeq}{\end{equation}}
\providecommand{\beqa}{\begin{eqnarray}}
\providecommand{\eeqa}{\end{eqnarray}}

\providecommand{\cs}{{\cal S}}

\providecommand{\fp}{{\mathfrak{p}}}
\providecommand{\fk}{{\mathfrak{k}}}

\providecommand{\fa}{{\mathfrak{a}}}

\providecommand{\fg}{{\mathfrak{g}}}
\providecommand{\fe}{{\mathfrak{e}}}
\providecommand{\ff}{{\mathfrak{f}}}
\newcommand{\faut}{\mathfrak{aut}}
\newcommand{\flie}{\mathfrak{lie}}
\newcommand{\fder}{\mathfrak{der}}
\renewcommand{\H}{{\mathbb{H}}}
\providecommand{\R}{{\mathbb{R}}}
\providecommand{\D}{{\mathbb{D}}}
\providecommand{\C}{{\mathbb{C}}}
\providecommand{\Oct}{{\mathbb{O}}}

\providecommand{\S}{{\mathbb{S}}}
\providecommand{\N}{{\mathbb{N}}}
\providecommand{\E}{{\mathbb{E}}}
\providecommand{\Z}{{\mathbb{Z}}}
\providecommand{\ad}{\mbox{ad }}

\providecommand{\tri}{{\triangle}}
\providecommand{\diag}{{diag}}

\providecommand{\join}{\vee}
\providecommand{\meet}{\wedge}

\providecommand{\iso}{\simeq}

\providecommand{\union}{\cup}
\providecommand{\intersect}{\cap}
\providecommand{\interior}{{\rm int~}}

\providecommand{\face}{{\rm Face}}
\providecommand{\eface}{{\rm ExpFace}}
\providecommand{\ad}{{\rm ad}}

\providecommand{\cone}[1]{{\mathrm{Cone}~}(#1)}

\providecommand{\aut}{\mathrm{Aut~}}
\providecommand{\Aut}{\mathrm{Aut~}}
\providecommand{\conv}{{\rm Conv~}}
\providecommand{\lin}{{\mathrm{lin~}}}

\providecommand{\hbcommentoff}[1]{{}}

\providecommand{\tempout}[1]{{}}

\begin{document}

\title{Strongly symmetric spectral convex bodies are Jordan algebra
  state spaces} 


\author[1]{Howard Barnum\thanks{\texttt{hnbarnum@aol.com}}}
\affil[1]{Los Alamos, New Mexico, USA}
\author[2]{Joachim Hilgert\thanks{\texttt{hilgert@math.upb.de}}}
\affil[2]{Department of Mathematics, University of Paderborn, Germany}

\date{\today}


\maketitle



\thispagestyle{empty}

\pagestyle{myheadings}

\begin{abstract}
We show that the finite-dimensional convex compact sets having the properties of \emph{spectrality} and \emph{strong symmetry} are precisely the normalized state spaces of finite-dimensional simple Euclidean Jordan algebras and the simplices.  Various assumptions are known characterizing complex quantum state spaces among the Jordan state spaces,  
which combine with this theorem to give simple characterizations of finite-dimensional quantum state
space. 

Spectrality and strong symmetry arose in the study of ``general probabilistic theories" (GPTs), in which convex compact sets are considered as state spaces of abstractly conceivable physical systems, though not necessarily ones corresponding to actual physics.  We discuss some implications of our result---which is purely convex geometric in nature---for such theories.   A major concern in the study of such theories, and also in the theory of operator algebras, has been the characterization of the state spaces of finite and infinite-dimensional Jordan algebras, and the 
characterization of standard quantum theory over complex Hilbert spaces, or the state spaces of von Neumann or $C^*$-algebras, within the class of Jordan-algebraic state spaces.   In the finite-dimensional case, our characterization of simple Jordan-algebraic state spaces and classical (i.e. simplicial) state spaces, can serve as an alternative to existing characterizations of Jordan-algebraic systems, for example via the properties of positive projections associated with faces \cite{ASBook2, Araki80}. 
 
While our result is purely convex-geometric in nature, it  has strong implications for work relating geometric properties of the state spaces of systems to physical and information processing characteristics of GPT theories.   
It shows that some generalizations of important aspects of quantum and classical thermodynamics to theories satisfying natural postulates, e.g. in  \cite{KrummEtAlThermo, ChiribellaScandoloDiagonalization}, apply to a narrower class of theories than might have been hoped, already relatively close to complex quantum theory since their systems are Jordan algebraic.   Sorkin's notion of irreducibly $k$-th order interference, involving $k$ or more possibilities, generalizing the interference between two possibilities in quantum theory, has been studied in the GPT framework and looked for in experiments.
Our result shows that the assumption of no higher-order ($k \ge 3$) interference, used along with 
spectrality and strong symmetry to characterize the same class of Jordan-algebraic convex sets in \cite{BMU}, was superfluous.   
It also implies that  \cite{LeeSelbyGrover}'s
extension, on the assumption that interference is no greater than some fixed maximal degree $k$, of the 
important $\Omega(\sqrt{N})$ lower bound on the quantum black-box query complexity of searching $N$ possibilities for one having a desired property (which is achieved by Grover's celebrated quantum algorithm), to a class of GPTs satisfying
certain postulates  allowing the formulation of a generalized notion of query algorithm actually applies in the Jordan-algebraic setting where higher-order interference ($k \ge 3$) is not possible.

On the other hand, our work suggests that whether one is interested in investigating the physical and information properties of possible probabilistic theories or in the geometric properties of convex compact sets, it is worth focusing on the implications of postulates weaker than the conjunction of strong symmetry and spectrality.   

\end{abstract}

\section{Introduction}

In \cite{BMU}, the normalized state spaces of simple
finite-dimensional Euclidean Jordan algebras, and the simplices, were
characterized as the unique finite-dimensional compact convex sets
satisfying three properties: spectrality, strong symmetry, and the
absence of higher-order interference.  In this paper, we show that
the same class of convex compact sets is characterized by 
 the first two of these properties:   

\begin{theorem}\label{theorem: main}
A finite-dimensional convex compact set is spectral and strongly
symmetric if and only if it is a simplex or affinely isomorphic to the
space of normalized states of a simple Euclidean Jordan algebra.
\end{theorem}

Simplices are the state spaces of particular nonsimple Euclidean Jordan
algebras, namely products of several copies of the trivial
(one-dimensional) Euclidean Jordan algebra, so this theorem implies the claim 
in our title.

Spectrality (in this convex sense) and strong symmetry are notions originating in an area of
research, now often called ``general probabilistic theories''
(GPTs), that considers convex compact sets as an abstract notion of
state space of a system, physical or otherwise, on which one can make
measurements, attempt to control the dynamics, etc.; a state encodes
the probabilities of the results of all measurements that it is
possible to make on a system.  Roughly speaking, a system specified by a 
convex compact state space $\Omega$ is spectral if every state is a
convex combination of pure (i.e. extremal, in the convex sense) states
that are perfectly distinguishable from each other via some
measurement it is possible to perform on the system.  A system is
strongly symmetric if the symmetry group of its state space acts
transitively on the set of lists (of a given length) of perfectly
distinguishable states.  These two properties are convex abstractions
of properties that hold for quantum systems: in the quantum case, the
state space is the set of density matrices on a Hilbert space,
perfectly distinguishable sets of states are the sets of rank-one
projectors (pure density matrices) $vv^\dagger$ corresponding to
orthonormal sets of Hilbert space vectors $v$ (modulo phase factors),
spectrality is the spectral decomposition of density matrices, and
strong symmetry reflects the facts that any orthonormal basis can be
taken to any other via a unitary operator and that conjugations by
unitaries, $\rho \mapsto U \rho U^\dagger$, are symmetries of the set
of density matrices.  Although much of quantum physics is best represented 
in infinite-dimensional Hilbert spaces, quantum information theory has focused
on degrees of freedom that can be represented in finite-dimensional Hilbert
spaces, or on controlling finite-dimensional subalgebras of observables in larger
systems, corresponding to effective finite-dimensional Hilbert spaces, e.g. ones  
describable as tensor products of two-dimensional complex Hilbert 
spaces (``qubits").  And most of the characteristically quantum phenomena such as interference, 
tradeoffs between information gain and disturbance in the measurement process, the existence of 
incompatible observables, and entanglement, that are 
associated with the cryptographic and computational advantages of quantum information
processing over classical, are already present in finite-dimensional systems.   So finite-dimensional results such
as those of the present paper can help us understand the conceptual and physical basis of such quantum phenomena.  Indeed most recent research on general probabilistic theories has been done in finite-dimensional
settings.

That simple Euclidean Jordan algebras and simplices are spectral and
strongly symmetric is known  (cf. \cite{BMU}), so the new result here is the
converse.  It is obtained by showing that compact convex sets
in finite dimension satisfying spectrality and strong symmetry are (up
to affine isomorphism) a subset of the regular convex bodies, defined
in \cite{FarranRobertson} as bodies whose symmetry group acts
transitively on maximal chains of faces, and using the classification
of those convex bodies in \cite{MaddenRobertson} to verify that this
subset consists precisely of the abovementioned classes of sets.

Several recent works have explored the consequences of the conjunction
of spectrality and strong symmetry, or sets of principles that imply
these, for the state space of a system in the GPT framework.  In
particular, in \cite{LeeSelbyGrover}, the important lower bound of $\Omega(\sqrt{N})$
queries (due to Bennett, Brassard, Bernstein, and Vazirani \cite{Bennett97b}) on 
computing OR (that is, $x_1 \vee x_2 \vee \cdots \vee x_N)$ via quantum queries to an $N$-bit-string $(x_1, x_2, \cdots ,x_N) \in \{0,1\}^N$, which shows that one cannot improve on the $\sqrt{N}$ 
queries used in Grover's quantum algorithm for ``searching" over 
$N$ possibilities for one having a desired property, was extended to systems 
satisfying a set of five postulates.
These postulates include strong symmetry, and imply (by results in
\cite{ChiribellaScandoloEntanglementAxiomatic}) spectrality.  In \cite{KrummEtAlThermo},
the conjunction of strong symmetry and spectrality was shown to imply
thermodynamical results generalizing important facts about quantum
thermodynamics, and similar results were obtained in
\cite{ChiribellaScandoloDiagonalization}.  Our work implies that these recent results, as well 
as additional results on query complexity obtained in \cite{BarnumLeeSelbyOracles}, 
apply to a much smaller set of GPT systems than might have been
hoped.  We discuss this further in the concluding section of the
paper.

The paper is organized as follows.  After this outline of the paper's structure, 
Section \ref{subsec: terminology} briefly 
describes some terminological and mathematical conventions used in the rest of the paper.
Section \ref{sec: background}
provides the necessary background on convex compact sets and
associated structures, especially their groups of automorphisms, and
defines the notion of measurement, which is used in formulating
the notion of perfect distinguishability and the properties of
spectrality and strong symmetry which depend on it.  It also touches
on the use of this framework to represent potential physical systems
abstractly, which motivates notions like measurement,
distinguishability, spectrality, and strong symmetry.  Section
\ref{sec: distinguishability spectrality strong symmetry} defines
perfect distinguishability, spectrality and strong symmetry, and
reviews important consequences (mostly from \cite{BMU}) of the latter
two properties that will be needed in proving the main theorem.
Section \ref{sec: Euclidean Jordan algebras} describes the main class
of convex compact sets that are the target of the characterization via
spectrality and strong symmetry in the main theorem: the normalized
state spaces of simple Euclidean Jordan algebras.  It also reviews the
already-known direction of the main theorem for these cases,
explaining how known results in Jordan theory imply that these state
spaces are spectral and strongly symmetric.  Finally, it reviews the
structure of these Jordan algebras considered as representations of
the automorphism groups of their normalized state spaces, and of the
automorphism groups of the cones over these state spaces.  This
structure will be used in the proof of the main theorem in Section
\ref{sec: main result}, to compare them with the polar representations
in which Madden and Robertson \cite{MaddenRobertson} embed all regular
convex bodies in order to classify them.
Section \ref{sec: simplices} briefly reviews the easily seen fact that
simplices, the other part of the class of bodies we characterize, are
spectral and strongly symmetric, which completes the proof of the
``if'' direction of the main theorem.  Section \ref{sec: sss bits are
  balls} establishes, via a direct and simple geometric argument
essentially from \cite{DakicBruknerQuantumBeyond} (but 
also using strong symmetry which was not explicitly assumed there), the ``only if''
direction for the special case of strongly symmetric spectral systems
in which no more than a pair of states can be perfectly distinguished
at once: they are affinely isomorphic to balls.  This enables us to
avoid some case-checking later.  Section \ref{sec: regular convex
  bodies} introduces the main tools we will use in the new part of our
characterization theorem: the notion of regular convex body and
associated theory, and Section \ref{sec: Madden-Robertson
  classification} describes the classification by Madden and Robertson
\cite{MaddenRobertson} of these bodies as convex hulls of particular
orbits in polar representations of compact groups, along with some of
the theory of polar representations used in the classification.  The
classification uses a one-to-one correspondence (which was described
in Section \ref{sec: regular convex bodies}) which associates each
regular convex body $\Omega$ embedded as the convex hull of an orbit
in a polar representation with a polytope $\pi(\Omega)$ given as the
convex hull of a particular orbit of the ``restricted Weyl group'' of
the representation, acting in a subspace $\fa$, with $\pi(\Omega) =
\fa \intersect \Omega$.  Finally, Section \ref{sec: main result}
completes the proof of the characterization theorem by proving the new
part: that all strongly symmetric spectral convex compact sets are
normalized state spaces of simple Jordan algebras, or simplices.  This is done by
first proving that they are regular, and then that the associated
polytope $\pi(\Omega)$ is a simplex.  We then go through the list of
representations in the Madden-Robertson classification of regular
convex bodies, and verify that all of those (other than the simplices themselves)
whose polytope is a
simplex with three or more vertices occur in polar representations of
automorphism groups of Euclidean Jordan algebras acting on those
algebras.  We compare the representation-theoretic description of the
orbits leading to normalized Jordan state spaces from Section
\ref{sec: Euclidean Jordan algebras}, to the description of the points
whose orbits yield the regular convex body in the Madden-Robertson
construction, and see that, up to an affine transformation, they are
the same.  Section \ref{sec: discussion} discusses the
implications of the result for GPTs, including recent work assuming
strong symmetry and spectrality; discusses several simple and
physically or informationally meaningful ways of adding additional
assumptions to further narrow the class of systems to that of standard
complex quantum theory, and also discusses the relation of the present
work to other characterizations of Jordan-algebraic state spaces, and Section 
\ref{sec: conclusion} briefly concludes.

\subsection{Some mathematical terminology}
\label{subsec: terminology}
A few terminological conventions and definitions are worth noting.  We
will sometimes abuse notation by referring to a singleton set,
$\{x\}$, by the name of its element $x$, especially when $\{x\}$ is
the face of a convex set consisting of the single extremal point $x$.
For $S$ any subset of an inner product space $V,(\,\cdot \, , \, \cdot \,)$ we
define $S^\perp := \{x \in V: \forall y \in S ~(x,y) = 0\}$.  We
also adopt the common practice whereby, when we substitute a set in a
place in an expression where an element of the set is expected, the
expression is taken to refer to the set of all its referents when
elements of the set are substituted in that place.  For example, 
for $S$ a
subset of an inner product space $V$, $\{x: (x,S) = 0\}$ has the same
meaning as $\{x: \forall y \in S ~(x,y) = 0 \}$, i.e. it refers to
$S^\perp$; for $G$ a group acting on a set $S$, $G.x$ for fixed $x \in
S$ is the orbit through $x$, and so forth.

We use the notation $\E^n$ to indicate $n$-dimensional
\emph{Euclidean space}, by which we mean a finite-dimensional real
vector space equipped with a distinguished positive definite inner
product (a concrete example is $\R^n$ equipped with the dot product).
When $S$ is a subset of a real affine space or linear space, 
we use the notation $\conv S$ for the convex hull of $S$, i.e. the 
set of all convex combinations of elements of $S$.   

For any group $G$ acting linearly on a vector space $V$, and any subspace $L$
of $V$, we denote the subgroup that preserves the subspace $L$ by
$G_L$, and the subgroup that fixes $L$ pointwise by $G^L$.  ($N_G(L)$
and $Z_G(L)$ are also common notation for these.)  We write
$G_0$ for the connected identity component of a Lie group $G$, and
$G_0^s$ for the semisimple part of $G_0$.  

We write $:=$ or $=:$ to indicate that the expression on the side with
the colon is defined by the expression on the side with the equals
sign, both when defining an expression for the first time and in
occasional reminders.  Finally, as is common in the
literature, we use ``positive'' to mean ``nonnegative''.

\section{Background and framework}
\label{sec: background}

We study convex compact subsets $\Omega$ of finite-dimensional real
affine spaces $A$.  From now on the term ``convex compact set'' refers
to such sets.  Unless otherwise noted, it refers to full-dimensional
sets, i.e. ones whose affine span, $\aff \Omega$, is $A$.  We will be
concerned with intrinsic properties of such a set, which are invariant
under affine transformations.  In this setting, unlike the setting of
compact convex subsets of Euclidean space (defined as a finite-dimensional
real vector space with positive definite inner product), there is no way of
distinguishing a notion of rectangle from that of square or
parallelogram, for example; a trapezoid, however, represents a
distinct affine isomorphism class from the other three.

Some intrinsic properties of convex compact sets are best formulated
by embedding $\aff{\Omega}$ as an affine hyperplane, not containing
the origin, in a real vector space $V$ of dimension one greater than
the dimension, which we'll call $n$, of $\aff{\Omega}$.  The
properties we study are independent of the particular embedding.  Such an 
embedding determines the convex, topologically closed, pointed,
generating cone\footnote{These concepts will be defined below.} $V_+ := \R_+ \Omega$, its dual cone $V^*_+ \subset
V^*$, and a unique $u$ in the interior of $V^*_+$, defined by the
property that $u(\Omega)=1$ (meaning $u(\omega) = 1$ for all $\omega
\in \Omega$).  
\begin{definition}\label{def: cone and order unit} 
Let a compact convex set $\Omega$ of dimension $n$ be embedded in an
affine hyperplane in $V \setminus \{0\}$ where $V$ is a vector space
of dimension $n+1$, as above.  We call $V_+ := \R_+\Omega \subset V$
\emph{the cone over} $\Omega$, \emph{the cone generated by} $\Omega$, or
$\cone{\Omega}$.\footnote{Note that $V_+$ by itself does not determine
  $\Omega$; different choices of $u \in \interior V^*_+$ give rise to
  different bases for $V_+$, which need not in general be affinely
  isomorphic.  For example, a cone with square base also has 
  trapezoidal bases, obtained for example by swinging the hyperplane
  containing the square base around one edge of the square (which we
  can think of as a ``hinge'').  $V_+$, on the other hand, \emph{is}
  determined (up to affine isomorphisms) by $\Omega$ in the above construction 
(independently of the particular embedding of $\aff{\Omega}$ into $V\backslash
  \{0\}$).}  We call the unique $u \in \interior V_+^*$ defined by the
property that $u(\Omega)=1$, the \emph{order unit} associated with
$\Omega$.\footnote{The term comes from the theory of order-unit and base-norm
Banach spaces; see for example \cite{ASBook1,ASBook2}.}
\end{definition}

We also define \emph{the cone generated by} $S$, sometimes writing it as
$\cone{S}$, for general subsets $S$ of a real vector space $V$, as
the set of nonnegative linear combinations of elements of $S$; in the
special case of $S=\Omega$ in the setting of Definition \ref{def: cone and order unit} this
agrees with the definition given there.  For our purposes, a cone 
$C$ is defined to be a subset of a real vector space, closed under addition
and nonnegative scalar multiplication.  It is called \emph{generating} if it 
spans the vector space, and \emph{pointed} if $C \intersect - C = \{0\}$.
The cones obtained from Definition \ref{def: cone and order unit} are 
pointed, generating, and topologically closed.   For reasons that will become
clearer later, sometimes $V^*$ is referred
to as the \emph{observable space}.


An important interpretation of this formalism, which was the
motivation and setting of \cite{BMU},
views $\Omega$ as the space of normalized states of an abstract
physical system.  This may or may not correspond to the state space of
any system in physical theories currently in use. 
Because of this interpretation, we sometimes use the terms ``state''
for elements of $\Omega$, and ``pure state'' for extremal elements of
$\Omega$.  For instance, for finite-dimensional quantum systems
corresponding to a Hilbert space of finite dimension $d$, $\Omega$, of
affine dimension $d^2-1$, is the set of density matrices
(i.e. unit-trace positive semidefinite $d \times d$ complex Hermitian
matrices).  The pure states are the rank-one density matrices (which
are the Hermitian projectors onto one-dimensional subspaces
of the underlying Hilbert space).  The cone $V_+$ is the positive
semidefinite matrices.

An $(n-1)$-\emph{simplex} in an affine space is the convex hull of $n$
affinely independent points; if $x_1,...,x_n$ are such points, we
write $\triangle(x_1,...,x_n)$ for this simplex.  The $(n-1)$ in 
$(n-1)$-simplex refers to its dimension, equal to the dimension of the
affine space it spans.  All
$k$-simplices are affinely isomorphic; the abstract
$k$-simplex, i.e. the affine equivalence class of such simplices,
or an arbitrary such simplex considered only up to affine equivalence,
is often referred to as $\tri_{k}$.  In the ``operational'' or GPT 
interpretation, the $(n-1)$-simplex is often thought of as a
finite-dimensional \emph{classical} state space.  The associated
vector space $V$ of one higher dimension is usually taken to be
$\R^n$, with the extremal points $x_i$ embedded as the unit coordinate
vectors $e_i = (0,0,...,0,1,0,...0)$ with $1$ in the $i$-th place, so
that $\tri_{n-1}$ is identified with the standard probability simplex.

The \emph{barycenter} (also known as \emph{centroid}) $c(\Omega)$ of a
full-dimensional compact convex set $\Omega \subset A$, where $\dim A
= \dim \Omega = n$, is defined by introducing coordinates so as to
identify $A$ with $\R^n$, and letting: \beq c(\Omega) := \int_{\Omega}
d\mu(x) x \eeq where $d\mu$ is Lebesgue measure on $\R^n$.  One shows
that $c(\Omega)$ is
independent of the particular coordinatization of $A$ as $\R^n$, and
(what amounts to the same thing) that $c(\Omega)$ is covariant under
affine transformations, i.e. for any affine transformation $T$,
it holds that $c(T(\Omega)) = T(c(\Omega))$.

The automorphism group, $\aut{\Omega}$, of $\Omega$ is the set of
affine transformations of $\aff{\Omega}$ that take $\Omega$ \emph{onto} itself.  It is a
compact group.  (This is often called $\Omega$'s \emph{symmetry group}
but (mainly for consistency with \cite{FarranRobertson, MaddenRobertson}) 
we'll reserve that term for use in a slightly
different context to be introduced later.)  We'll often denote it with
the letter $K$.  In the GPT literature, its elements are often called
\emph{reversible transformations}.  Any affine space can be viewed as
a vector space by choosing a point to be $0$; it is natural to do this
for $\aff{\Omega}$ by taking the barycenter of $\Omega$ as zero.
(This should not be confused with the extension of $\aff \Omega$ to
the vector space $V$ described above, which is of one greater
dimension.)
The transformations in $\aut{\Omega}$, because they are affine, fix
the barycenter, so with respect to this vector space structure on
$\aff{\Omega}$, they are linear.  They also extend (as do any affine
transformations on $\aff{\Omega}$) linearly to all of $V$.  We may
also refer to the extension as $\aut{\Omega}$, or as $K$ ($K$ and its
extension are of course isomorphic as groups).  
In fact we have:

\begin{proposition}
\label{prop: canonical embedding}
Let $\Omega$ be a compact convex subset of $A \iso \aff{\Omega}$.  $A$
may be equipped with the structure of a Euclidean space $E$ in such a
way that $\aut{\Omega} \subseteq O(E)$.  In doing so, the barycenter
of $\Omega$ becomes $0 \in E$, and $\aut{\Omega}$ is precisely the
subgroup of $O(E)$ that preserves $\Omega$.
\end{proposition}

This follows from the well-known fact that if a compact group $K$ acts
linearly on a real vector space $V$, there exists a $K$-invariant inner product
on $V$\footnote{One proves this by verifying that if $\langle . , . \rangle: V \times V \rightarrow \R
:: (x, y) \mapsto \langle x  , y \rangle$ is an arbitrary inner product on $V$, 
then $(x,y) \mapsto \int_K dk \langle kx , ky \rangle$, where $dk$ is normalized Haar measure
on $K$,  is a $K$-invariant inner product.} 
and the fact  that $\aut{\Omega}$ is 
a compact group that fixes $c(\Omega)$ (Proposition \ref{prop: barycenter from group averaging}).

\begin{definition}
\label{def: canonical embedding}
We call the identification of $\aff{\Omega}$ with a Euclidean vector
space $E$ such that $\aut{\Omega} \subseteq O(E)\equiv O(\aff{\Omega})$
a \emph{canonical embedding} of $\Omega$ into Euclidean space.
\end{definition}

Recalling the construction of the cone $V_+ \subset V$ over $\Omega$
in Definition \ref{def: cone and order unit}, we note that we can also
equip $V$ with an inner product so that $K$'s extension to $V$ is a
subgroup of $O(n+1)$.  This subgroup will fix the ray $\R_+
c(\Omega)$, i.e. the ray over the image of $0 := c(\Omega)$ under the
embedding of $\aff{\Omega}$ into $V$. 
With this choice of inner
product, and identifying $V^*$ with $V$ via the inner product, $u \in
V$ is the embedded image of $0 \in E \iso \aff{\Omega}$.  

The following fact is very useful:
\begin{proposition}\label{prop: barycenter from group averaging}
Let $K := \aut{\Omega}$ act transitively on the extreme boundary
$\partial_e \Omega$ of $\Omega$, and let $\omega \in \partial_e
\Omega$.  Then $\int_{K} d\mu(k) k.\omega = c(\Omega)$, where
$d\mu$ is Haar measure on $K$.  $c(\Omega)$ is the \emph{unique}
$K$-invariant point in $\Omega$.  
\end{proposition}

The property in the premise of this proposition, transitive action of
$\aut{\Omega}$ on $\partial_e \Omega$, is sometimes called
$\emph{reversible transitivity on pure states}$.  A convex set with
the property is also sometimes called an \emph{orbitope}, following
\cite{Orbitopes}. When the compact group and representation $K$ are
clear, we will sometimes refer to $\conv{K.\omega}$ as \emph{the
  orbitope over} $\omega$ or \emph{the orbitope generated by}
$\omega$.  If the connected identity component of $\aut\,{\Omega}$ acts
transitively on $\partial_e\Omega$, we say $\Omega$ has \emph{continuous
reversible transitivity}.

%

The next two definitions are essential for defining the central
properties we will investigate: spectrality and strong symmetry.

\begin{definition}\label{def: effect, measurement}
An \emph{effect} is an element of the interval $[0,u] \subseteq
V^*_+$, with respect to the (partial) ordering of $V^*$ induced by
$V^*_+$ (namely, $x \ge y := x - y \in V^*_+$).  A \emph{measurement} is a finite sequence $e_i, i \in \{1,...,m\}$ of
effects such that $\sum_{i=1}^m e_i = u$.
\end{definition}

  An equivalent characterization of effects is as precisely the
  elements of $V^*$ that satisfy, for all $\omega \in \Omega$,
  $e_i(\omega) \in [0,1] \subset \R$.  (While the first clause only
  requires them to be in $V^*$, the rest of the definition implies
  they are positive, i.e. in $V^*_+$.)  In the operational
  interpretation effects, since they give probabilities when
  evaluated on states, are used to mathematically represent the probabilities of 
  outcomes of measurements.
  The definition of measurement guarantees that for every state
  $\omega \in \Omega$, the probabilities $e_i(\omega)$ for the
  outcomes $1,\ldots,m$ of a measurement $e_1,...,e_m$ sum to $1$: $\sum_i
  e_i(\omega) = u(\omega) = 1$.

Since we are in finite dimension, a positive definite inner product
$(\,\cdot\, , \, \cdot \,): V \times V \rightarrow \R$ induces (as does any
nondegenerate bilinear form, positive definite or not) an isomorphism
between $V^*$ and $V$.  It is common to use such an inner product to
represent the dual space ``internally'' in the primal space; then the
dual cone becomes $V^{*internal}_+ := \{ y \in V: \forall x \in V_+, (y,x) \ge 0\}$.  We
say that a cone $V_+$ is \emph{self-dual} if there exists an inner
product with respect to which $V_+^{*internal} = V_+$.\footnote{This is a properly
stronger property than affine isomorphism of $V^*_+$ with $V_+$, which
is sometimes called ``weak self-duality''.  The cone with square
base, for example, separates the two properties.}   We will 
call such an inner product \emph{self-dualizing}. 

When, as in much of this paper, we deal with self-dual cones, we will normally fix a  
self-dualizing inner product, use it to identify $V$ with $V^*$ and drop the 
superscript ``\emph{internal}".

A \emph{face} of a convex set $C$ is a convex subset $F$ such that
whenever $x = \sum_i p_i x_i$, with $p_i > 0$, is a finite convex decomposition of $x
\in F$ in terms of $x_i \in C$, all the $x_i$ are in $F$.  An
\emph{exposed face} is the intersection of $C$ with a supporting
hyperplane; it is easily shown to be a face, but the two notions are
not equivalent.  Where necessary we modify these definitions so that
(1) $C$ itself is considered to be an exposed face as well as a face, and
(2) the empty set is considered to be an exposed face, and a face, when $C$ is compact, but
not when $C$ is a cone.
This gives a bijection $\phi$ between the faces of the cone $V_+ \equiv
\R_+ \Omega \subset V$ over $\Omega$, and the faces of a convex
compact set $\Omega$: for each face $G$ of $\R_+ \Omega$, $\phi(G) := G
\intersect \aff{\Omega}$ is a face of $\Omega$.  $\phi$'s inverse takes
each nonempty face $H$ of $\Omega$ to the face $\R_+ H$ of $G$, and
$\emptyset$ to $\{0\}$.  There is an analogous bijection in the case of exposed faces. 
 When ordered by inclusion, the set of faces
and the set of exposed faces are each lattices, in which we write the
least upper bound (``join'') of a pair $x,y$ of elements as $x \join
y$, and their greatest lower bound (``meet'') as $x \meet y$.  
In any lattice, meet and join are associative.  For any subset
$S$ of a convex set $C$, we define the face generated by $S$,
$\face{(S)}$ as the smallest face that contains $S$.  (Of course it
must be shown unique for this to be an admissible definition; it is.)
So for a pair of elements $\face{(\{x,y\})} = x \join y$, and by
associativity of join, for a finite set of elements,
$\face{(\{x_1,....,x_k\})} = x_1 \join x_2 \join \cdots \join x_k$.  A
similar notion applies to the lattice of exposed faces: $\eface{(S)}$,
the exposed face generated by $S$, is the smallest exposed face
containing $S$.  When several convex sets may coexist in the same
space, we may use the name of a set as a subscript to indicate which 
convex set we are generating a face in, e.g. $\face_C(S)$ is the face of $C$
generated by $S$.  

Both the set of faces and the set of exposed faces of a compact
convex set $\Omega$ or of a (pointed, closed, generating) cone $V_+$,
are  bounded lattices.  Their upper bounds are
$\Omega$ (resp. $V_+$), and lower bound $\emptyset$ (resp. $\{0\}$).
\footnote{The face lattice is an affine invariant of compact convex
  sets, but does not fully separate affine equivalence classes of
  compact convex sets, as the fact that a cone and its base have
  isomorphic face lattices, but the cone does not in general determine
  its base up to affine equivalence, illustrates.}

A cone $V_+$ is said to be \emph{perfect} if $V$ can be equipped with
an inner product such that every face $F$ of $V_+$, including $V_+$
itself, is equal to its dual with respect to the restriction of the
inner product to $\lin F$.\footnote{The cone with pentagonal base
  (indeed, the cone over any regular $n$-gon with odd number of
  vertices greater than $3$) illustrates that this is properly
  stronger than self-duality.}

In the literature on general probabilistic theories, systems
are sometimes modeled using an explicit specification of a proper
subset of the sets of effects and of measurements, or positive maps,
subject to reasonable conditions, called the ``allowed'' or
``physical'' effects, measurements, or positive maps.  
When the full set of effects is allowed, the system is said to satisfy
the ``no-restriction hypothesis''.  We are concerned with intrinsic
properties of the compact convex set $\Omega$ which cannot depend on
such a choice, so we have no need of this possibility.  We mention it
because the framework used in \cite{BMU} does allow for the
possibility of a restricted set of allowed effects, and the notions of
spectrality and strong symmetry in \cite{BMU} are defined as we define
them here except that the notion of perfect distinguishability used is
with respect to the set of allowed effects.   In principle this
could affect the notions of spectrality, and of strong symmetry, for a
fixed state space.  However, the conjunction of spectrality and strong
symmetry, even with respect to a restricted set of effects, is shown
in \cite{BMU} to imply no-restriction.  All of the results we use from
\cite{BMU} concern consequences of the conjunction of these two
properties, so they hold equally for strongly symmetric spectral sets
in our more specialized setting.

\section{Distinguishability, spectrality, and strong symmetry}
\label{sec: distinguishability spectrality strong symmetry}

In this section we continue to use $\Omega$ to refer to an arbitrary
finite-dimensional compact convex set.

\begin{definition}\label{def: perfectly distinguishable}
A set of points $\omega_i, i \in \{1,...,k\}$ in $\Omega$ is called
\emph{perfectly distinguishable} if there is a measurement $\{e_i\}, i
\in \{1,...,k\}$ such that $e_i(\omega_j) = \delta_{ij}$.  
\end{definition}


  In the operational interpretation, perfect distinguishability means
  that if we are guaranteed that a physical system has been prepared
  in one of the states $\omega_1,..,\omega_k$, but we do not know
  which one, we may ascertain which one was prepared with perfect
  certainty by doing the measurement $\{e_i\}$.  If we get the result
  $j$, then the state that was prepared must have been $\omega_j$.
  Perfect distinguishability, and approximations to it, are central to
  considerations of the information storage and processing properties
  of abstract systems in this interpretation.

An equivalent characterization of distinguishability will also be
useful to us.  We call an indexed set of effects $E=\{e_i\}$ a
\emph{submeasurement} if $\sum_i e_i \le u$.  For each submeasurement
$E = \{e_1,...,e_r\}$, the indexed set $E' := \{e_1,...,e_r,
e_{r+1}\}$, where $e_{r+1}:= u - \sum_{i=1}^r e_i$, is a measurement.
We could have equally well defined perfect distinguishability of
$\omega_1,...,\omega_r$ as the existence of a submeasurement such that
$e_i(\omega_j) = \delta_{ij}$, for it is easy to see that
$e_{r+1}(\omega_i) = 0$ for all $i \in \{1,..,r\}$, whence there are
many ways of constructing an $r$-outcome measurement
$\{e'_{1},...,e'_{r}\}$ such that $e_i(\omega_j)= \delta_{ij}$---for
instance, let $e'_1 = e_1 + e_{r+1}$, and $e'_i = e_i$ for all $i \in
\{2,...,r\}$.

\begin{definition}\cite{BMU}
\label{def: frame}
A sequence $\omega_1,...,\omega_k$ of perfectly distinguishable 
\emph{pure} states is called a \emph{frame}, or a $k$-frame if we
wish to specify its cardinality.
\end{definition}

In GPT models in which restricted sets of allowed effects are
specified, as done in \cite{BMU}, distinguishability and frames are
usually defined with respect to the \emph{allowed} effects, so fewer
sets of states may be distinguishable, and there may be fewer frames,
than in the no-restriction case.  Although we referenced \cite{BMU}
for the above definitions, we have defined distinguishability, and
hence all concepts dependent on it, with respect to the set of
\emph{all} effects.  So for us, these notions depend only 
the convex geometry of $\Omega$.  

\begin{definition} \label{def: maximal frame}
A frame in $\Omega$ is called \emph{maximal} if it is not a
subsequence of any other frame. 
\end{definition}
 Although the cardinality of a
frame can be no greater than $n+1$ (recall that $n:= \dim (\aff{\Omega}$)),
for a given $\Omega$ the maximal cardinality of a frame may be much
less than $n$.\footnote{An example of the latter phenomenon is the mixed-state
space (the density matrices) of a $d$-dimensional quantum system,
where as mentioned before, $\dim(\aff{\Omega}) = d^2-1$.  The frames are
just lists of mutually orthogonal rank-one projectors, and the
cardinality (length) of a largest such list is $d$, which is far from
achieving the general upper bound of $n+1$ which is here $d^2$.
Modulo reorderings a simplex has a unique maximal frame, the set of
vertices of the simplex; the frames of a simplex are just the ordered
subsets of the vertices, i.e. finite sequences, without repetition, of
vertices.  In this case, the general upper bound is achieved.}


We are now ready to introduce the two properties of convex sets that
we will use in our characterization theorem.

\begin{definition}[\cite{BMU}]
A convex compact set $\Omega$ is called \emph{spectral} if, for each
point $\omega \in \Omega$, there is some frame whose convex hull
contains $\omega$.\footnote{There are other notions of spectrality for convex compact sets in the
literature, for instance Alfsen and Shultz's (\cite{ASBook2},
Definition 8.74).  Although related, Alfsen and Shultz's notion should not be
confused with the one used here.  However the conjunction of our weak
notion of spectrality with strong symmetry also implies \cite{BMU}
that $\Omega$ is spectral in Alfsen and Shultz's sense.  Riedel
\cite{RiedelSpectral} also introduced a notion of spectral ordered linear
space, but we will not use it here.}
\end{definition}

This is a very strong property of convex compact sets.
Jordan algebraic state spaces satisfy a stronger property, called
\emph{unique spectrality}, which requires that for any two convex 
decompositions $\omega = \sum_{i=1}^k p_i \omega_i = \sum_{i=1}^r q_i
\tau_i$ of $\omega$ into the elements $\omega_i, \tau_i$ of two
frames, $k=r$ and there exists a permutation $\sigma$ of $\{1,...,k\}$
such that $q_{\sigma(i)} = p_i$.  Unique spectrality was not assumed in
\cite{BMU} and will not be assumed in our characterization theorem
either, but it follows \cite{BMU} from the conjunction of
spectrality and the next property, strong symmetry.

\begin{definition}
A convex compact set $\Omega$ of dimension $n$ is called
\emph{strongly symmetric} if, for each $k \in \{1,...,n+1\}$,
$\aut{\Omega}$ acts transitively on the set of $k$-frames.
\end{definition}

Note that the set of $k$-frames may well be empty for many values of $k$
(in which case, trivially, $\aut{\Omega}$ acts transitively on it).  

``Strong symmetry'' was the term introduced in \cite{BMU}.
``Frame-symmetric'' or ``frame-transitive'' might have been a better
choice.  The notion of frame that is involved in the definition of
strong symmetry is not itself obviously symmetry-related.  And for
sets having few frames, it is actually not a very strong symmetry
requirement.  

Since a $k$-frame is a finite \emph{sequence} of perfectly
distinguishable pure states, strong symmetry implies that one can
arbitrarily permute the elements of any set of perfectly
distinguishable pure states, via symmetries.  It is therefore at least
\emph{prima facie} a stronger property than requiring that every
\emph{set} of perfectly distinguishable pure states (which we might
call an \emph{unordered frame}) can be mapped onto any other such set
of the same size by a symmetry.\footnote{The pentagon (like any regular $n$-gon with
  $n>3$) provides an example of a non-spectral but strongly symmetric
  convex compact set.  Spectral but not strongly symmetric convex
  compact sets also abound: for example, any ball is easily smoothly
  deformed to another smooth, strictly convex set with trivial
  automorphism group; all smooth, strictly convex sets are spectral.}

We collect some more consequences of the conjunction of spectrality
and strong symmetry in the following proposition.  

\begin{proposition}[Mostly from \cite{BMU}] \label{prop: consequences}
For a convex compact set $\Omega$ that is spectral and strongly symmetric, the
following hold:
\begin{enumerate}
\item \label{item: frames and faces} Every face of $\Omega$ is
  generated by a frame.  Any two frames that generate the same face
  $F$ have the same cardinality, which we call the \emph{rank}, $|F|$,
  of the face.  If the face $G$ is a proper subset of $F$, then $|G| <
  |F|$.
\item \label{item: faces exposed} 
Every face of $\Omega$ is exposed.
\item \label{item: perfection} The cone $V_+ := \R_+\Omega$ over
  $\Omega$ is a perfect self-dual cone.  The self-dualizing
  inner product $(.,.)$ can be chosen to be $\aut{\Omega}$-invariant, and
such that $(\omega, \omega) = 1$ for all
  pure states (extremal points) $\omega$ of $\Omega$.
\item \label{item: frames orthogonal} With respect to the
  self-dualizing inner product on $V$, the elements of any frame are
  an orthonormal set.  The states of a frame, viewed as elements of
  the dual space via this inner product, are effects, and are
  therefore a distinguishing submeasurement for that frame.  If
  $\omega_1,...,\omega_n$ is a maximal frame, i.e. a frame for
  $\Omega$, then $\sum_{i=1}^n \omega_i$ is the order unit.  
\item \label{item: barycenters of faces} If $F = \omega_1 \join \cdots
  \join \omega_k$ for some frame $\omega_1,...,\omega_k$ (equivalently, $F$
  is the face generated by that frame) then the barycenter of $F$ is
  $\sum_{i=1}^k \omega_i /k$.  (It follows that 
$\sum_{i=1}^k \omega_i$ does not depend on which frame 
$\omega_1,..., \omega_k$ for $F$ is summed over.)  
\item \label{item: orthocomplement} If $F$ is a face of $V_+$, 
(resp. $\Omega$) then $F' := F^\perp \intersect V_+$ (resp. $F^\perp
  \intersect \Omega$)  is a face of $V_+$
  (resp. $\Omega$) such that $F \meet F' = \{0\}$ (resp. $F
  \meet F' = \emptyset$) and $F \join F' = V_+$ (resp. $F \join
  F' = \Omega$). (Here all $\perp$'s and linear spans are taken in
  $V$, with respect to the self-dualizing, invariant inner product.)
In a lattice, these two conditions define what it means for an element $F'$ to be a 
\emph{complement} of $F$, so we call $F'$ the face complementary to $F$, or simply 
$F$'s complement.
\item \label{item: frame extensions and maximal frames}
The face generated by a maximal frame is $\Omega$ itself.  Every
  frame $A$ of $\Omega$, generating a face $F$, extends to a maximal
  frame $M$, by appending a frame $B$ for $F'$.  Similarly if $F < G$
(i.e. $F \subsetneq G$), every frame for $F$ extends to a frame
for $G$.  
\item \label{item: orthomodularity} The map $F \mapsto F'$ on the face
  lattice of $V_+$ (equivalently of $\Omega$) is an orthocomplementation,
  with respect to which the lattice is orthomodular, i.e. for $F \le
  G$, $G = F \join (F' \meet G)$.  The additional states appended to a
  frame on $F$ in order to extend it to a frame on $G$ (cf. item \ref{item:
    frame extensions and maximal frames} above), are a frame for $F' \meet
  G$.
\end{enumerate}
\end{proposition} 

We may think of orthomodularity as stating that $F' \meet G$ behaves
as a ``relative orthocomplement'' of $F$ in $G$, i.e. an orthocomplement
in the sublattice consisting of elements below or equal to $G$.  

\begin{proof}
Item \ref{item: frames and faces} is Proposition 2 of \cite{BMU}.  
Item \ref{item: faces exposed} follows easily from item 
\ref{item: frames and faces}.

Item \ref{item: perfection} is established in \cite{BMU} in the course
of proving Theorem 8 of that paper, which states that for every face
$F$, the orthogonal projection (with respect to the self-dualizing $\aut{\Omega}$-invariant
inner product) $P_F$ onto the linear span of a face $F$, is a positive
map.  Iochum \cite{IochumBook} showed that for self-dual cones,
positivity of all such projections with respect to a self-dualizing
inner product is equivalent to perfection.  The
self-duality of the cones $\R_+ \Omega \subset V$ over strongly
symmetric spectral convex sets, with respect to an inner product with the
stated properties, was established as Proposition 3 of
\cite{BMU}.\footnote{In fact self-duality does not require spectrality, nor
does it require the full strength of strong symmetry: in
\cite{MuellerUdudecSelfDuality} it was shown that transitivity of
$\aut{\Omega}$ on $2$-frames (ordered pairs of perfectly distinguishable
states) implies self-duality. 
}

Item \ref{item: frames orthogonal} is
Proposition 6 of \cite{BMU}.

Item \ref{item: orthocomplement} is
partly stated in Proposition 7 of \cite{BMU}, and the rest can be
extracted from the proof of that Proposition. 

Item \ref{item: frame extensions and maximal frames} is part of
Proposition 7 of \cite{BMU}.    

Item \ref{item: orthomodularity} begins with the claim that the
map $'$ is an orthocomplementation.  A complementation is an
involutive map of a bounded lattice such that $F \join F' = 1$ and $F
\meet F' = 0$.  It is called an orthocomplementation if it is
order-reversing: $F \le G \Leftrightarrow G' \le F'$.  The
involutiveness of $'$ is also part of Proposition 7 of \cite{BMU}.
The other two conditions on a complementation are part of item
\ref{item: orthocomplement} above.  The last part of
orthocomplementation, order-reversingness, is shown as part of the
proof of Theorem 9 in \cite{BMU}; it follows directly from the
extendibility of frames to maximal frames.

The second part of item \ref{item: orthomodularity},
i.e. orthomodularity, is Theorem 9 of \cite{BMU}.  The crucial element
in its proof (given that we have already established that $'$ is an
orthocomplement) is the ``relative frame extension property'',
i.e. the last sentence in item \ref{item: frame extensions and
  maximal frames}.  The last sentence of item \ref{item:
  orthomodularity} is a step in this proof from \cite{BMU}.

The only item we have not covered yet is item \ref{item: barycenters
  of faces}. It does not appear to be explicitly stated in \cite{BMU},
although it is likely known.  We include a proof in Appendix
\ref{sec: faces of sss sets are sss}, as part of a proof that the
faces of strongly symmetric spectral sets are strongly
symmetric and spectral.
\end{proof}
 
\section{A class of examples: strongly symmetric spectral convex sets from simple Euclidean Jordan algebras}
\label{sec: Euclidean Jordan algebras}

Finite-dimensional Euclidean Jordan algebras were introduced by
Pascual Jordan around 1932 \cite{Jordan}, as a possible algebraic
setting for the formalism of quantum theory.  The notion abstracts
properties of the complex Hermitian matrices, which are the
observables\footnote{As noted in Section \ref{sec: background}, in the setting of general probabilistic theories,
  a reasonable generalization of the space of observables is the
  vector space $V^*$.} of a finite-dimensional quantum system; in
particular, properties of the \emph{symmetrized} product $A \bullet B
:= (AB + BA)/2$ which, unlike the ordinary associative matrix product,
preserves Hermiticity.  Our main result asserts that we can identify,
up to affine isomorphisms, the spaces of normalized states of a
certain class of these algebras with the strongly symmetric spectral
convex compact sets.  So in this section, we will introduce these
algebras, define their normalized state spaces, and give their
classification so that we can identify them when they appear in the
classification of regular convex bodies (as defined in \ref{def: regular}) .  We also give the proof,
mainly consisting of references to known results, that they are
strongly symmetric and spectral, since this is one direction (the
already-known one) of our main theorem.  Finally, we explain that the
action of the automorphism group of the normalized state space on the
subspace where it acts nontrivially, is via a polar representation,
indeed a symmetric space representation, and identify a point in the
representation, the convex hull of whose orbit is the normalized state
space. This will be used in Section \ref{sec: main result} to connect
these state spaces with the Madden-Robertson classification of regular
convex bodies, since the latter proceeds by showing that their
symmetry groups' actions can be induced by polar representations.

Our main reference for facts about Euclidean Jordan algebras and the
associated cones will be Chapters I-V of the book of Faraut and
Koranyi \cite{FarautKoranyi}, which deals with symmetric cones in
finite dimension.  We also refer to Chapter I of
\cite{Satake}.\footnote{\cite{FarautKoranyi} is the most elementary
  and the most focused on our concerns, but because it is very
  detailed, relevant material is occasionally somewhat scattered and
  one sometimes needs to chase chains of definitions and theorems.
  Satake's Chapter I \cite{Satake} is succinct and extremely well
  organized, but uses notions somewhat more general than we need
  (e.g. Jordan triple systems), and also uses some
  machinery from algebraic geometry and algebraic groups.  Another
  good reference, but less focused than the other two on the concerns
  of this paper, and somewhat more involved because it also covers
  infinite-dimensional cases, is the first part of \cite{ASBook2}.
  The spectral theorem is Theorem 2.20 on p. 46.  Neither this nor
Satake's chapter is as explicit as \cite{FarautKoranyi} about
  frames and frame-transitivity.}

\label{subsec: Euclidean Jordan algebras}

A Jordan algebra is a real algebra (i.e., a real vector space $V$
equipped with a bilinear product $\bullet: V \times V \rightarrow V$)
that is commutative and that, while not in general associative, satisfies a
special case of associativity, the \emph{Jordan property}: $a^2\bullet
(a \bullet b) = a \bullet (a^2 \bullet b)$, where we use the notation
$a^2 := a \bullet a$.  It need not have a unit, but one can always be
adjoined if it is absent.  We will consider only unital
finite-dimensional Jordan algebras.  A Jordan algebra is called
\emph{formally real} if $a^2 + b^2 = 0$ implies that $a=b=0$.    In
finite dimension, formal reality coincides with another property,
Euclideanity.
A Jordan algebra $(V, \bullet)$ is said to be 
Euclidean if it is possible to introduce an inner product $( \, \cdot
\, ,\, \cdot \,): V \times V \rightarrow \R$ that is ``associative'',
i.e. $(a \bullet b, c) = (a, b \bullet c)$.  

The set of squares in a Jordan algebra is obviously closed under
nonnegative scalar multiplication.  In a formally real (equivalently
Euclidean), Jordan algebra, it is also closed under addition, so it 
is a convex cone (cf. 
e.g. \cite{FarautKoranyi}, Ch III $\S$2), which we call $V_+$.   It is immediate from 
formal reality that $V_+$ is pointed.  Since it is in
addition topologically closed, it has a compact convex base, giving an
example of the formalism of compact convex sets embedded as bases
of regular cones, described in Section \ref{sec: background}, and
permitting an ``operational'' interpretation as the state space of a
physical system.  $V_+$  is self-dual with respect to 
the associative inner product.

Soon after Jordan introduced them, Jordan, von Neumann and Wigner
\cite{JNW} classified the finite-dimensional formally real Jordan
algebras. They are precisely the $n \times n$ self-adjoint matrices
with entries in $\R, \C$, or $\H$ and the $3 \times 3$ octonionic
self-adjoint matrices, equipped in each case with symmetrized matrix 
multiplication $x \bullet y = (xy + yx)/2$ as Jordan product, and the \emph{spin factors} $\R^n \oplus \R$ for
every $n \ge 1$, equipped with the product 
\beq
\label{eq: spin factor
  Jordan product} (\x,s)\bullet (\y,t) = (t\x + s\y , \langle \x, \y
\rangle + st).  \eeq Here $\x, \y \in \R^n$, $s, t \in \R$.  Self-adjoint
(``Hermitian'' is also used) means $M = M^\dagger$, where $M^\dagger
:= \overline{M}^t$, and $\overline{M}$'s entries are the conjugates of
$M$'s with respect to the canonical conjugation on $\R, \C$, $\H$, or 
$\Oct$.
The conjugation is the identity in the case of $\R$, thus the
self-adjoint real matrices are just the real symmetric matrices.

Although we do not make direct use of them in obtaining our results, 
important facts about the positive cones of Euclidean Jordan algebras are (1) the Koecher-Vinberg theorem \cite{Koecher,
  Vinberg} \footnote{Proofs may also be found in \cite{Satake}, Ch. I
  Theorem 8.5, \cite{FarautKoranyi}, Theorem III.3.1, and
  \cite{BarnumWilceLocalTomography}.} that the cones of squares in
finite-dimensional formally real Jordan algebras are precisely the
homogeneous self-dual cones (homogeneity being defined as transitive
action of the automorphism group of the cone on the interior), and (2)
the result that they are precisely the \emph{symmetric cones},
i.e. the regular cones whose interiors are Riemannian symmetric
spaces (\cite{Koecher:57}; see also \cite{Vinberg}).

\subsection{Normalized Jordan algebra state spaces: spectrality and strong symmetry}
 
Now we define the normalized state spaces of Euclidean Jordan algebras
(which we will henceforth sometimes abbreviate as ``EJAs''), and
explain how to show the known fact that they are spectral and strongly
symmetric.  

Recall that an idempotent in an algebra is an element $c$ such that 
$c^2 = c$, and it is called \emph{primitive} if it is not a nontrivial
sum of other idempotents.  

\begin{definition}\label{def: Jordan frame}
A \emph{Jordan frame} is defined to be a complete set of orthogonal
primitive idempotents $c_i$ in a Euclidean Jordan algebra, where
orthogonality means that $c_i \bullet c_j = \delta_{ij} c_i$ and
completeness means $\sum_i c_i = e$, the unit of the algebra.
\end{definition} 

\begin{theorem}[Spectral theorem for finite-dimensional Euclidean Jordan algebras]
\label{theorem: Jordan spectral theorem}
Every element $x$ of a Euclidean Jordan algebra has a decomposition
\beq \label{eq: Jordan spectral decomposition 1} x = \sum_{i=1}^r
\lambda_i c_i \eeq where $\lambda_i \in \R$ and $c_i$ are a Jordan
frame.  When we rewrite this as \beq \label{eq: Jordan spectral
  decomposition 2} x = \sum_\alpha \lambda_\alpha c_\alpha, \eeq where
$c_\alpha := (\sum_{i \in \alpha} c_i)$ and the sets $\alpha \subseteq
\{1,...,r\}$ are a partition of the indices into the largest subsets
within which $\lambda_i =: \lambda_\alpha$ is constant, then the
decomposition (\ref{eq: Jordan spectral decomposition 2}), into
not-necessarily-primitive idempotents, is unique.
\end{theorem}
This is a combination of Theorems III.1.1
and III.1.3 in \cite{FarautKoranyi}.  

The values $\lambda_\alpha$ are the \emph{spectrum} of $x$.
We can define the \emph{trace} of $x \in V$ as $\sum_i \lambda_i$, 
and the determinant as $\Pi_i \lambda_i$, 
where $\lambda_i$ are the coefficients in the spectral decomposition
(\ref{eq: Jordan spectral decomposition 1}).  Using the trace we define
a bilinear form $(x,y) := \tr (x \bullet y)$ on the Jordan algebra.
Proposition III.1.5 of \cite{FarautKoranyi} states that the positive
definiteness of this form (making it an inner product) is equivalent
to Euclideanity.  Indeed, in a simple Euclidean Jordan algebra every
associative inner product is a positive scalar multiple of this one.  

\begin{definition}
In a Euclidean Jordan algebra $V$, 
the squares satisfying $\tr x = 1$ form a compact convex
base $\Omega$ for the cone $V_+$ of squares.  We call this 
the \emph{normalized state space} of $V$.
\end{definition}

\begin{proposition}[e.g. \cite{FarautKoranyi}, Corollary IV.3.2]
\label{prop: primitive idempotents extremal}
The primitive idempotents are precisely the extremal points of
$\Omega$.
\end{proposition}

In a Euclidean Jordan algebra $V$ the spectrum of a square is
nonnegative.  One usually represents the dual space internally using
the trace inner product.  Since $\tr x = \tr (e \bullet x) \equiv (e,
x)$, we then have that $e$ is the order unit for this choice of base.

\begin{proposition}\label{prop: Jordan state spaces are spectral}
The normalized state spaces of finite-dimensional 
Euclidean Jordan algebras are spectral.
\end{proposition}
\begin{proof}
The primitive idempotents of an EJA $V$ are precisely the extremal
points of its normalized state space $\Omega$.  Since the cone is
self-dual, and all idempotents are below or equal to the order unit
$e$,\footnote{An easy argument from the spectral theorem gives that
  every primitive idempotent is part of a Jordan frame, and it is part
  of the definition of Jordan frame that it sums to the order unit.}
they are also effects.  Orthogonality of primitive idempotents in the
sense $e_i \bullet e_j = 0$ of Definition \ref{def: Jordan frame}
implies orthogonality with respect to the inner product, $\tr e_i
\bullet e_j = 0$, so any subset of a Jordan frame, considered as a set
of effects, is a submeasurement that perfectly distinguishes the same
subset, considered as states.  Consequently, an ordered subset of a
Jordan frame (and in particular, an ordered Jordan frame itself) is a
frame in the sense of Definition \ref{def: frame}.  So the spectral
theorem implies spectrality.
\end{proof}

In order to complete the proof of strong symmetry, we also show: 
\begin{proposition}\label{prop: all frames are Jordan}
All the frames in an EJA state space are ordered subsets of Jordan frames.  
\end{proposition} 
This is known, and implicitly assumed
in \cite{BMU}, but we give a proof.  
 \begin{proof}
Since $\partial_e \Omega$ is the set of primitive idempotents, and $V_+$ is self-dual
with respect to the inner product $\langle a , b \rangle = \tr a \bullet b$, a frame is a sequence 
$c_i, ~ i \in \{1,\ldots,s \}$ of 
primitive idempotents such that there exists a submeasurement $e_i$ for which 
$\langle e_i , c_j \rangle = \delta_{ij}$.  In a general setting, not only in Jordan
state spaces, it follows immediately from the condition 
$\langle e_i , \omega_j \rangle = \delta_{ij}$ on a frame that if $e_i = \sum_k p_k f_k$ is a convex decomposition of 
$e_i$ into effects $f_k$, each of the $f_k$ also has the property $f_k (\omega_j) = \delta_{kj}$.  So 
the condition that $\omega_i$ is a frame may be restated as the existence of a submeasurement
consisting of \emph{extremal} (in the convex body $[0,u]$) effects.  Proposition 1.40 of \cite{ASBook2} states
that the extreme points of the positive part $[0,u]$ of the unit ball of a JB-algebra are the idempotents.  The finite-dimensional JB-algebras are the EJAs.      
It is also known (cf. \cite{ASBook2}, Proposition 2.18) that for idempotents $p_i$, 
$\sum_{i=1}^k p_i \le u$ implies that $p_i  \perp p_j$ for all $i,j \in \{1,\ldots,k\}$ with $i \ne j$.
So in the condition for $c_i$ to be a frame, we may take the $e_i$ to be a mutually orthogonal
set of idempotents.   We show that $\langle e_i , c_i \rangle = 1$ for an idempotent $e_i$ and
a primitive idempotent $c_i$ implies that $c_i \le e_i$.
To do so we use the fact, from \cite{ASBook2}, that JB-algebras $V$
are equipped with normalized self-adjoint idempotent positive linear maps   $P_{p}: V \rightarrow V$ 
called \emph{compressions}, in bijection with the idempotents $p$, such  that $p = P_p e$ and the exposed faces (all faces in finite dimension) are the positive parts of the images of compressions.   We have $1 = \langle e_i , c_i \rangle = \langle P_{e_i} e , c_i \rangle
= \langle e, P_{e_i} c_i \rangle$.  Since it follows from Proposition 1.41 
(in finite dimensions, where the dual space may be identified with the primal space) of \cite{ASBook2} that compressions on 
an EJA are \emph{neutral}, i.e. $||P \omega || = ||\omega|| \implies  P \omega = \omega$, we 
have that $P_{e_i} c_i = c_i$, hence $c_i \in \im_+ P_{e_i}$, where the latter is defined as 
$\im \, P_{e_i} \intersect V_+$.  By Lemma 1.39 of \cite{ASBook2}, 
$\im_+ P_{e_i} \intersect [0, e] = [0, e_i]$, and since $c_i \in [0, e]$, we have $c_i \le e_i$.

With $c_i \le e_i$, and $e_i \perp e_j$ for all $i \ne j$, it follows that $c_i \perp c_j$ for all $i \ne j$, and 
consequently that the $c_i$ are a subsequence of an ordered Jordan frame.
\end{proof}

\begin{proposition}\label{prop: Jordan state spaces are strongly symmetric}
The normalized state space of a Euclidean Jordan algebra is 
strongly symmetric.
\end{proposition}
\begin{proof}
Corollary IV.2.7 of Theorem IV.2.5 in \cite{FarautKoranyi} states that
the compact group $K$, defined as the subgroup of $\aut_0(V_+)$ that
fixes the Jordan unit $e$, acts transitively on the set of Jordan
frames.  Since we've adopted a canonical inner product, this is a
subgroup of $\aut{\Omega}$.  It is clear from the proof of Theorem
IV.2.5 in \cite{FarautKoranyi} that this transitive action is on
\emph{ordered} Jordan frames.  Since we showed, in the proof of
Proposition \ref{prop: Jordan state spaces are spectral} and in
Proposition \ref{prop: all frames are Jordan}, that the frames (in the
sense of Definition \ref{def: frame}) are precisely the ordered
subsets of Jordan frames, the group $K$, and hence $\aut{\Omega}$,
acts transitively on the set of $k$-frames for each $k$.  
\end{proof}

In particular, we have the following: 
\begin{corollary}
The normalized state space of a Euclidean Jordan algebra is an
orbitope, i.e. $\Omega = \conv{K.\omega_0}$, for any $\omega_0 \in
\partial_e \Omega$.  In particular, for $Herm(n,\D), 
\D \in \{\R,\C,\H,\Oct\}$, $\Omega = \conv{K.e_{11}}$, where $e_{11}$ is the
matrix unit $diag (1,0,0,...,0)$.
\end{corollary}
This is so because extremal states are $1$-frames, and the last
sentence follows from Proposition \ref{prop: primitive idempotents
  extremal} and the fact that $e_{11}$ is a primitive idempotent.

We note here that the state space of a spin factor, $V=\R^n \oplus \R$,
is the unit $n$-ball $\{(x,t): t = 1, |x|^2 \le 1 \}$ in the affine subspace
$t=1$.

\subsection{Classification of Euclidean Jordan algebras, their cones,
state spaces, and automorphism groups}

In this section we give a more detailed description of
the Euclidean Jordan algebras, their cones of squares, their
normalized state spaces, and the automorphism groups of these objects.
We also give a representation-theoretic description of normalized
Jordan algebra state spaces that will allow us to identify them with
the convex hulls of certain orbits in polar representations.  

The automorphism group of the cone of squares of a Jordan algebra $V$,
like the automorphism group of any self-dual cone, is reductive
\cite{Mostow}.  If the algebra is simple, then $\aut{V_+} =
G^s \times \R_+$, where $G^s$ is simple.  Also, $\aut_0(V_+) = G^{s}_0 \times
\R_+$.\footnote{In the non-simple case, $\aut_0{V_+} \iso G_0^{s}
  \times \R_+^k$, where $k$ is the number of simple factors of $V$,
  each copy of $\R_+$ acts as dilations on simple factors, and
  $G_0^{s}$ is a product of the groups $\aut_0^{s}{V_+^i}$ acting on
  simple factors $V_i$ (in other words, $\aut_0{V_+} = \Pi_i
  \aut_0{V_+^{i}}$); the full (not necessarily connected) automorphism
  group allows, besides non-identity components in each factor,
  permutations of isomorphic factors.}  We will also write $\aut^s(V_+)$ and $\aut_0^s(V_+)$ for
the groups $G^s$ and $G_0^s$.  $V$ is an irreducible
representation space for $G^s$ (cf. e.g.  \cite{Satake}, p. 42), and for
its connected identity component $G^s_0$.  
Like any semisimple Lie group, $G^s$ has
Cartan decompositions $\fg = \fk \oplus \fp$ of its Lie algebra and,
correspondingly, $G^s = K \exp \fp$ of the group.  We may choose one such
that $K$ is the subgroup of $G^s_0$ that fixes the identity $e$, since this
is a maximal compact subgroup of the linear group $G^s$.   $V$ is an irreducible \emph{spherical}
representation of $G^s$, and of its connected identity component $G^s_0$,  
which means that $G^s$'s (and also $G^s_0$'s) maximal compact subgroup, 
$K$, has a one-dimensional fixed-point space (in this case, $\R e$).  

In Table \ref{table: EJAs} (essentially from \cite{FarautKoranyi}) we
list the finite-dimensional simple Euclidean Jordan algebras and associated
data.  $V$ is the algebra, with positive cone $V_+$, $\fg$ is the Lie
algebra of $\aut{V_+}$, and $\fk$ the Lie algebra of $\aut{\Omega}$
(where $\Omega$ is the normalized state space).  We use the notation
$Herm(m, \D)$ for the Jordan algebra of self-adjoint matrices over a
classical division algebra $\D \in \{\R, \C, \H, \Oct\}$, $PSD(m,X)$
for the positive semidefinite matrices, and $Lorentz(1,n-1)$ for the
Lorentz cone $\{(z,x) \in \R \oplus \R^{n-1}: |z|^2 \ge |x|^2 , z \ge 0\}$.  Also
we write $Sym(m, \D)$ for the space of symmetric matrices with entries
in $\D$, and note in particular that $Herm(m, \R) \equiv Sym(m,\R)$. 
\begin{table}
\caption{Euclidean Jordan algebras with associated cones and Lie algebras}
\label{table: EJAs}
\[
\begin{array}{c}
\begin{array}{lccccc} 
\hline
  V  & V_+ & \fg & \fk & \dim V & \rank ~V \\
 \hline 
Sym(m,\R) & PSD(m,\R) & \fsl(m, \R) \oplus \R & \fo(m)  & m(m+1)/2 & m  \\
Herm(m,\C) & PSD(m,\C) & \fsl(m, \C) \oplus \R & \fsu(m) & m^2 & m \\
Herm(m,\H) & PSD(m,\H) & \fsl(m,\H) \oplus \R & \fsu(m,\H) & m(2m-1) & m \\
\R \oplus \R^{n-1}& Lorentz(1,n-1) & \fo(1,n-1) & \fo(n) & n & 2 \\
Herm(3,\Oct) & PSD(3,\Oct) & \fe_{6(-26)} & \ff_4 & 27 & 3 \\
\hline 
\end{array} \\
\\ 
\end{array}\]
\end{table}

In the three infinite families of matrix cases, $\aut_0{V_+}$ is the
group of transformations $X \mapsto AXA^\dagger$, where $A$ is any
nonsingular matrix over the relevant division algebra $\D$, i.e. 
$A \in GL(m, \D)$.\footnote{$\aut_0{V_+} = \aut{V_+}$ except in the
  following cases where $\aut_0{V_+}$ is of index two, with a
  representative of the nonidentity component as stated: $Herm(n,\R)$
  for $n$ even, $X \mapsto MXM^\dagger$ with $M =
  \diag{(-1,1,...,1)}$; $Herm(n,\C)$, transpose; $\R \oplus \R^{n-1}$
  for $n \ge 3$, $x \mapsto Mx$ with $M = \diag{(1,-1,1,1,..,1)}$
  \cite{Satake}.  The groups generated by $X \mapsto MXM^\dagger$ with
  $M \in SL(n,\C), SL(n,\H)$ are not precisely $SL(n,\C), SL(n,\H)$ in
  general, although they have the same Lie algebras as those groups;
  they are homomorphic images.}
The identity component, $\aut_0(\Omega)$, of
the maximal compact subgroup consists of those  transformations for
which $A$ is drawn from the subgroups $SO(m), SU(m, \C), SU(m, \H)$
respectively.  This
preserves the identity matrix in $V \iso Herm(m, \D)$, which is the
unit $e$ of the Jordan algebra.  The space orthogonal to this consists
of the traceless self-adjoint matrices over $\R, \C, \H$ respectively,
and it is an irreducible representation space for $K$.

The set of such transformations for
which $A$ is drawn from the subgroups $SO(m)$, $SU(m, \C)$, $SU(m, \H)$
respectively is the identity component of a maximal compact subgroup of $\aut_0(V_+)$.
 It preserves the identity matrix in $V \iso Herm(m, \D)$, which is the
unit $e$ of the Jordan algebra.  The subspace of $V$ orthogonal to the identity matrix consists
of the traceless self-adjoint matrices over $\R, \C, \H$ respectively,
and it is an irreducible representation space for $K$.

A similar description is not obviously possible in the case of $Herm(3,
\Oct)$, because $\Oct$ is nonassociative. 
However, one can
deal with this by observing that the transformations $X \mapsto
AXA^\dagger$ for $A \in SL(m,\D), \D \in \{ \R, \C, \H \}$ are
determinant-preserving, and symmetric octonionic matrices have
well-defined determinant (indeed, there is a Jordan-algebraic
definition of determinant, cf. \cite{FarautKoranyi}).  
 Then one can show
(cf. \cite{DrayManogue}) that for arbitrary $m$ and $\D \in \{\R, \C,
\H\}$, and for $m= 3$ and $\D = \Oct$, the 
determinant-preserving linear transformations of $V$
are $\aut_0^{s}{(V_+)}$.\footnote{Even though in the case of $\D =
  \Oct$ we can identify a subset of $3 \times 3$ octonionic matrices
  $M$ such that $X \mapsto MXM^\dagger$ is determinant-preserving on
  octonionic-hermitian matrices $X$, and this generates the group of
  determinant-preserving linear transformations, this group is not
  \emph{identical} to the maps $X \mapsto MXM^\dagger$ because of the
  nonassociativity of octonionic matrix multiplication (see
  \cite{DrayManogue}, where the group is dubbed $SL(3,\Oct)$).}  The
identity-preserving subgroup, in the cases $\D \in \{\R,\C,\H\}$, is
as stated in the previous paragraph.  In the octonionic case it is the
compact real form of $F_4$.  (In \cite{DrayManogue} this is dubbed $SU(3, \Oct)$.)

Returning to consideration of general EJAs, the Lie algebra of the 
subgroup that fixes the Jordan unit $e$ is the compact part $\fk$ in a
Cartan decomposition $\fg = \fk \oplus \fp$ of the reductive
group $\aut{V_+}$, and $V$ itself can be identified with $\fp$ in this
decomposition.\footnote{For example, this is part of Theorem 8.5 of
  \cite{Satake}, with the Cartan decomposition given in Lemma 8.6 of that book and the
  proof.  It also follows from the conjunction of Theorem III.2.1 and
  Theorem III.3.1 in \cite{FarautKoranyi} (this conjunction is,
  roughly, Satake's Theorem 8.5).} As with any Cartan decomposition,
there is a natural $K$-invariant
inner product on $V$, such that $\fk$ is represented by real antisymmetric
matrices and $\fp$ by symmetric ones.  Furthermore for a simple
Jordan algebra $\fp$ decomposes as $\R \oplus \fp_0$, where $\fp_0 =
\fp \intersect \fg^{s}$, $\fg^{s}$ is the semisimple
part of the Lie algebra of $\aut{V_+} = \R_+ \times G^s$ (so $\fg^s = \flie{G^s}$) 
and $\R$ is generated by the
Jordan unit $e$.  Although the actions of $\aut{V_+}$, and of $G^s$, on
$V$ are not their adjoint actions (on, i.e. corestricted to, $\fp$),
the restriction of these actions to $K$ does coincide (on $\fp$) with the
restriction of the adjoint action.\footnote{In the matrix cases
with $\D \in \{\R,\C,\H\}$, this is manifested in the fact that while
in general $A^\dagger \ne A^{-1}$, equality does hold when $A$ is
respectively orthogonal, unitary, or quaternionic-unitary; this
exemplifies the general fact that for compact groups, representations
are isomorphic to their duals, which establishes the claim for the
spin-factor and octonionic cases too.}  In other words, $\fp_0$ is the
representation space of a polar representation (Definition \ref{def: polar representation}), 
called a symmetric space representation (Definition \ref{def: symmetric space representation}), 
of the compact group $K$.  The Jordan trace gives the
component in the $\R$ factor of the Jordan algebra $V \iso \fp = \R
\oplus \fp_0$.  Recall that the normalized state space, $\Omega := V_+
\intersect \{x \in V: \tr x = 1\}$, is also $\conv{K.\omega}$ for any
extremal $\omega \in \Omega$ (from strong symmetry).  Everything in
the affine plane $\{x \in V: \tr x = 1\}$ has the form $c \oplus
\omega_0$, where $c = e/\rank \, {V}$ is the unit-trace element of the
fixed-point space $\R e$, and $\omega_0 \in \fp_0$.  Thus $\Omega$ is
affinely isomorphic to $\conv{K.\omega_0}$, an orbitope in the polar
representation of $\fk$ on $\fp_0$.  This is what will permit us, in
Section \ref{sec: main result}, to identify these Jordan algebra state
spaces with certain regular convex bodies, since regular bodies are
described in the Madden-Robertson classification \cite{MaddenRobertson} as convex hulls of
orbits in polar representations.

\section{Simplices are strongly symmetric and spectral}
\label{sec: simplices}
In the previous section, we saw that the state spaces of simple
Euclidean Jordan algebras are strongly symmetric and spectral, part of
the ``if'' direction of our main theorem.  In this section, we
establish the other part of the ``if'' direction: the easy fact that
simplices are strongly symmetric and spectral.

The standard presentation of an $n$-simplex as embedded in a vector
space $V$ of one higher dimension takes the $n+1$ vertices of the
simplex as the unit vectors $e_i$ of $V = \R^{n+1}$ considered as a
Euclidean space in the usual way, i.e. with ``dot product'' as inner
product, and represents effects in the same space, with evaluation
given by this inner product.  $V_+$ is then the nonnegative orthant,
$\R^n_+$, and it is manifestly self-dual with respect to this inner
product.  So the order unit $u$ is the all-ones vector $(1,1,...,1)$,
the dual cone $V_+^*$ is equal to the primal cone, and the set of
effects is the unit hypercube (the convex hull of the $2^{n+1}$ vectors
of length $n+1$ with entries in $\{0,1\}$).  The unit vectors are
therefore not only the pure states, but are also effects, and they sum
to $u$, so they constitute a measurement.  Every ordered subset of the
unit vectors is a submeasurement, and perfectly distinguishes the same
subset considered as states, so is a frame; since by definition frames are 
ordered subsets of the pure states, and in a simplex, the pure states are the unit 
vectors, these are all the frames.
The full set of unit vectors is the unique maximal frame, up to order.
Since the simplex is defined as its convex hull, the simplex is spectral.
The group $\aut{\Omega}$ is the symmetric group 
$S_n$ acting to permute the vertices, so it is obviously transitive 
on $k$-frames for each $k \in \{1,...,n+1\}$, i.e. the simplex is 
strongly symmetric.   

It is relatively easy to show that the simplices are the only strongly
symmetric spectral polytopes.  For example, from item \ref{item: perfection} of Proposition \ref{prop:
  consequences}  we know that the cone over a
strongly symmetric spectral body must be perfect.  Theorem 1 of
\cite{BarkerForan} states that the only self-dual polyhedral cones in
$\E^n$ whose maximal faces are all self-dual in their spans (with
respect to the restriction of the inner product) are isometric (hence
also affinely isomorphic) to the simplicial cone $\R^n_+$.

\section{Strongly symmetric spectral bits are balls}
\label{sec: sss bits are balls}

In this section, we prove a special case of our main theorem, for
convex sets whose largest frame has cardinality $2$.  We show that
they are all affinely isomorphic to balls, which are the normalized
state spaces of the Euclidean Jordan algebras known as spin factors.
This case can be handled without the Farran-Madden-Robertson theory
\cite{FarranRobertson, MaddenRobertson} of regular convex bodies,
and we can use it in the proof of our main theorem, where it allows us
to avoid some tedious case-checking involving the Madden-Robertson
classification.

We call a convex body $\Omega$ a \emph{bit} if the largest frame it
contains has cardinality $2$.  We say $\Omega$ has \emph{reversible
  transitivity} if $\aut{\Omega}$ acts transitively on the set
$\partial_e \Omega$ of pure states of $\Omega$.  All strongly
symmetric convex bodies have this property, since extremal states are
$1$-frames.

In \cite{DakicBruknerQuantumBeyond} (Section IV) B. Daki{\'c} and C. Brukner claimed 
that spectral bits that have reversible transitivity on pure states (i.e. on 
$1$-frames) are necessarily balls, and give an argument for this claim.  The claim may well be
true, but their argument makes an implicit assumption.  When this 
assumption is made explicit, one immediately sees that it follows
from transitivity on $2$-frames, which in the case of bits is 
identical to strong symmetry.  Once this is observed, one sees that
Dakic and Brukner's argument establishes the following result, which is weaker than Daki{\'c} and
Brukner's claim.

\begin{theorem}
\label{theorem: sss bits are balls}
Let $\Omega$ be a strongly symmetric spectral compact convex body whose
largest frame is of cardinality $2$.  Then $\Omega$ is affinely
isomorphic to a ball.
\end{theorem}

The proof, which is essentially Daki{\'c} and Brukner's argument done in 
slightly more detail and with an explicit assumption of  
transitivity on $2$-frames, is in Appendix \ref{sec: proof that sss bits are balls}.
 
\section{Regular convex bodies and regular convex compact sets}
\label{sec: regular convex bodies}

In this section we will use some definitions and results from
\cite{FarranRobertson} and \cite{MaddenRobertson} concerning compact
convex sets embedded in Euclidean space, which the authors of 
\cite{FarranRobertson} and \cite{MaddenRobertson} call \emph{convex
  solids} or, when they are full-dimensional, \emph{convex bodies}.
Some of the notions studied in \cite{FarranRobertson} and
\cite{MaddenRobertson} are sensitive to the particular embedding of a
compact convex set in Euclidean space, and fail to be invariant under
affine transformations, whereas our concern is with affine-invariant
structure.
Nevertheless their results are immediately applicable to our situation, 
because of the canonical embedding (Definition \ref{def: canonical embedding} of compact convex
sets of dimension $n$ into $\E^n$.  In the following, the group of 
\emph{rigid transformations} of $\E^n$ is defined as the group generated by
rotations, translations, and reflections. 

\begin{definition}[\cite{FarranRobertson}]
A \emph{solid} is a compact convex subset of Euclidean space $\E^n$.
It has affine dimension $m \le n$, and if we wish to indicate its
dimension, we may call it an $m$-solid in $\E^n$ or simply $m$-solid.
A \emph{convex body} is an $n$-solid in $\E^n$, i.e. a
full-dimensional solid; it may also be called an $n$-body.  The
\emph{symmetry group} $GB$ (sometimes $G$, for short) of a convex body
$B$ is the subgroup of the group of rigid transformations of $\E^n$
consisting of those transformations that take $B$ into (so in fact,
onto) itself.
\end{definition}

The symmetry group $GB$ of $B$ consists of affine automorphisms, but
in general it may be a proper subgroup of the group $\aut{B}$ of
affine automorphisms of $B$.  For example, a nonsquare rectangle has a
smaller symmetry group than a square, and a parallelogram whose sides
are not all the same length has a smaller symmetry group than a
rectangle, whereas $\aut{\Omega}$ is the same in all three cases.
Similarly, a general solid ellipsoid has a smaller symmetry group than
a (spherical) ball, but the two are affinely isomorphic.

The notion of symmetry group, and the associated results of
\cite{FarranRobertson} and \cite{MaddenRobertson} are nevertheless
useful to us because of the following (which is immediate from 
Proposition \ref{prop: canonical embedding}):

\begin{proposition}\label{prop: canonically embedded cc sets}
Let $\Omega$ be canonically embedded in a Euclidean space $E$, in 
the sense of Definition \ref{def: canonical embedding}.  Then
$\aut{\Omega}$ is the symmetry group $G\Omega$.
\end{proposition}

In \cite{FarranRobertson} and \cite{MaddenRobertson} the term ``face''
is used to mean ``\emph{exposed} face'', that is, the intersection of
the set with a supporting hyperplane.  Since we will use some
definitions and results from \cite{FarranRobertson} and
\cite{MaddenRobertson}, and since all faces in
strongly symmetric spectral convex bodies are exposed,\footnote{In
fact this is true of all regular convex bodies as well, by  
Proposition \ref{prop: regular bodies and polar representations}  and the fact \cite{BiliottiGhigiHeinznerPolarOrbitopes, BiliottiGhigiHeinznerInvariantPolar} that all faces of orbitopes in polar representations 
are exposed.   We suspect it would remain 
true if the definition of regularity were changed to refer not to exposed faces, but just to
faces.} that is how we will use the
term from now on.  It is a fairly standard convention to include
$\emptyset$ and $\Omega$ as faces, and even as exposed faces, of
$\Omega$; they are called \emph{improper} faces and the others are
called \emph{proper} faces.

\begin{definition}\label{def: flag}
A \emph{flag} of a convex compact set $\Omega$ is a
sequence $F_1,...,F_r$ of distinct nonempty exposed faces of $\Omega$
such that $F_0 \subset F_1 \subset \cdots F_r$.
\end{definition}

In other words it is a chain in the face lattice, not containing the
empty face.\footnote{This is almost identical to the definition in
\cite{FarranRobertson}, except that they also exclude the face
$\Omega$, whereas we prefer to allow (but not require) its inclusion.
When we need their notion we will use the term ``short flag''.}

Every convex body may be considered as a compact convex set relative
to the natural affine space structure of $\E^n$ (obtained by
forgetting about the inner product and $0$).  So the notion of flag
also applies to convex bodies.

\begin{definition}
A \emph{maximal flag} is a flag 
that is  not a
subsequence of any other flag. 
\footnote{We define a maximal flag exactly as in
  \cite{MaddenRobertson}, except that our flags contain the full set, 
whereas theirs do not.   Maximal flags of these two types are obviously 
in bijection with each other, by deleting or appending $\Omega$ at the end of
  the sequence of faces.
In \cite{FarranRobertson} a slightly different definition of
maximal flag is
  given, but the definitions give rise to the same notion of
  regularity, and are equivalent for regular convex bodies.
  Explicitly, $\sigma_\Omega \subset \N$, which we may  call the
  ``signature'' of $\Omega$, is the ordered (by restriction from the
  usual ordering of $\N$) set of integers $i$ such that there is a
  proper face of $\Omega$ of dimension $i$.  Farran and Robertson call
  a flag maximal if the sequence $\dim{F_0},...,\dim{F_r}$ of
  dimensions of faces in the flag is equal to $\sigma_\Omega$.  
  Once \emph{regular} convex bodies are defined, it will be clear the two
  definitions coincide for such bodies (though not in general, cf. the
examples in Fig. 2 of \cite{FarranRobertson}).}
\end{definition}

We have the following elementary fact.
\begin{proposition}\label{prop: action on flags}
Let $g$ be an automorphism of $\Omega$.  If $F$ is a face of $\Omega$, $g.F$ is also a face of 
$\Omega$, and $g.c(F) = c(g.F)$.  Let $\Phi$ be a flag of $\Omega$.  Then $g.\Phi$ is 
a flag of $\Omega$; it is maximal if and only if $\Phi$ is. 
\end{proposition}

\begin{definition}\cite{FarranRobertson}\label{def: regular}
A convex body $B$ is called \emph{regular} if its symmetry group
$G{B}$ acts transitively on the set of maximal flags of $B$.
\footnote{Once again, this is essentially the definition from Farran
  and Robertson \cite{FarranRobertson}, except that their notion of
  flag omits $B$ itself.  It obviously gives the same notion of
  regularity.}
\end{definition}

We give an analogous definition for arbitrary convex compact sets:
\begin{definition}\label{def: regular compact convex set}
A convex compact set $\Omega$ is called \emph{regular} if its affine
automorphism group $\aut{\Omega}$ acts transitively on the set of
maximal flags of $\Omega$.
\end{definition}

Unlike the notion of regular compact convex set, the notion of regular
convex body is not affine-invariant, because it is sensitive to the
embedding in Euclidean space.  But we have the following elementary
relation between the notions of regular convex body (Definition
\ref{def: regular}) and of regular compact convex set (Definition
\ref{def: regular compact convex set}).

\begin{proposition}\label{prop: regular bodies and sets}
Every regular convex body $B$, considered as a compact convex set, is 
regular.  Not every convex body that is regular when considered
as a compact convex set, is a regular convex body, but every 
\emph{canonically embedded} regular compact convex set is a regular convex
body.
\end{proposition}  

\begin{proof}
The first sentence holds because because the symmetry group $GB$ is a
subgroup of $\aut{B}$, and transitive action by a subgroup trivially
implies transitive action by the group.  The second holds because
a regular compact convex set may be embedded in such a way that its symmetry 
group is too small a subgroup of the automorphism group to act transitively
on frames (consider a triangular simplex, embedded as a non-equilateral
triangle), but in its canonical embedding, the symmetry group is equal 
to the automorphism group (cf. Proposition \ref{prop: canonically embedded cc sets}).  
\end{proof}

The following is Farran and Robertson's version (probably originating, for the case
of groups generated by a finite number of reflections, with
Coxeter or earlier), adapted to this
setting, of a standard notion from the theory of group actions on topological spaces, that of a
\emph{fundamental set} for an action.  
\begin{definition}\label{def: fundamental region}
Let $B$ be a convex body of affine dimension $n$ in 
$V \iso \E^n$, with barycenter $0$, and let $G$ be a compact subgroup
of $O(\E^n)$ that 
preserves $B$.  
A convex solid (not necessarily full-dimensional) $D \subset \E^n$, 
 is called a \emph{fundamental
  region} for the action of $G$ on $B$ if
\begin{enumerate}
\item $B \subseteq GD$, and
\item  Every $G$-orbit in $B$ meets the relative interior of 
$B$ in at most one point.
\end{enumerate}
\end{definition}

\begin{theorem}[Farran and Robertson (Theorem 7 of \cite{FarranRobertson})]
\label{theorem: fundamental region}
Let $B$ be a convex body embedded in $\E^n$.  Suppose in addition
that $B$ is regular.  Let $\Phi = (F_1,....,F_r)$ be a maximal flag of
$B$, and let $c_i$ be the barycenter of $F_i$.  Then the
$(r-1)$-simplex $\triangle_\Phi := \triangle(c_1,..., c_{r-1}, c_r)$ is a
fundamental region for the action of $B$'s symmetry group $G$ on
$B$.
\end{theorem}

When we work with a fixed flag $\Phi$, we often write $\tri$ for $\tri_\Phi$, and
 omit the subscripts indicating the dependence
of other objects on $\Phi$ as well.  
The automorphism group, $\aut{B} =: K$, of any convex compact set is a
compact Lie group.  If it has trivial Lie algebra ($\fk = \{0\}$),
then it is a finite group.  The same two statements hold for the
symmetry group of a convex body $B$ as defined by Farran and Robertson
(since their definition of convex body requires $B$ to affinely span
the Euclidean space in which it is embedded). 
We write $\tri$ for $\tri_\Phi$ and 
$\triangle'$ for the $r-2$-simplex $\triangle(c_1, ..., c_{r-1})$ in
$\aff{B} \equiv \R^n$.  This differs from the definition of
$\triangle$ only by omitting the barycenter $c_r\equiv 0$ of $B$; its affine
dimension is $r-2$, whereas $\triangle$'s is $r-1$.  

The facial structure of the simplex $\tri$ encodes important information about how generic the orbits  are, via
the following theorem.

\begin{theorem}[Theorem 8 of \cite{FarranRobertson}]
\label{theorem: isotropy subgroups and fundamental domain}
Let $\tri$ be the fundamental domain for the regular convex body $B$, as
defined in Theorem \ref{theorem: fundamental region}, and let $G = GB$, $B$'s symmetry 
group.  Let $F$ be a face of $\tri$ of nonzero dimension.  If $x$ is in the relative interior of $F$, 
and $y \in F$ is arbitrary, then $G_x \le G_y$.  In particular, if $x$ and $y$ are both in the
relative interior, $G_x = G_y$.
\end{theorem}

We now assume without loss of generality that $B$ is canonically
embedded in $\E^n$, so in particular $c_r = 0$.  
The points $c_1,....,c_{r-1}$ linearly 
generate---equivalently, $\triangle'$ linearly generates---an $(r-1)$-dimensional
vector space in $\aff{B} \equiv \R^n$, which we call $L$.  
(Of course, $\tri$ also (both linearly and affinely) generates $L$.)

Following \cite{FarranRobertson} we define a map $\pi: B \mapsto \pi(B) := B \intersect L$.    
In \cite{FarranRobertson}, $\pi$ is called a projection, presumably because it is idempotent: if
$B$ is a regular polytope, then $\pi(B) = B$.

\begin{definition}
\label{def: farran-robertson section}
For a regular convex body $B$ embedded in Euclidean space $\E^n$, we call the
subspace $L$ defined in the preceding paragraph a \emph{Farran-Robertson section},
and the polytope $\pi(B) := L \intersect B$  its \emph{Farran-Robertson polytope}.  
We define the \emph{Farran-Robertson polytope of a regular compact convex set}
$\Omega$ as the Farran-Robertson polytope of $\Omega$ when $\Omega$ is 
considered as a convex body by canonically embedding it in Euclidean space.
\end{definition}

As embedded in $\E^n$, the convex solid $\pi(B)$ depends on the choice of a maximal flag of $B$,
but as we shall see in Propositions \ref{prop: regular bodies and polar representations} and
\ref{prop: regular compact convex sets and polar representations}, all $\pi(B)$ are $G$-conjugate, 
hence isometric.
  
\begin{theorem}[Theorem 9 of \cite{FarranRobertson}]
\label{theorem: Farran-Robertson polytope is regular}
Let $B$ be a regular convex body with symmetry group $K$.  The
Farran-Robertson polytope $\pi(B)$ is a
regular polytope, of affine dimension $k = \dim L$, whose symmetry group $W$ is a
finite group generated by reflections, isomorphic to $K_L/K^L$.
\end{theorem}

As an example, consider a 3-dimensional ball $\Omega$ centered on
the origin.  $\aut{\Omega} \equiv G \Omega$ is $O(3)$.  $L$ may be taken to be any
line through the origin and $\pi(\Omega)$ the simplex $\tri_1$ consisting of
the diameter $L \intersect \Omega$.  $\aut{\pi(\Omega)} \iso \Z_2$ is generated
by the reflection through the plane orthogonal to this diameter.

Examples where $\Omega \ne \pi(\Omega)$ but with a more interesting $\pi(\Omega)$, for
example where $\pi(\Omega)$ is the three-vertex simplex $\tri_2$, are harder to
visualize because the affine span of $\Omega$ will tend to be too
large.  In the lowest-dimensional example (as it will turn out)
$\Omega$ is the unit-trace positive semidefinite real symmetric $3 \times 3$ 
matrices, with $L$ the diagonal unit-trace matrices (two dimensions)
and $L \intersect \Omega \iso \tri_2$.  $\aff{\Omega}$, the unit-trace
real symmetric matrices, is 5-dimensional in that case.  $\Omega$ in
this second example is affinely isomorphic to the example on
pp. 378-379 of \cite{FarranRobertson}, where it is presented as the
convex hull of the Veronese surface, an embedding of $\PP_2(\R)$
into $\E^6$.  This surface is the extreme boundary of its convex hull
$\Omega$, which spans a 5-dimensional subspace.

The following theorem allows us to understand the facial structure of
$\Omega$ by understanding that of $\pi(\Omega)$.  We will need it
later to show that frames in $\pi(\Omega)$ correspond to frames in
$\Omega$ in the strongly symmetric spectral case; it is of course also
of intrinsic interest.  This theorem is stated near the top of p. 369
of \cite{MaddenRobertson}.\footnote{An essentially identical (except for its
more restricted premise) theorem is
  stated for a more restricted setting (the subclass of polar
  representations occuring within adjoint representations of real
  semisimple groups) in \cite{FarranRobertson}.}

\begin{proposition}[\cite{MaddenRobertson}; see also \cite{FarranRobertson}, Theorem 10, and its proof]
\label{prop: flags and faces of B and P}
Let $B$ be a regular convex body with symmetry group $G$, and
Farran-Robertson polytope $\pi(B)$.  Let $F$ be a face of $\pi(B)$ with centroid
$c(F)$, and write $G^{c(F)}$ for the isotropy subgroup of $G$ at
$c(F)$.  Then the orbit $G^{c(F)} F$ is a face of $B$, which we call
$H_F$, and each face of $B$ is of the form $gH_F$ for some face $F$ of
$\pi(B)$ and some $g \in G$.  Moreover, if $F_1,F_2,...,F_r$ is a maximal
flag of $P$, then $H_{F_1}, H_{F_2},....,H_{F_r}$ is a maximal flag of
$B$, and every maximal flag of $B$ arises from a flag of $\pi(B)$ in this way.  
\end{proposition}

Since the symmetry group of a convex body is its automorphism group,
we also have the analogous statement for convex compact sets and their
automorphism groups.

\section{Classification of regular convex bodies via polar representations}
\label{sec: Madden-Robertson classification}

In this section, we sketch the classification of regular convex bodies
from \cite{MaddenRobertson}, via polar representations as defined,
 classified (in the irreducible case), and related to noncompact symmetric spaces in
\cite{Dadok1985}.  The results of this section will be used through
the classification set out in \cite{MaddenRobertson}, Tables 2, 3, and
4.  Besides the tables themselves we need to understand which orbit,
in the polar representation listed in the table, gives rise to the
regular convex body.

Let us write $\rho$ for the representation of the action of the
symmetry group $G\Omega$ on $\E^n$.  Proposition 2.1 of
\cite{MaddenRobertson} states that for a regular convex body $\Omega$,
this representation is irreducible.  \footnote{In \cite{MaddenRobertson} the
  proof is attributed to M. Pinto in a paper ``to appear'' which we
  have not been able to find.  It involves showing that regular convex
  bodies belong to the wider class of \emph{perfect} (not to be
  confused with the notion of perfection of cones) convex bodies
  \cite{FarranRobertson}, and the observation that this implies $\rho$
  is irreducible.  It is not hard to reconstruct a simple proof along
  these lines.  We note that Markus M{\"u}ller
  \cite{MarkusMuellerIrreduciblePersonalCommunication} has shown more
  directly that the automorphism group of a strongly symmetric
  spectral convex body $\Omega$ acts irreducibly on $\aff{\Omega}$.}
Proposition 2.2 of \cite{MaddenRobertson} states that for regular
convex bodies $\Omega$, $\rho$ is a \emph{polar} representation,
defined as one for which the subspace normal to a principal
(i.e. highest-dimensional) orbit of the action meets every orbit
orthogonally.  

More formally, let $G$ be a compact Lie group with Lie algebra $\fg$ and let 
$\rho: G \rightarrow O(V)$ be a  representation on a real
inner product space.  (As usual we often write $g$ for $\rho(g)$ where
$g \in G$.)  $v \in V$ is called \emph{regular} if the orbit $G.v$ is
of maximal dimension (i.e. no orbit has higher dimension).  The
tangent space at $v$ to an orbit $G.v$ is $\fg.v$.  It is a linear subspace
of $V$.
We define  a \emph{linear cross-section} of the
set of $G$-orbits in $V$ to be a subspace of $V$ that intersects every $G$-orbit.
Although this term is not explicitly defined in \cite{Dadok1985}, this is the sense in which it is used there.
We will usually omit the qualifier \emph{linear}, since we make no use of any other kind of cross-section.  

\begin{definition}[following \cite{Dadok1985}]
\label{def: a_nu}
For any $v \in V$, define $\fa_v$ to be the subspace of $V$ normal 
to the $G$-orbit at $v$, i.e. $\fa_v := (\fg.v)^\perp$.
\end{definition}

\begin{proposition}[\cite{Dadok1985}, Lemma 1]
 For any $v$, $\fa_v$ is a cross-section of the set of $G$-orbits.
\end{proposition}

Note that by Definition \ref{def: a_nu}, $V$ itself is a cross-section.\footnote{Sometimes a definition is used in which a cross section 
must intersect each orbit finitely may times, which would rule out $V$ except in the case of $\fg = 0$.}  More interesting are the minimal-dimensional cross sections.

\begin{definition}[\cite{Dadok1985}]
\label{def: Cartan subspace}
Cross-sections of the form $\fa_v$ for
regular $v$ are called \emph{Cartan subspaces}.\footnote{It seems likely that
in \cite{Dadok1985} this definition is meant to apply
  only in the polar case, but in the proof of Proposition 2.2 of
  \cite{MaddenRobertson} (Proposition \ref{prop: regular bodies and polar representations} below), it is clearly meant to apply prior to
  establishing the representation is polar.
}
\end{definition}

Such cross-sections are of minimal dimension.


\begin{proposition}[\cite{Dadok1985}]
\label{prop: polar}
Let $v_0$ be regular. The following are equivalent:
\begin{enumerate}
\item For any regular $v \in V$, there is a $k \in G$
such that $\fg.v = k.(\fg.v_0)$.
\item For any regular $v \in V$, there is a $k \in G$ 
such that $\fa_v = k.\fa_{v_0}$.
\item For any $u \in \fa_{v_0}$, $(\fg.u, \fa_{v_0}) = 0$.
\end{enumerate}
Thus we have uniqueness (up to the action of $G$) of minimal cross-sections if and only if 
the orbits intersect one such cross-section orthogonally.
\end{proposition}

\begin{definition}[\cite{Dadok1985}]\label{def: polar representation}
A representation satisfying any (and hence all) of the three equivalent
conditions in Proposition \ref{prop: polar} is called \emph{polar}.  
\end{definition}

Thus the polar representations are ones for which the minimal
dimensional cross-sections of the form $\fa_v$ for regular $v$,
i.e. normals to maximal-dimensional orbits, are all $G$-conjugate
(item 1 in Proposition \ref{prop: polar}), or equivalently ones for
which, for every maximal-dimensional orbit, the subspace normal to
that orbit (at any point) intersects every orbit orthogonally (item 3
in Proposition \ref{prop: polar}).

\begin{remark}
A cross-section that meets all orbits orthogonally is minimal.  Hence 
an equivalent definition (cf. \cite{GichevPolar, GozziPolar, EschenburgHeintzePolar}) of polar representation is that it is a representation for which there exists a cross-section that meets all orbits orthogonally.
\end{remark}

\begin{proposition}[\cite{MaddenRobertson}, Proposition 2.2 and its proof]
\label{prop: regular bodies and polar representations}
Let $B$ be a regular $n$-solid in
$E \iso \E^n$, with centroid $0$, and let $\pi$ be the inclusion of 
the symmetry group $G$ of $B$ in $O(n) \iso O(E)$.   Then the representation $\pi$ 
is polar, and for any maximal flag $\Phi$ of
$B$, the centroids of the faces in $\Phi$ are a basis for a Cartan subspace, $L_\Phi$.
\end{proposition}

\begin{proof}[Proof (greatly expanded version of proof in \cite{MaddenRobertson})]
Let $x$ and $v$ each be on principal orbits (not assumed to be the same orbit) of this $G$-action, 
and recall that 
$\fa_x := \{u \in \E^n: \langle u, \fg.x \rangle = 0\}$
and 
$\fa_v := \{u \in \E^n: \langle u, \fg.v \rangle = 0\}$.
We will show:  
\begin{claim} \label{claim: 1}
There is a $g \in G$ such that 
$\fa_v = g . \fa_x$.
\end{claim}   
This will prove Proposition 
\ref{prop: regular bodies and polar representations}, because 
transitive action of $G$ on the set $\{a_y: y \in V\}$ of Cartan subspaces is one of the equivalent conditions
that defines (via Proposition \ref{prop: polar}) Dadok's notion of polar representation.  


To do this we show that the Farran-Robertson sections $L_\Phi$
are Cartan subspaces; recall that $L_\Phi$ is the linear space spanned by the centroids
of the faces in the maximal flag $\Phi$, and hence the linear span of the simplex 
$\tri$ whose vertices are those centroids, and indeed of the simplex $\tri'$ whose
vertices are those of $\tri$, with the exception of $0$ (the centroid of $\Omega$).
Since the topological boundary $\partial B$ contains representatives of all nonzero orbits,  
$L$ is a cross-section of the representation.  In the proof of Theorem 9 in \cite{FarranRobertson}, it 
is established that ``$G.y$ is perpendicular to $L$ for all $y \in L$."  In other words,  for all $y \in L$, 
$L \subseteq \fa_y$.  Since $L$ is a cross-section it contains representatives of principal orbits, i.e. 
regular elements; letting one such be 
$x$, we have $L \subseteq \fa_x$.  But since $x$ is regular 
$\fa_x$ is a minimal cross-section by Lemma 1 of \cite{Dadok1985}, so $L = \fa_x$, and $L$ is a Cartan 
subspace.  

To show 
that $G$ acts transitively on the set of Cartan subspaces, we fix one such subspace $\fa_x = L_\Phi$.   
Since $\tri_\Phi$ is a fundamental region for the action of $G$ on $B$, every nonzero $v$ on a principal orbit 
has the form $v=g.x$ for some $g \in G$ and some  $x$ in an element 
of $\tri_\Phi$.  We have $x \in  L_\Phi$.  By Proposition \ref{prop: action on flags} $g.\Phi$ is also a maximal flag, and (since by that Proposition $g$ takes centroids of face to centroids of faces) $g.L_{\Phi} = L_{g.\Phi}$.  Of course 
$g.x \in L_{g.\Phi}$; moreover, by the same argument used to establish that $L_{\Phi}$ is a Cartan subspace,
$L_{g.\Phi} = \fa_{g.x}$, and hence $L_{g.\Phi}$ is a Cartan subspace.  Since $L_{g.\Phi} = g.L_{\Phi}$, this shows that $\fa_v = g.\fa_x$, establishing Claim \ref{claim: 1}.
\end{proof}
 
\begin{proposition}[Corollary of Proposition \ref{prop: regular bodies and polar representations}]
\label{prop: regular compact convex sets and polar representations}
Let $\Omega$ be a regular compact convex set canonically embedded in
a Euclidean space $E$.  The representation of $\aut{\Omega}$ in $O(E)$
is polar.  The centroids of a maximal flag of $\Omega$ are a basis for
a Cartan subspace.
\end{proposition}

To study regular nonpolytopal convex compact sets $\Omega$, we may assume that the Lie
algebra $\fg$ of $\Aut \Omega$ is nontrivial, i.e. not equal to $\{0\}$.  For
compact Lie groups, this is equivalent to assuming the group is not
finite.
The next definitions and
results concern situations where the Lie group is connected, and hence
(except in the degenerate case of the trivial group) has nontrivial
Lie algebra.  If $\rho$ is a representation of a Lie group $G$, we
write $d\rho$ for the derived representation of $G$'s Lie algebra.

\begin{definition}[\cite{Dadok1985}]\footnote{The statement in \cite{Dadok1985} appears to have a couple
of typographical errors which we have corrected: it has $\fa$ in place
of $\fg$ in the last line, and omits the $L^{-1}$ that we've inserted
on the left-hand side.  Other than that, we quote \cite{Dadok1985} verbatim.}
\label{def: symmetric space representation}
Let $G$ be a compact connected Lie group, $\fg$ its Lie algebra.  A
representation $\rho: G \rightarrow SO(V)$ is called a \emph{symmetric
  space representation} if there are a real semisimple Lie algebra
$\fh$ with Cartan decomposition $\fh = \fk \oplus \fp$, a Lie algebra
isomorphism $A: \fg \rightarrow \fk$, and an isomorphism $L$ of
linear spaces $V \rightarrow \fp$ such that $L \circ d\rho(X) \circ
L^{-1}(y) = [A(X), y]$ for all $X \in \fg, y \in \fp$.
\end{definition}

In other words, 
$L \circ d\rho(X) \circ L^{-1} = \ad(A(X))$.  

The symmetric spaces in the proposition are $H/K$,
for groups such that $K$ is compact with Lie algebra $\fk \iso \fg$,
and the Lie algebra $\fh$ of the real semisimple group 
$H$ has Cartan decomposition $\fh = \fk
\oplus \fp$, where $\fp$, as mentioned in the definition, can be identified with the
representation space $V$.   They are of noncompact type
since one normally takes the term ``Cartan decomposition" to imply that $\fk \oplus \fp$ is noncompact semisimple.  
However by the duality bijection between compact and noncompact symmetric spaces, the
definition also includes the isotropy representations associated with symmetric spaces of compact type, 
because the compact duals (with Lie algebra decomposition $\fk \oplus i \fp$)
give rise to equivalent polar representations $K \curvearrowright  i\fp$.

Symmetric space representations are polar, with $\fa$, defined as
usual as a maximal abelian subspace of $\fp$, a Cartan subspace in the
terminology of \cite{Dadok1985}.\footnote{Of course
$\fa$ is not a Cartan \emph{subalgebra} of $\fh$, unless $\fh$
  is a split real form of a complex Lie algebra.  But $\fa$ \emph{is}
  the intersection of $\fp$ with a particular type of Cartan
  subalgebra of $\fh$ (one that is ``maximally noncompact'', i.e. 
  whose intersection with $\fp$ is maximal).  The term ``Cartan subspace" was already
standard in the symmetric space context.}   The action of $G$ on $\fp$ is 
often called the \emph{isotropy representation} of $G$ in the context of symmetric space
theory; a symmetric space representation in the above sense is just one that is equivalent to such an 
isotropy representation.

Dadok is able to use the symmetric space representations to classify
the irreducible polar representations (up to orbit equivalence) by
using the following theorem:

\begin{theorem}[Proposition 6 of \cite{Dadok1985}]
\label{theorem: associated symmetric space representation}
Let $\rho: G  \rightarrow SO(V)$ be a polar representation of a connected
compact Lie group $G$.  There is a connected Lie group $\tilde{G}$ with 
symmetric space representation $\tilde{\rho}: \tilde{G} \rightarrow SO(V)$
such that the $G$-orbits and the $\tilde{G}$-orbits in $V$ coincide.  
\end{theorem}

    The key point, from \cite{MaddenRobertson},
is that ``for a[n irreducible] noncompact symmetric space $H/K$, knowledge of $K$ and
the dimension of $H/K$ are sufficient to determine $H$.  But these are
already given by the symmetry group $G$ and the dimension $n$,
respectively.  Thus we have associated to the given $n$-solid $B$ a
symmetric space $H/K$.'' (The bracketed modification is necessary for the truth 
of the claim.) 

The following incorporates and extends Proposition 
\ref{prop: regular bodies and polar representations}, giving more detail on the embedding of a regular
convex body in polar representations of its symmetry group and some subgroups thereof, including consequences of Theorem \ref{theorem: associated symmetric space representation}, 
mostly extracted from the discussion on pp. 366-367 of \cite{MaddenRobertson}, and from \cite{Dadok1985} .
\begin{proposition}\label{prop: polytope as Weyl orbit}
Let $B$ (of dimension $n$) be a regular convex body in a Euclidean space $V \iso \E^n$, with centroid $0$
and symmetry group $G$, and $\flie(G) = \fg$.  Fix a maximal flag $\Phi$ and a corresponding fundamental 
region $\tri := \tri_\Phi$ (cf. Theorem \ref{theorem: fundamental region}).  Then 
\begin{enumerate}
\item \label{item: L is an a_x}
The Farran-Robertson section $L_\Phi \subseteq V$ is a linear cross-section of the 
form $\fa_x$ for regular $x$.
\item \label{item: conjugacy of Ls}
For every regular $x' \in V$, $\fa_{x'}$ is a linear cross-section and there is a $g \in G$ 
such that $\fa_{x'} = L_{g.\Phi}$.  
\item \label{item: both inclusions polar} Both the inclusions $G \rightarrow O(V)$ and $G_0 \rightarrow SO(V) \rightarrow O(V)$ are polar.  
\item \label{item: restricted root space is FR section} 
By Theorem \ref{theorem: associated symmetric space representation} we may make the 
identification  
$V \iso \fp$ where $\fp$ is the isotropy representation of $\tilde{G} \le G$ associated with an irreducible
noncompact symmetric space 
$H/\tilde{G}$.   (So, $\fh = \fg \oplus \fp$ is a Cartan decomposition of $\fh$.)  
The representation $G \curvearrowright V \iso \fp$ is the restriction of the 
representation $\tilde{G} \curvearrowright \fp$.  We have $\tilde{\fg} := \flie{\tilde{G}} = 
\fg := \flie(G) = \flie(G_0)$.  
 Every maximal abelian subspace 
$\fa \subseteq \fp$ is a Cartan subspace of the polar representations 
$\tilde{G} \curvearrowright \fp, G \curvearrowright \fp$, and $G_0 \curvearrowright \fp$,
and is a Farran-Robertson section of $B$, so $\pi(B) = \fa \intersect B$.
\end{enumerate}
\end{proposition}

\begin{proof}
The polarity of the 
inclusion $G \rightarrow O(V)$ in item \ref{item: both inclusions polar} is Proposition \ref{prop: regular bodies and polar representations}, while 
items \ref{item: L is an a_x} and \ref{item: conjugacy of Ls} are the 
key points in its proof.   Item (iii) in Proposition \ref{prop: polar} characterizes the polarity
of a representation entirely in terms of the derived representation $\fg \rightarrow \fso(V)$, whence the inclusion
$G_0 \rightarrow SO(V) \rightarrow O(V)$ is also polar since $\fg \equiv \flie(G) = \flie(G_0)$.

Turning to item 4, Theorem \ref{theorem: associated symmetric space representation} and its proof in 
\cite{Dadok1985}  show that because $\tilde{G}$ has the same orbits as $G$, 
 its Lie algebra $\tilde{\fg}$ is a quotient of $\fg$ and the derived symmetric space representation 
$\tilde{\fg} \rightarrow \fso(V)$ factorizes $\fg \rightarrow \fso(V)$.   So 
$\fg = \tilde{\fg} \oplus \fz$, where $\fz$ is central in $\fg$ and does not act on $V$, i.e. $\fz.V = \{0\}$.
Hence the connected subgroup $Z$ of $G_0$ with $\flie(Z) = \fz$ acts trivially on $V$.  Since $G_0$ is a subgroup 
of $SO(V)$ this implies that $Z$ is trivial, whence $\fg = \tilde{\fg}$.  But the Riemannian symmetric space
doesn't change if we replace the maximal compact subgroup $\tilde{G}$ of ${H}$ with a locally 
isometric one, so we can assume that $G_0 = \tilde{G}/\tilde{Z}$, where 
$\tilde{Z} := \{g \in \tilde{G}:   g|_{V} = \text{id} \}$.  It follows that $V = \fp$, $\fh = \fg \oplus \fp$.
Since it is known from the theory of symmetric spaces that the subspaces $(\fg.x)^\perp$ for regular $x$, 
i.e. the Cartan subspaces of Definition \ref{def: Cartan subspace}, 
of a symmetric space isotropy representation are precisely the maximal abelian subspaces of $\fp$, we may choose
any such subspace as a Cartan subspace.  The Cartan subspaces of the polar representations of $G$, $\tilde{G}$, and $G_0$ coincide because their Lie algebras do; consequently $\fa$ is a Farran-Robertson section for 
$G \rightarrow V$ (with respect to some choice of maximal flag of $B \subset V$), and $\pi(B) = \fa \intersect B$.  
\end{proof}

Madden and Robertson do not state the above Proposition (Proposition 
\ref{prop: polytope as Weyl orbit}) as formally as we do.  But their classification
applies it to all of the irreducible polar representations that are symmetric
space representations in the sense of Definition \ref{def: symmetric
  space representation}, which were classified by Cartan.    Combined with the regularity of 
the Farran-Robertson polytope (Theorem 
\ref{theorem: Farran-Robertson polytope is regular}), and Coxeter's work \cite{CoxeterBook} (building on 
the classification of regular polytopes by Schl{\"a}fli \cite{Schlaefli} and a construction by Wythoff) classifying the 
embeddings of regular polytopes as orbitopes in finite groups generated by reflections, this enables
them to classify regular convex bodies, a classification given in their Tables 2, 3,
and 4.   We formally summarize their classification as Theorem \ref{theorem: Madden-Robertson classification} below, 
and reproduce their tables, with minor changes and the addition of the rightmost column, 
indicating which symmetric spaces are associated with Euclidean Jordan algebras, as 
Tables \ref{table: classical noncompact ss reps}, \ref{table: exceptional noncompact ss reps} and
\ref{table: ss reps from complex groups} below.  We precede this with a brief sketch of their argument 
involving Coxeter's work.


Groups generated by a finite set of reflections in finite-dimensional real affine space, or isomorphic to 
such a concrete group, are usually called \emph{reflection groups}; each finite (not merely finitely generated) reflection group has a canonical representation on Euclidean space $\E^n$, in which it is generated by a finite number of orthogonal reflections through linear hyperplanes and there is no subspace on which the action is trivial.   The automorphism groups of regular polytopes are finite reflection groups; when the polytope is
canonically embedded in Euclidean space, its automorphism group acts acts in this canonical representation.
If $W$ is a finite reflection group, we use the term $W$\emph{-orbitope} for  
the convex hull $\conv{W.v}$ of an orbit of $W$ in its canonical representation, and also call $\conv{W.v}$ 
\emph{the }$W$\emph{-orbitope of }$v$, or simply  
\emph{the orbitope of} $v$ when $W$ is clear from context.  
Coxeter considered all finite reflection groups groups $W$, specified by what are now called Coxeter diagrams, and showed that the only regular $W$-orbitopes  are, up to dilation by a positive scalar, convex hulls of orbits of particular points in the canonical representation, identified by marking a node of the Coxeter diagram.

With a finite reflection group $W$ fixed, a unit-length normal $\alpha \in \E^n$ to the fixed hyperplane $\alpha^\perp$ (the ``mirror") of a reflection that is an element of $W$ is called a \emph{root}, and the associated reflection is denoted by $s_\alpha$.   (A different normalization constant is sometimes chosen, and in the theory of Weyl groups of Lie groups, which are a subset of the finite reflection groups, a different normalization is often used, in which not 
all roots have the same length.)   Thus each reflection is associated with two roots $\pm \alpha$, and the set $R$ of roots is finite.  Since for any $g \in O(n)$, $g s_\alpha g^{-1}$ is  the reflection $s_{g \alpha}$, in particular for any two roots $s_\alpha s_\beta s_\alpha$ is the reflection with root $s_\alpha \beta$, so 
$R$ is $W$-invariant, $WR = R$.  

A set $\Phi$ of roots is called \emph{simple} if every root $\beta \in R$ can be written 
$\beta = \sum_{\alpha \in \Phi} c_\alpha \alpha$, where the coefficients $c_\alpha$ are all of the
same sign (which may of course depend on $\beta$), or zero.   Simple sets of roots exist, and are 
bases for the canonical representation space $\E^n$; the associated reflections $s_\alpha$ are a minimal set of generators for $W$. 
The complement of the union of the mirrors of the reflections in $R$ is a dense subset of Euclidean 
space,  and we call its connected components \emph{open Weyl chambers} and their closures 
\emph{Weyl chambers}.\footnote{The terminology is not uniform in the literature; sometimes ``Weyl chamber" is 
used for the open Weyl chambers, and ``closed Weyl chamber" for their closures.}   The Weyl chambers are regular cones with simplicial 
base, and are fundamental regions for the group action.   The elements of the basis dual to the simple roots (when the roots are normalized in a particular way) are called the \emph{fundamental weights}.  The cone generated by the fundamental weights
is a Weyl chamber, called the \emph{positive Weyl chamber} $C_+$: each fundamental weight generates an extremal ray of $C_+$, and each extremal ray of $C_+$ is generated by a fundamental weight.

For a group generated by a finite set of linear reflections in $\E^n$ to be finite, obviously the product 
$s_\alpha  s_\beta$ of each pair of reflections must 
have finite order $n_{\alpha \beta}$.  (It is a nontrivial fact that this is sufficient.)    
It follows from this and the fact that the product of two reflections is a rotation by twice the angle between them, that the angles between $\alpha$ and $\beta$ must be rational multiples $k/n$ of 
$\pi$.  For $\alpha, \beta$ in a simple set $\Phi$, we will have $k=n_{\alpha \beta} - 1$, i.e. the angles are $\pi - \pi/n_{\alpha \beta}$, although we still say the lines generated by the roots have angle $\pi/n_{\alpha \beta}$. 
The Coxeter diagram of the group is determined by these angles for a simple set of roots: its nodes correspond to the simple roots and edges between nodes correspond to the angles other than $\pi/2$ between the lines generated by roots---these edges are labeled by the integer $n_{\alpha \beta}$ 
(less commonly, $n_{\alpha \beta}$ lines are used to link the pair of nodes, as in a Dynkin diagram).  
Usually the label  $n_{\alpha \beta}=3$ is omitted, so the angle $\pi/3$ is represented by an unlabeled edge; there is no
edge between nodes corresponding to pairs of roots at angle $\pi/2$.\footnote{Coxeter diagrams are also used to specify not-necessarily-finite, but discrete, groups generated by a finite number of reflections in finite-dimensional real affine space, and further groups sharing algebraic properties with these---see e.g. \cite{HumphreysCoxeterGroups}.  In this case, the labels on edges 
are drawn from the set $\{4,5,6,...\} \union \{\infty\}$, the label indicating the order of the product of the two reflections associated with the linked pair of nodes, still with order $3$ unlabeled and no edge for order $2$.  The full set of groups thus specified by Coxeter diagrams are known as \emph{Coxeter groups}.  For finite reflection groups, only the integer labels appear.}   It turns out that the relations $(s_\alpha s_\beta)^{n_{\alpha \beta}}=1$ between pairs of generators encoded in such a diagram suffice to uniquely determine a finite reflection group 
(this nontrivial theorem immediately implies the nontrivial fact mentioned above).
  
For the Weyl groups of Lie groups, which includes the reflection groups relevant to the classification 
of nonpolytopal regular
convex bodies, only the labels $4$ and $6$ appear, reflecting the fact that in this case the angles between lines through the simple roots are limited to $\pi/2$, $\pi/3$, $\pi/4$, and $\pi/6$.  In this case the Coxeter diagram, in the version with angles represented by edge multiplicities, is just the Dynkin diagram with the information about root lengths (associated with the different normalization usually used in this case) omitted.    

Returning to the general case, the point identified by marking a node is, up to (strictly) positive scalar multiples, the fundamental weight dual to the marked root.  As mentioned above, Coxeter showed that the regular polytopes arise as convex hulls of 
orbits (i.e. the orbitopes) of particular fundamental weights; 
every such polytope arises as the orbitope of the weight specified by marking an end node (node connected to only one other node) in some Coxeter diagram, although not all end nodes have such polytopes as orbitopes, and some non-end nodes give rise to regular polytopes.  To classify non-polytopal regular convex bodies, Madden and Robertson needed only to consider Coxeter's classification of regular orbitopes for the case of finite reflection 
groups that are Weyl groups of irreducible symmetric spaces, 
cf. the last two paragraphs on p. 366 of \cite{MaddenRobertson}.

\begin{theorem}[Madden-Robertson classification of regular convex bodies \cite{MaddenRobertson}]
\label{theorem: Madden-Robertson classification}
\mbox{}
\begin{enumerate}
\item 
Let $G/K$ be an irreducible noncompact symmetric space, with the compact group $K$ connected, 
and let $K \curvearrowright \fp$ be the isotropy representation of $K$, i.e. 
$\fg = \fk + \fp$ is the Cartan decomposition of $\g := \flie \, G$ and $K \curvearrowright \fp$ is the restriction to $K$ of the corestriction to $\fp$ of the adjoint action $G \curvearrowright \fg$.  
Let $\fa$ be a maximal abelian 
subspace (also called Cartan subspace) of $\fp$ and let $v \in \fa$ be such that $P := \conv{W_K.v} \subset \fa$ is a regular polytope, 
where $W_K := K_\fa/K^\fa$ is the Weyl group of $K$.  Then $B:= K.P = \conv{K.v}$ is a nonpolytopal 
regular convex body, $P$ is its Farran-Robertson polytope $\pi(B)$, and $P = \fa \intersect B$  .  
$B$ is completely determined, up to affine isomorphism, by the affine isomorphism class of $\pi(B)$ and the 
symmetric space.
\item 
Conversely, for every regular convex body $B$ that is not a polytope, there exists an 
irreducible noncompact symmetric space $G/K$ such that $B$ is an orbitope $\conv{K.v}$
in the isotropy representation $K \curvearrowright \fp$ of the compact connected Lie group $K$.
Its Farran-Robertson polytope $\pi(B)$ is $\fa \intersect B$, and is a $W_K$-orbitope.  
\item 
In the above items $v$ may be taken to be any strictly positive multiple of one of a certain set of 
fundamental weights in $\fa$.  Any weight $w$ for which $\conv{W_K.w}$ is a regular 
polytope may be chosen.  These
weights were classified by Coxeter: fundamental weights are specified by marking a node of 
the Coxeter diagram, and Coxeter indicated which marked nodes correspond to weights giving
rise to regular polytopes.  
\item
The list of irreducible noncompact symmetric space representations presented as $G/K$ for connected
$K$, together with information on their
dimensions, ranks, and the regular polytopes that occur as Weyl orbitopes in Cartan subspaces, is
given in Tables \ref{table: classical noncompact ss reps}, \ref{table: exceptional noncompact ss reps} and
\ref{table: ss reps from complex groups}, which except for their last column are
essentially Tables 2,3, and 4 of \cite{MaddenRobertson}.
\end{enumerate}
\end{theorem}

Note that the
same regular convex body may appear more than once in the tables describing the classification.    
These occurrences include the coincidences between noncompact irreducible symmetric
  spaces already noted by Cartan (cf. \cite{HelgasonBook}, pp. 519-520), which are  are finite in
  number, plus some cases, including some infinite series, where the same invariant regular convex body $B$ appears in different
symmetric spaces.  For the spectral strongly symmetric bodies, the coincidences between symmetric spaces
arising from EJAs are indicated in the tables by coincidences of EJAs in the rightmost column.\footnote{For completeness we note 
a few coincidences between listed EJAs, which only happen for the polytope $\tri_1$: $Herm(2,\R) \iso \R^2 \oplus \R, 
Herm(2,\C) \iso \R^3 \oplus \R$ and $Herm(2, \H) \iso \R^5 \oplus \R$.  (Although it only occurs once in the table, we note 
that $\R^9 \oplus \R$ may be interpreted as $Herm(2,\Oct)$.)}  The only other coincidences
between spectral strongly symmetric convex bodies in the tables are the actions of proper subgroups of $SO(n)$ on 
balls $B_n$, which occur in type  AIII and CII when $p=1$.    Appendix \ref{appendix: more information on classification}
contains more detail.

\begin{sidewaystable}

\caption{Classical noncompact symmetric spaces with associated isotropy representations and Farran-Robertson polytopes (adapted from \cite{MaddenRobertson}, Table 2)}
\label{table: classical noncompact ss reps}
\[
\begin{array}{c}
\begin{array}{lcccccc} 
\hline
  \text{Type} & \text{Symmetric space } G/K & \shortstack{Rank of \\ symmetric \\ space} & \shortstack{Dimension of \\ isotropy \\ representation} & \text{Root space} & \text{Polytope} & \text{EJA} \\
 \hline 
\text{AI} & SL(n,\R)/SO(n) & n-1 & (n-1)(n+2)/2 & A_{n-1} & \tri_{n-1} & Herm(n,\R) \\
\text{AII} & SU^*(2n)/Sp(n) & n-1 & (n-1)(2n+1) & A_{n-1} & \tri_{n-1} & Herm(n,\H) \\
\text{AIII} & SU(p,q)/S(U_p \times U_q) & q & 2pq & \begin{cases} C_q & (q < p) \\ B_q &  (q=p) \end{cases} & \Box_q, \Diamond_q  & \R^2 \oplus \R ~~(p=q=1) \\
\text{BI} &  \begin{matrix} SO_0(p,q)/(SO(p) \times SO(q)) \\ p + q \text{ odd}, ~q < p \end{matrix}
& q & pq & B_q & \Box_q, \Diamond_q & \R^p \oplus \R  ~~(q=1) \\
\text{DI} & \begin{matrix} SO_0(p,q)/(SO(p) \times SO(q)) \\ p + q \text{ even} \end{matrix}
& q & pq & \begin{cases} B_q & (q < p) \\ D_q &  (q=p) \end{cases}  & \Diamond_q & \R^p \oplus \R  ~~(q=1) \\
\text{DIII} & SO^*(2n)/U(n) & \lfloor n/2 \rfloor = q & n(n-1) & 
\begin{cases} C_q & q \text{ odd } \\ BC_q & q \text{ even } \end{cases} & \Box_q, \Diamond_q & \R^2 \oplus \R ~~(q=1) \\
\text{CI} & Sp(n,\R)/U(n) & n=q & n(n+1) & C_q & \Box_q, \Diamond_q & \R^2 \oplus \R ~~(q=1) \\
\text{CII} & Sp(p,q)/(Sp(p) \times Sp(q)) & q & 4pq & \begin{cases} C_q & (q=p) \\  BC_q & (q < p) \end{cases} & \Box_q, \Diamond_q & \R^4 \oplus \R ~~(p=q=1) \\
\hline 
\end{array} \\
\\
\end{array}
\]

\end{sidewaystable}

\begin{table}
\caption{Exceptional noncompact symmetric spaces with associated isotropy representations and Farran-Robertson polytopes (adapted from \cite{MaddenRobertson}, Table 3)}
\label{table: exceptional noncompact ss reps}
\[
\begin{array}{c}
\begin{array}{llccccc} 
\hline
  \text{Type} & \begin{matrix} \text{Symmetric space }G/K \\ \text{presented by } \fg, ~~\fk  \\\end{matrix} & \shortstack{Rank of \\ symmetric \\ space} & \shortstack{Dimension of \\ isotropy \\ representation} & \text{Root space} & \text{Polytope} & \text{EJA} \\
 \hline 
\text{EIII} & \begin{matrix} \fe_{6(-14)}  & \fso(10) \oplus \R \end{matrix}  & 2 & 32 & B_2 & \Box_2 & \\
\text{EIV} & \begin{matrix} \fe_{6(-26)} & \ff_4  \end{matrix} & 2 & 26 & A_2 & \tri_2 & Herm(3,\Oct) \\
\text{EVI} & \begin{matrix} \fe_{7(-5)}  & ~~\fso(12) \oplus \fsu(2)  \end{matrix}  & 4 & 64 & F_4 & \text{24-cell} &  \\
\text{EVII} & \begin{matrix} \fe_{7(-25)}  & \fe_6 \oplus \R  \end{matrix} & 3 & 54 & C_3 & \Box_3, \Diamond_3 & \\
\text{EIX} & \begin{matrix} \fe_{8(-24)}    & \fe_7 \oplus \fsu(2) \end{matrix}  & 4 & 112 & F_4 & \text{24-cell} &  \\
\text{FI} & \begin{matrix} \ff_{4(4)}    & ~~~~\fsp(3) \oplus \fsu(2)   \end{matrix} & 4 & 28 & F_4 & \text{24-cell} &  \\
\text{FII} & \begin{matrix} \ff_{4(-20)}  & \fso(9)   \end{matrix} & 1 & 16 & A_1 & \tri_1 &  \\
\text{G} & \begin{matrix} \fg_{2(2)} &   ~~ ~\fsu(2) \oplus \fsu(2)  \end{matrix}  & 2 & 8 & G_2 & \text{hexagon} & \\
\hline 
\end{array} \\
\\
\end{array}\]
\end{table}

\begin{table}
\caption{Noncompact symmetric spaces arising as $K^{\C}/K$ for simple $K$, with associated isotropy representations and Farran-Robertson polytopes (\cite{MaddenRobertson}, Table 4)}
\label{table: ss reps from complex groups}
\[
\begin{array}{c}
\begin{array}{lcccc} 
\hline
  \begin{matrix}\text{Type and} \\ \text{root space} \end{matrix} 
& \text{Symmetric space} & 
\begin{matrix}\text{Dimension of} \\ \text{isotropy} \\ \text{representation} \end{matrix} & \text{Polytope} & \text{EJA}  \\
 \hline 
A_n (n \ge 1) & SL(n+1, \C)/SU(n+1) & n(n+2) & \tri_n & Herm(n,\C) \\
B_n (n \ge 2) & SO(2n+1, \C)/SO(2n+1) & n(2n+1) & \Box_n, \Diamond_n &  \\
C_n (n \ge 3) & Sp(n, \C)/Sp(n) & n(2n+1)  & \Box_n, \Diamond_n &  \\
D_n (n \ge 4) & SO(2n, \C) / SO(2n) & n(2n-1) & \Diamond_n & \\
F_4 & F_4^{\C}/F_4 & 52 & \text{24-cell} & \\
G_2 & G_2^{\C}/G_2 & 14 & \text{hexagon} & \\
\hline 
\end{array} \\
\\
\end{array}\]
\end{table}

\section{Classification of strongly symmetric spectral convex compact sets}
\label{sec: main result}

In this section we apply the theory of the preceding sections to prove
our main result, which is the ``only if'' direction of Theorem
\ref{theorem: main}: that strongly symmetric spectral convex sets are
normalized EJA state spaces or simplices.  We proceed by way of
several intermediate propositions.

\begin{proposition}\label{prop: strongly symmetric spectral implies regular}  
Let $\Omega$ be a strongly symmetric spectral convex compact set.  
Then $\Omega$ is regular.
\end{proposition}

In order to prove this proposition we first establish: 
\begin{lemma}\label{lemma: flags and frames}
Let $\Omega$ be a strongly symmetric spectral convex set.
Let $\omega_1,...,\omega_r$ be a maximal frame in $\Omega$.  The
sequence 
\beq\label{eq: faces generated by subsequences} F_i =
\bigvee_{1 \le j \le i} \{\omega_j\}, ~~ i \in \{1,...,r\} \eeq 
is a
maximal flag.  Conversely, let $(F_1,...,F_r)$ be a maximal
flag of $\Omega$. Then there exists a maximal frame
$\omega_1,\omega_2,....,\omega_{r}$ such that (\ref{eq: faces generated by
  subsequences}) holds.
\end{lemma}

In other words, the formula 
(\ref{eq: faces generated by subsequences}) 
gives a bijection between maximal frames in $\Omega$ and 
maximal flags of $\Omega$.

\begin{proof}
Let $X=[\omega_1,...,\omega_{r}]$ be a maximal frame.  It follows from
item \ref{item: frames and faces} of Proposition \ref{prop:
  consequences} that the initial segments of $X$ generate a sequence
of faces, the $F_i$ of (\ref{eq: faces generated by subsequences}),
each properly contained in the next.  This is 
a flag, which we will call
$\Phi_X$.  In this sequence, the rank $|F_i|$ of $F_i$ is $i$.  Suppose this
flag is not maximal.  Then it can be enlarged, either by extending it
before $F_1$, or by extending it after $F_r$, or by inserting some
face $G$ with $F_i \subsetneq G \subsetneq F_{i+1}$.  It cannot be
extended before $F_1 = \{\omega_1\}$, because the only face below the
pure state $\omega_1$ is the improper face $\emptyset$.  It must have
$F_r = \Omega$ by Proposition \ref{prop: consequences} (\ref{item: frame extensions and
  maximal frames}), so it cannot be extended beyond $F_r$. 
So there
must be an $i \in \{1,...,r\}$ and a face $G$ such that $F_i \subsetneq
G \subsetneq F_{i+1}$.
By Proposition \ref{prop: consequences} (\ref{item: frames and faces}) $F_i \subsetneq G
\subsetneq F_{i+1}$ implies $|F_i| < |G| < |F_{i+1}|$, which
contradicts the fact, observed above, that $|F_i|=i$ and $|F_{i+1}| =
i+1$ by construction.  Since every way of extending the flag $\Phi_X$
is inconsistent with the maximality of the frame $X$, $\Phi_X$ is 
a maximal flag.

Conversely, suppose $\Phi= \{F_1,...,F_r\}$ is a maximal flag.  By
Proposition \ref{prop: consequences} (\ref{item: frames and faces}) each $F_i$ is the join of (the
singletons corresponding to) a frame, whose cardinality is $|F_i|$.
We will show (``Claim 1'') that $F_1 = \{\omega_1\}$ for some extremal
point $\omega_1$, and (``Claim 2'') that for each $i \in \{2,...,r\}$,
there exists an extremal $\omega_i$, distinguishable from every 
state in the frame 
$F_{i-1}$, such that $F_{i} = F_{i-1} \join \omega_i$.  It follows
from the associativity of join that (\ref{eq: faces generated by
  subsequences}) holds for each $i$, and since $F_r = \Omega$ for a
maximal flag, $\omega_1,...,\omega_r$ generates $\Omega$.  
Since $\Omega$ is generated by a maximal frame, and by
Proposition \ref{prop: consequences}(\ref{item: frames and faces})
all frames generating the same face have the same cardinality, 
$\omega_1,...,\omega_r$ has the same cardinality as a maximal frame,
whence it \emph{is} a maximal frame.

To show Claim 2 we use the fact, which is part of Prop.  \ref{prop:
  consequences}(\ref{item: frame extensions and maximal frames}), that
if $F \subsetneq G$, any frame for $F$ extends to a frame for $G$ 
by adjoining a frame for $F' \meet G$.
Consider $F = F_{i-1}, G = F_{i}$, for $i \in \{2,...,r\}$.  The frame
for the face $F'_{i-1} \meet F_{i}$, whose join with $F_{i-1}$ is
$F_{i}$, is nonempty because $F$'s containment in $G$ is
strict.\footnote{Were $F' \meet G = 0$ (i.e. $\emptyset$) then we'd
  have $(F' \meet G) \join F = F$, while orthomodularity says $(F'
  \meet G) \join F = G$.}.  Say it is $\Sigma = [\sigma_1, ...,
  \sigma_m]$, for $m \ge 1$.  We show that $m=1$.  If $m>1$, then we
can extend the flag $\Phi$ to a flag $\tilde{\Phi}$ defined by
$\tilde{F_j} = F_j$ for $j \in \{1,...,i-1\}$, $\tilde{F}_{j} :=
\tilde{F}_{j-1} \join \sigma_{j-i+1}$ for $j \in \{i,...,i+m-1\}$,
and $\tilde{F}_j := F_{j-m+1}$ for $j \in \{i + m,...,r\}$.  For
$m>1$, $\Phi$ is a proper subflag of $\tilde{\Phi}$, contradicting
$\Phi$'s maximality.  So we must have $m=1$, and $\tilde{\Phi} =
\Phi$.

To show Claim 1, that $F_1 = \{\omega_1\}$ for some extremal
$\omega_1$, we use essentially the same argument: there is a frame
$\eta_1,..,\eta_{|F_1|}$ for $F_1$, nonempty because $F_1 \ne
\emptyset$, and if $|F_1| \ne 1$, then we can extend the flag by
prefixing it with the nonempty sequence of subfaces $[H_i := \bigvee_{k
  \in 1,..i} \{ \eta_k \}]_{i \in \{1,...,|F_1|-1\}}$, generated by the 
initial segments of 
that frame.  Since the flag was maximal, the extension must be
impossible, so $|F_1|=1$, hence $F_1 = \{\eta_1\}$, with $\eta_1$
extremal.
\end{proof}

Proposition \ref{prop: strongly symmetric spectral implies regular}
follows almost immediately.
\begin{proof}[Proof of Proposition \ref{prop: strongly symmetric spectral implies regular}]
Let $\Phi_1 = \{F_1,\ldots,F_r\}$ and $\Phi_2 = \{G_1,\ldots,G_r\}$ be two
maximal flags of $\Omega$.  Then by Lemma \ref{lemma: flags and
  frames} the faces of $\Phi_1$, resp. $\Phi_2$, are the sequences of
faces generated by the initial segments of the maximal frames
$\omega_1,...,\omega_{r}$, $\eta_1,...,\eta_{r}$ respectively, defined
by the bijection (\ref{eq: faces generated by subsequences}).  By
strong symmetry, there exists $g \in \aut{\Omega}$ such that for all
$i \in \{1,...,r\}$, $g \omega_i = \eta_i$.  It follows that $g\Phi_1
= \Phi_2$.
\end{proof}

Next we determine the implications of the conjunction of strong
symmetry and spectrality for the polytope $\pi(\Omega)$ that exists
because of regularity.

\begin{proposition}\label{prop: polytope of sss body is a simplex}
Let $\Omega$ be a strongly symmetric spectral convex compact set of
rank $r$.  The polytope $\pi(\Omega)$ is a simplex with $r$ vertices, 
which constitute a maximal frame.
\end{proposition}

\begin{proof}
Let the dimension of $\Omega$ be $n$; without loss of generality we view $\Omega$ as canonically embedded in
$\E^n$, i.e. assume that the barycenter of $\Omega$ is $0
\in \aff{\Omega} \iso \E^n$, viewing $\aff \Omega$ as a vector space
as described earlier, and introduce an invariant inner product.
Picking a maximal flag (it does not matter which one), $F_1,...,F_r$,
and letting $\omega_1,...,\omega_r$ be the maximal frame corresponding
to it via the bijection (\ref{eq: faces generated by subsequences}),
we consider the simplex $\triangle'$ of Farran and Robertson (defined
following Theorem \ref{theorem: fundamental region} above), which,
using the description of the barycenters of faces from Proposition
\ref{prop: consequences} (\ref{item: barycenters of faces}), is equal
to $\triangle(\omega_1, (\omega_1 + \omega_2)/2,..., (\omega_1 +
\cdots \omega_{r-1})/(r-1))$.  We next identify the linear subspace
$L$ of $\aff(\Omega)$ generated by this simplex.  Since the
barycenters just listed are manifestly linearly independent in
$\aff(\Omega)$ (as must be any set of barycenters of faces of a 
flag), $L$ is $(r-1)$-dimensional (as also noted following Theorem
\ref{theorem: fundamental region}).  It is easy to see
that $\omega_1,...,\omega_{r-1}$ are a (linear) basis for the space
$L$.

So far in this proof, we have been working in the setting where
$\aff{\Omega}$ is viewed as a Euclidean vector space $E$ with
$c(\Omega)$ as its zero, and inner product chosen so that
$\aff{\Omega}$ is a subgroup of the orthogonal group.  It is now
useful to go to the setting, described in Section \ref{sec:
  background} and used extensively in Proposition \ref{prop:
  consequences}, in which this Euclidean vector space is embedded as
an \emph{affine} subspace in the vector space $V$ of dimension one
greater, in such a way that the embedded version of $E \iso
\aff{\Omega}$ does not contain the origin of $V$, and hence the cone
$V_+ = \R_+ \Omega \subset V$ has affine dimension one more than
$\Omega$'s.  Of course the barycenter of $\Omega$, which was $0$ in $E
\iso \aff{\Omega}$, embeds as a nonzero element $c$ of $V$.  As
described just before Proposition \ref{prop: barycenter from
  group averaging}, we extend the action of $SO(E)$ to an action of
$SO(E)$ on $V$, which must fix the ray over $c$ (pointwise), and equip
$V$ with a corresponding invariant inner product such that this action
of $SO(E)\iso SO(n)$ is as a subgroup of $SO(V) \iso SO(n+1)$.  By
Proposition \ref{prop: consequences}(\ref{item: perfection})  
this invariant inner product can be chosen to be
self-dualizing for the cone $V_+$ and such that all pure states have
unit Euclidean norm, and we do so.  

Now $L$, which was a linear subspace in $E \iso \aff{\Omega}$, is
embedded in $V$ as an affine subspace but--since it is an affine subspace of 
$\aff{\Omega}$--not a linear subspace.\footnote{This was equivalent to its being linearly
  generated by $c_1,...,c_{r-1}$ when we had identified $c_r$ with $0
  \in E$.}  $L$ was the linear space spanned by 
$c_1,...,c_{r-1}$; as an affine space, it is generated by 
$c_1,...,c_{r-1}, c_r$, where $c_r= c(\Omega)$.  

The simplex $\pi(\Omega) = L \intersect \Omega$ is therefore embedded
as the intersection of the affine subspace $L \subset V$ with
$\Omega$. 
$L \subset \aff{\Omega}$ is also affinely 
generated by $\omega_1,....,\omega_r \in \aff{\Omega}$.
When viewed as vectors in $V$ (the vector space of dimension one
greater than $\aff{\Omega}$), $\omega_1,...,\omega_{r}$ are
orthonormal.  Defining $\tilde{L}$ as the linear span of
$\omega_1,...,\omega_r$ in $V$, we have that $L = \tilde{L} \intersect
\aff{\Omega}$.
From this it follows that 
$\pi(\Omega) := L \intersect \Omega$ is just the intersection of
$\Omega$ with the subspace $\tilde{L} = \lin\{\omega_1,...,\omega_r\}
\subseteq V$.  

We will show that this intersection $\tilde{L} \intersect \Omega$ is
the simplex $\triangle(\omega_1,....,\omega_r)$.  Recall (Proposition
\ref{prop: consequences}, item \ref{item: frames orthogonal}) that the
extremal points $\omega_1,...,\omega_r$ are orthonormal in $V$, and
are therefore an orthonormal basis for their span $\tilde{L}$.  By the
self-duality of $V_+$, $\omega_1,...,\omega_r$ are also on extremal rays of the
dual cone with respect to the inner product.  So everything in $V_+$
has nonnegative inner product with each of $\omega_1,..., \omega_r$.
These constraints impose in particular that $\tilde{L} \intersect V_+$
lies in the closed positive halfspaces $H_i^+ := \{x \in \tilde{L}:
(\omega_i, x) \ge 0\}$, $i \in \{1,...,r\}$ of each the hyperplanes
$H_i = \{x \in \tilde{L}: (\omega_i, x) = 0\}$ in $\tilde{L}$.  These
constraints define a polyhedral cone which (using the mutual
orthogonality of the $\omega_i$) is identical to the cone over the
simplex $\triangle(\omega_1,...,\omega_r)$.  Since $\omega_1,...,
\omega_r$ are in $V_+$ and in $\tilde{L}$, we know that $\tilde{L}
\intersect V_+$ contains this cone, and since we have just shown that
$\tilde{L} \intersect V_+$ is contained in this cone, we have that
$\tilde{L} \intersect V_+$ is equal to it, and hence that $\pi(\Omega)
:= \tilde{L} \intersect \Omega = \triangle(\omega_1,...,\omega_r)$.

Since the states $\omega_1,\ldots,\omega_r$ are the vertices of the
Farran-Robertson polytope of $\Omega$, the faces $F_1,...,F_r$ of the
polytope defined by the formula (\ref{eq: faces generated by
  subsequences}) are a maximal flag of that polytope.  Then by
Proposition \ref{prop: flags and faces of B and P}, the faces
$H_{F_i}$ of $\Omega$ are also a maximal flag of $\Omega$.  Since
$H_{F_i}$ is $G^{c(F_i)}. F$, $c(F_i)$ is the centroid of $H_{F_i}$.
Since $c(F_i) = (1/i) \sum_{j=1}^i \omega_i$, $H_{F_i}$ is the face of
$\Omega$ generated by $\omega_1,...,\omega_i$, 
i.e. 
\beq
H_{F_i} = \bigvee_{i=1}^i \omega_i,  i \in \{1,...,r\}.
\eeq 
So by Lemma \ref{lemma: flags and frames}, 
$\omega_1,...,\omega_i$ are a frame in $\Omega$, not merely 
in $\pi(\Omega)$, and $\omega_1,...,\omega_r$ is a maximal frame.
\end{proof}

We now have results sufficient to prove Theorem \ref{theorem: main}, 
which we reiterate here.

\begingroup
\def\thetheorem{\ref{theorem: main}}
\begin{theorem}
A convex compact set is strongly symmetric and spectral if and only if
it is a simplex or affinely isomorphic to the space of normalized
states of a simple Euclidean Jordan algebra.
\end{theorem}
\addtocounter{theorem}{-1}
\endgroup

\begin{proof}

The ``if'' direction, that Euclidean Jordan algebra state spaces and
simplices are strongly symmetric and spectral, was already known
\cite{BMU} based on known results about Euclidean Jordan algebras); an
explicit proof was reviewed in Sections \ref{sec: Euclidean Jordan
  algebras} and \ref{sec: simplices}.

By Propositions \ref{prop: strongly symmetric spectral implies
  regular} and \ref{prop: polytope of sss body is a simplex}, the
strongly symmetric spectral convex compact sets are all to be found
among the regular convex compact sets whose Farran-Robertson polytope
is a simplex.  Since a regular polytope is its own Farran-Robertson
polytope, the only polytopes that are strongly symmetric and spectral
are the simplices, of every dimension.
To identify the nonpolytopal strongly symmetric spectral compact
convex sets we will use the Madden-Robertson classification of
nonpolytopal regular convex bodies $B$ in $\E^n$, which is given
in Tables 
\ref{table: classical noncompact ss reps}, \ref{table: exceptional noncompact ss reps} and 
\ref{table: ss reps from complex groups} (Tables 2, 3 and 4 of \cite{MaddenRobertson}).
Because each regular compact convex set admits a canonical
embedding as a regular convex body in $\E^n$, with symmetry group
equal to its affine automorphism group (Prop.  \ref{prop: canonically
  embedded cc sets}) all regular compact convex sets appear in the
table.  Conversely all table entries correspond to regular convex
compact sets, because all regular convex bodies, considered as convex
compact sets, i.e. forgetting about the Euclidean and vector space
structure of $\E^n$, are regular
(see Proposition \ref{prop: regular bodies and sets}).

All possibilities for nonpolytopal strongly symmetric spectral
convex bodies are found among the entries of Tables 
\ref{table: classical noncompact ss reps}, \ref{table: exceptional noncompact ss reps} and 
\ref{table: ss reps from complex groups}:  they are the cases 
 whose Farran-Robertson polytope is a simplex.
For each such case, Theorem \ref{theorem: Madden-Robertson
  classification} tells us, the given data suffices to determine, up
to isomorphism, the regular convex body as the convex hull of the
orbit of a (positive scalar multiple of a) fundamental weight in
$\fa$, corresponding to one of the ends of the Coxeter diagram
associated with the Weyl group.
The entries with simplicial polytope are those for which the polytope
is explicitly indicated to be a simplex, and the cases $q=1$ of
families where the polytope is $\Box_q$ or $\Diamond_q$, since $\Box_1
= \Diamond_1 = \tri_1$, the closed line segment, which is the unique
$1$-dimensional polytope.

Theorem \ref{theorem: sss bits
  are balls}, which states that strongly symmetric spectral convex
bodies whose maximal frames have size two must be balls, 
establishes that the $\Box_1$ and $\Diamond_1$ cases in the tables
in \cite{MaddenRobertson} (as well as the $\tri_1$ cases) all have
$\Omega$ a ball of some dimension; these are Jordan-algebraic normalized
state spaces (for the spin factors).\footnote{Appendix \ref{appendix: more information on classification} 
includes a case-by-case
    examination of the $\Box_1$ and $\Diamond_1$ cases, but this is not necessary for the proof.  
It includes
the cases in which the symmetric space is that associated with a spin factor (noted in the ``EJA" column in Tables 
\ref{table: classical noncompact ss reps}, \ref{table: exceptional noncompact ss reps} and 
\ref{table: ss reps from complex groups}), as well as the infinite series of symmetric spaces of type AIII and CII 
which illustrate the fact that balls $B_d$ may have actions by affine automorphisms, transitive on the sphere $\partial_e B_d
\iso S_{d-1}$,  of proper subgroups of the obvious rotation group $SO(d)$.}

In each of the symmetric space representations in these tables, 
the regular convex body is determined, up to 
isomorphism, by the regular polytope $\pi(B)$ and the symmetric space.  
A comparison of the cases with simplicial polytopes $\tri_n$ for $n \ge 2$ with our Table 
\ref{table: EJAs} (taken from \cite{FarautKoranyi}) shows that they all correspond to 
polar representations $K \curvearrowright \fp_0$ 
of compact groups $K$ coming from simple Euclidean Jordan algebras 
$V \iso \fp = \fp_0 \oplus \R e$, where $e$ is the Jordan unit.  
These are indicated in the rightmost columns of Tables \ref{table: classical noncompact ss reps}, \ref{table: exceptional noncompact ss reps} and \ref{table: ss reps from complex groups}.
By Proposition \ref{prop: polytope of sss body is a simplex}, in the strongly symmetric spectral cases, 
$\pi(B)$ is a simplex, and as embedded in $V$, its vertices are a maximal frame.  
The maximal abelian subspaces of $V \iso \fp$ (viewed as subspaces of $\fg = \flie(\aut{V_+})$) are 
of the form $\fa_0 \oplus \R e$, where $\fa_0$ are the maximal abelian subspaces of $\fp_0$.  
In, for example, \cite{FarautKoranyi}, Proposition VI.3.3, and the discussion preceding it, 
the Peirce decomposition $\oplus_{i,j=1}^r V_{ij}$ of a simple EJA of rank $r$ is used, and it is observed that
$\fa = \oplus_i V_{ii}$ where $V_{ii} = \R c_i$, and the $c_i, ~i \in \{ 1,\ldots,r\}$ are a Jordan frame.  
With $\Omega$ defined as usual as the normalized state space embedded in $V = \fp$, with the Jordan unit $e$ 
as order unit,  
we have that $\fa \intersect \Omega$ is the simplex
generated by the $c_i$, which is affinely isomorphic to $\fa_0 \intersect \Omega_0$, where
$\Omega_0 := \Omega - e/\tr e = \Omega - e/r$ is the translation of the normalized Jordan state space into $\fa_0$.  
This exhibits the strongly symmetric compact convex set $\Omega$ as affinely isomorphic to 
the $K$-orbit $\Omega_0$ of a Farran-Robertson polytope $\Omega_0 \intersect \fa_0$, which is a simplex.  
Because according to the Madden-Robertson classification
the Farran-Robertson polytope determines, up to isomorphism, the regular convex body obtained as its
$G$-orbit in a symmetric space representation, and all regular convex bodies (up to isomorphism) are
obtained in this way, all strongly symmetric convex bodies are isomorphic to Jordan state spaces.
\end{proof}

\begin{remark}
The reader may wonder how the irreducible Riemannian symmetric
spaces $H/K$ relate to the cones of squares in simple EJAs associated with the polar
representation, since
the interiors of these cones 
are precisely the irreducible symmetric cones, meaning the irreducible
open cones\footnote{with pointed closure, i.e.  proper cones rather
  than wedges.} $V^\circ_+$ that are noncompact Riemannian symmetric spaces $\aut_0{(V^\circ_+)}/K$ 
(where $K$ is a maximal compact connected subgroup) on which $\aut{(V^\circ_+)}$ acts via isometries.  
Writing $V_+ := \overline{V^\circ_+}$, note that $\aut{V^\circ_+} = \aut{V_+}$.  The answer lies
in recalling that the (reductive) lie algebra of $G=\aut{V_+}$ has in
general a Cartan decomposition $\fg = \fk \oplus \fp$, with respect to
which $\fp$ can be identified with the Jordan algebra $V$, and then
$V_+ = \exp \fp$ (cf. e.g.  \cite{FarautKoranyi}).  However, for
simple Jordan algebras also $\fg = \fg_{ss} \oplus \R$, where $\R$
generates the uniform dilations and the semisimple part $\fg_{ss}$ is
simple, and when we further break down $\fg_{ss} = \fk \oplus \fp_0$,
we see that $\fp_0$ is the traceless\footnote{We remind the reader that trace and determinant
are defined for EJAs in general (cf. \cite{FarautKoranyi}), and coincide with the usual matrix notions 
in the matrix cases $Herm(n, \D)$.)} part of $\fp$, which goes under
$\exp$ to the component of the unit-determinant part of $\fp$ that
lies in $V_+$, which is in fact a symmetric space $H/K$ for $H\equiv
G_{ss}$, the semisimple part of $\aut{V_+}$.  
 The symmetric cone is in
fact a reducible symmetric space, the product of $H/K$ with $\R$.\footnote{Concretely, the second
factor is represented as $\R_+$, but it is equipped with the multiplicative group structure, isomorphic
(via the exponential map) to the additive group structure on $\R$.}  The
$\R$ factor gives the flat direction in the tangent space of $V_+$,
i.e. the direction along a ray from the origin.  The interior of $V_+$ is 
the cone over the irreducible symmetric space associated with the
representation.  A nice example is the Lorentz cone, in which $H/K$ is
a paraboloid inside the cone, with focus on the axis of rotational
symmetry, and is in fact an orbit of the (connected, also known as ``orthochronous")  Lorentz
group; the full cone is obtained by including dilations.  
When $V$ is not simple, everything still works roughly as before
except that now the symmetric space is a product $\times_{i=1}^k
(H_i/K_i \times \R) \iso (\times_{i=1}^k H_i / K_i) \times \R^k$ of irreducible
symmetric spaces, with
$\R^k$ reflecting the possibility of independent dilations of the
cones over the simple factors $H_i/K_i$.
\hfill $\diamondsuit$
\end{remark}

\section{Discussion}
\label{sec: discussion}

We have characterized a particular subset of the
finite-dimensional Jordan algebraic state spaces---those corresponding
to simple Euclidean Jordan algebras, and to finite products of the
one-dimensional Euclidean Jordan algebra---as those convex sets
satisfying the properties of spectrality and strong symmetry.  In
this section, we consider extensions and implications of this result.
First we discuss various known ways in which the two properties can be
supplemented with additional ones to characterize precisely the state
spaces of complex quantum theory.  Then we discuss some other
characterizations of this class, or more general classes, of
Jordan-algebraic state spaces, and how these relate to our work.
Finally we review work in the setting of general probabilistic theories that derives
strong consequences from the conjunction of spectrality and strong
symmetry, or from sets of postulates that can be shown to imply these
two properties, emphasizing the new light our result throws on this work.

\subsection{From Jordan algebra state spaces to complex quantum theory}

Characterizations of quantum state space,
whether in terms of postulates whose appeal is mathematical, physical,
informational, or some combination of these, often proceed by first
characterizing Jordan-algebraic state spaces, or some subset thereof, and then adding an additional
postulate or set of postulates that narrows things down to standard, i.e. complex, quantum
theory.  In the following two subsections we describe two important classes of such postulates.  These
can, of course, be used in conjunction with Theorem \ref{theorem: main} to characterize the irreducible complex quantum
systems. 

\subsubsection{From Jordan state spaces to complex quantum theory via relations between continuous symmetries and observables}
The important characterizations by Alfsen and Shultz (\cite{Alfsen78a}, cf. \cite{ASBook2}) 
of the state spaces of two natural classes 
of Euclidean Jordan algebras, 
the JB-algebras and the JBW-algebras, which 
are Jordan analogues of the $C^*$-algebras and the von Neumann algebras,  were 
extended by them to characterize the state spaces of $C^*$-algebras and von Neumann algebras
(each of which reduces to direct sums of the state spaces of standard complex quantum theory,
in the finite-dimensional setting), in several ways.  One of these, by Proposition 10.27 of \cite{ASBook2}, is to postulate the existence of a ``dynamical correspondence'' (\cite{ASBook2}, Definition 6.10), 
 on the JB-algebra $A$.  The dynamical correspondence determines a unique $C^*$ product on $A + iA$.
 In the special case in which the JB-algebra is a JBW-algebra,  the
dynamical correspondences on $A$ are in bijection with the Connes orientations (Theorem 6.18 of \cite{ASBook2}); the 
existence of either of these determines a unique $W^*$ product on 
$A + iA$, and the normal state space of $A$ is isomorphic to the normal state space of this von Neumann algebra.

A dynamical correspondence $\psi$ on a JB-algebra $A$ is a linear map, $\psi: a \mapsto \psi_a$, of $A$ into the set of skew order-derivations of $A$, satisfying certain properties.  It is called \emph{complete} if it is surjective (note, also, that there is no  requirement that it be injective).  The order-derivations are the elements of the Lie algebra $\faut{V_+}$ of the 
group of affine automorphisms of the positive cone of a Jordan algebra.  This is the span of two complementary subspaces,
the self-adjoint (also called symmetric) and skew-adjoint (also called skew) order derivations, which in the finite-dimensional case are the 
$-1$ and $+1$ eigenspaces, respectively, usually denoted $\fp$ and $\fk$, of the Cartan involution on $\faut{V_+}$.  The self-adjoint ones may be identified with the space of Jordan multiplication operators, $L_{a}: b \mapsto a \bullet b$.  
The skew order-derivations are precisely the generators of one-parameter groups of automorphisms of the Jordan algebra (cf. Lemma 2.81 of \cite{ASBook2}).  The Jordan automorphisms are also precisely the Jordan unit preserving automorphisms of the
cone of unnormalized states, and hence in a
manifestly self-dual representation in our finite-dimensional setting, they are also precisely the (linearized extensions of) affine automorphisms of $A$'s normalized state space (cf. 
\cite{FarautKoranyi}).   The conditions defining a dynamical correspondence are (1) that the commutator of 
the images of Jordan algebra elements $a$ and $b$ is the negative of the commutator of the corresponding Jordan multiplication operators, that is, that $[\psi_a , \psi_b] = - [L_a, L_b]$, and (2) that the image of an element annihilate that element, $\psi_a a = 0$.  

The first of these conditions requires Jordan structure for its formulation.  As  Alfsen and Shultz note, their notion of dynamical correspondence ``axiomatizes the transition $h \mapsto L_{ih}$ from the self-adjoint part of a $C^*$-algebra to the set of skew order-derivations on the algebra."   The appearance of the minus sign in the commutator when one moves between 
commutators of Jordan algebra multiplication operators (which as noted above are the selfadjoint part of $\faut{V_+}$) and the Lie bracket of generators of reversible transformations (the skew-adjoint part of $\faut{V_+}$) reflects the close relation of this transition to the existence of a complex structure
on $\faut{V_+}$.  Such a complex structure, compatible with Lie brackets
and the Cartan involution, is what Alfsen and Shultz (\cite{ASBook2}, Definition 6.8) call a \emph{Connes orientation} on a JB-algebra.  

The second condition can be rephrased as saying that the one-parameter dynamical group generated by the 
image $\psi_a$ of an observable $a$ conserves that observable, so it can easily be reformulated, in our setting, in a way that refers only to convex, and not to Jordan, structure, using the abovementioned identification of the Jordan automorphisms with the order-unit-preserving affine order-automorphisms of the cone over the effects.  

 In finite dimension, Alfsen and Shultz' results (Theorem 6.15 or Proposition 10.27 of \cite{ASBook2}) combined with the
main result of the present paper imply that 
that the conjunction of spectrality, strong symmetry, and the existence of a dynamical
correspondence, characterizes complex quantum theory and classical theory, that is, the normalized state spaces
of Jordan algebras corresponding to the Hermitian parts of full matrix algebras over $\C$, and simplices.  

\begin{proposition}
Let $\Omega$ be a finite dimensional convex compact set satisfying (1) spectrality, (2) strong symmetry, and (3) dynamical correspondence, or equivalently the existence of a Connes orientation.  Then $\Omega$ is affinely isomorphic to 
the normalized state space of the Jordan algebra $Herm(n,\C)$, i.e. the set of density matrices of a finite-dimensional quantum system, or to a simplex, i.e. the state space of a finite-dimensional classical system. 
\end{proposition}

Another assumption is then needed to rule out classical
theory (i.e., the simplices).  Many natural alternatives are known, and among these we mention:
existence of a tradeoff between information gained about an unknown
state, and disturbance to that state (a result reported in \cite{Barrett07}); impossibility of universal
cloning, or of universal broadcasting \cite{BBLW2007, Barnum:2006},
the existence of a state having two different convex decompositions
into pure states  (a more or less folkloric mathematical fact that is the finite-dimensional case of Choquet's 
theorem); the lack of universal compatibility of measurements \cite{PlavalaCompatible};
nontriviality of the connected identity component of the automorphism
group of the normalized states (emphasized by Hardy \cite{Hardy2001a, Hardy2001b}).

The authors of \cite{BMU} narrowed down the class of Jordan algebras characterized there to the complex quantum
state spaces, using a principle they call \emph{energy observability}.  

\begin{definition}
\label{def: energy observability}
A normalized state space $\Omega$ is said to have
\emph{energy observability} (\cite{BMU}, Def. 30) if the Lie algebra
$\faut~{\Omega}$ of $\aut\,{\Omega}$ is nontrivial and there exists an
injective linear map $\phi$ from $\faut~{\Omega}$ to the observable
space $V^*$ of the system, such that for each $x \in \faut~{\Omega}$,
$\phi(x)$ is conserved by the one-parameter subgroup generated by $x$,
and $\phi(x)= \lambda u$ (for some $\lambda \in \R$) if and only if $x=0$. 
\end{definition}

Energy observability is closely related to dynamical correspondence, but it is formulated in the convex framework without reference to Jordan structure, incorporates nontriviality of the connected automorphism group of $\Omega$, and differs from the existence of a dynamical correspondence in other ways.
 The terminology is
motivated by the idea that a continuous one-parameter subgroup of
automorphisms is a potential dynamical time-evolution, and in quantum physics
the generator of such an evolution is a Hermitian operator $H$ (the
Hamiltonian) conserved by the evolution (identified with energy).  Here ``generator'' is meant
in the ``physicists'' sense that the evolution operator is $\omega
\mapsto e^{iHt}$ (where $i = \sqrt{-1}$).\footnote{In the usual mathematical
terminology, the generator of this evolution is instead the anti-Hermitian 
operator $iH$; then the injection from
the Lie algebra of generators of one-parameter subgroups of
automorphisms (i.e. $\flie(\aut{(\Omega)}) \equiv \faut~{\Omega})$ into
the observables is just $X \mapsto -iX$.}  The assumption that 
$\faut~{\Omega}$ is nontrivial is there because without it there is
nothing that fits the intuitive notion of energy that inspired the
definition.  (And if we were to formally extend the above definition
to that situation, energy would, logically speaking, be observable
because anything is true of the empty set.)  However, it should also
be noted that this nontriviality assumption is doing the work of ruling
out classical theory.  

\begin{proposition}
Let $\Omega$ be a finite dimensional convex compact set satisfying (1) spectrality, (2) strong symmetry, and 
(3) energy observability.  Then $\Omega$ is affinely isomorphic to 
the normalized state space of the Jordan algebra $Herm(n,\C)$, i.e. the set of density matrices of a finite-dimensional quantum system. 
\end{proposition}

The relation of energy observability to dynamical correspondence is,
roughly, that a dynamical correspondence is a linear map (but not necessarily an injection\footnote{In \cite{BMU} (on p. 29)  dynamical correspondences are mistakenly described as injections; they are injective in the simple case, which is the
one under consideration there, but not in general.}) 
of observables into the Lie algebra, i.e. in the opposite direction from the map required by energy observability, and also that dynamical 
correspondence imposes some
conditions of compatibility with the Jordan structure, while the
notion of energy observability uses only the convex structure of the state 
space and does not need Jordan structure for its
definition.  The conservation conditions are essentially the same for the two notions (modulo the reversal of direction of the map).  The possibility of noninjectivity of dynamical correspondences is needed in order to allow some 
non-simple Jordan algebras to have dynamical correspondences: for example, a finite product of one-dimensional Jordan algebras---which corresponds to a finite-dimensional classical system, has trivial (zero-dimensional) $\faut ~\Omega$, 
but may still have a dynamical correspondence, because all observables can map to the unique element $0$ of $\faut ~\Omega$; more generally, for a product of nontrivial Jordan algebras, observables that are linear combinations of the
Jordan units of the simple factors will map to zero.     

We have already mentioned the close relation of the existence
of a dynamical correspondence to  Connes' condition of the existence of an \emph{orientation}.  In 
\cite{ConnesSelfDualCones} Connes defined
an orientation in the setting of self-dual cones $V_+$ in (not necessarily finite-dimensional) Hilbert spaces.  The Lie algebra of the automorphism group of such a cone is involutive, and an orientation on such a cone was defined as a complex structure on the Lie algebra of $\faut~V_+$ compatible with this involution.     
Connes showed that the existence of an orientation characterizes the positive cones of von Neumann algebras within
the class of self-dual, facially homogeneous cones in Hilbert spaces (\cite{ConnesSelfDualCones}, Th\'eor\`eme 5.2).
Since Bellissard and Iochum \cite{BellissardIochumHSDCJA} showed that these are precisely the positive cones of JBW algebras, 
this provides a way of characterizing the von Neumann algebra 
state spaces within the class of JBW algebraic ones.  
Alfsen and Shultz explicitly transferred the definition of Connes orientation to JBW-algebraic state spaces and established, in 
Theorem 6.18 of \cite{ASBook2}, a bijection between Connes orientations in the sense 
of their Definition 6.8 (\cite{ASBook2}) and dynamical correspondences in the JBW-algebraic case.  

\subsubsection{From Jordan state spaces to complex quantum theory via local tomography}
 
A different approach to ruling out the Jordan algebraic systems other than complex quantum theory
involves introducing an appropriate notion of composite system consisting of two or more ``subsystems".     
The existence of ``tomographically local" Jordan-algebraic composites of Jordan-algebraic systems can then be used as a postulate to narrow things down to complex quantum systems.  Tomographic locality can be mathematically formulated as the requirement that the ambient vector space $V_{AB}$ spanned by the cone of unnormalized states of a composite of systems $A,B$ whose ambient vector spaces are $V_A$ and $V_B$, be the real tensor product $V_A \otimes V_B$.\footnote{The notion of ``tomography" in this context is that of determining the state of a system by making various measurements on identically prepared copies of a  system.  ``Local" tomography of a composite system is possible if one can estimate the state by making measurements on its parts, $A$ and $B$, and estimating the correlations between sufficiently many measurement results.  If $V_{AB} = V_A \otimes V_B$, then products of effects $e_A \otimes f_B$ span the dual of the state space, so determining the probabilities of a spanning set allows one to determine the components of the state in a basis.}   In \cite{BarnumWilceLocalTomography} it was shown that the existence of a locally tomographic Jordan-algebraic composite of a Jordan-algebraic system $A$ with a qubit (the lowest-dimensional nontrivial complex quantum system, whose state space is a three-dimensional ball), satisfying some other natural desiderata, implies that $A$ must be complex quantum.  However, this 	does not rule out the possibility of 
theories whose systems are spectral and strongly symmetric but in which qubits do not occur as a system type that must be composable with other systems.  

In \cite{MasanesAxiom}, Ll. Masanes and M. M\"uller formulated five postulates applicable to theories whose systems are described in the GPT framework, and showed that the only two theories satisfying them are finite-dimensional complex quantum theory and finite-dimensional classical theory.  One of these postulates is that composite systems are locally tomographic.  Their notion of theory is somewhat implicit in their arguments, rather than fully explicit, but it appears to require that for any two systems of the theory, there exists another system of the theory that is a locally tomographic composite of those systems.   Since the four postulates not referring to composite systems are satisfied by all systems with simple Jordan algebraic state spaces, and by systems whose state spaces are simplices, we can combine their result with the main result of this paper to conclude that any collection of (finite-dimensional) systems satisfying strong symmetry and spectrality, and closed under the formation of composites, must consist either entirely of complex quantum systems, or entirely of classical systems.  

That the tomographic locality of the assumed composites is necessary for these results is indicated, for example, by the constructions in \cite{BarnumGraydonWilceCCEJA} of theories in which some of the Jordan algebraic systems other than complex quantum ones can be combined to form composites that are not tomographically local; these theories even have the additional structure of dagger compact closed categories (the terminology of  \cite{Selinger} for a notion earlier defined in \cite{Abramsky-Coecke}).  In addition to the category of irreducible complex quantum systems, there is the category of real quantum systems, and
another category that includes all irreducible real and quaternionic quantum systems.  A third additional category allows the inclusion of real, quaternionic, and complex systems in a single category, at the price of a notion of composite that does not preserve irreducibility.  Not all Jordan algebraic systems occur in these constructions, though: the spin factors whose state spaces are $d$-balls for $d  \notin \{1,2,3,5\}$ do not 
occur in these categories, nor does the exceptional Jordan-algebraic system.\footnote{The $d$-balls with $d \in \{1,2,3,5\}$ are the classical, 
real quantum, complex quantum, and quaternionic quantum bits, respectively.}  The lattter, indeed, does 
not have Jordan-algebraic composites, tomographically local or not, with nonclassical Jordan-algebraic systems (\cite{BarnumGraydonWilceCCEJA}, Corollary 4.10), and Example 6.2 in \cite{BarnumGraydonWilceCCEJA} 
suggests that there may be obstructions to Jordan-algebraic composites involving the spin factors whose state
spaces are the $d$-balls for $d = 4$ and $d \ge 6$ as well.  (References  to 
\cite{BarnumGraydonWilceCCEJA} are to arXiv version v2.)

A similar argument involving tomographic locality can be made using the result of \cite{MasanesEtAlEntanglementAndBlochBall}, 
in which it is shown that only for $d=3$ does there exist a composite, satisfying tomographic locality and
continuous reversible transitivity on pure states, of two systems each of which has a Euclidean $d$-ball as state space.  
Since by 
Theorem \ref{theorem: main} (cf. Theorem \ref{theorem: sss bits are balls}), spectrality and strong symmetry imply that bits are balls, and also that nonclassical systems are simple Jordan-algebraic and hence have continuous reversible transitivity on pure states, it follows from Theorem \ref{theorem: main} and \cite{MasanesEtAlEntanglementAndBlochBall} that no tomographically local composite of a bit with itself can preserve spectrality and strong symmetry, except in the complex quantum case (for which bits are $3$-balls) or the classical case. 

\subsection{Relations with other characterizations of Jordan algebraic classes of state spaces}

Jordan algebraic systems share with standard (i.e. complex) quantum
theory many geometric properties of informational and physical
significance in addition to spectrality and strong symmetry.  These
include the absence of higher-order interference \cite{Niestegge,
  UdudecBarnumEmersonUnpublished}; purity-preserving projectivity
\cite{ASPaper, ASBook2, BMU}; homogeneity and self-duality.
Obviously, Theorem \ref{theorem: main} implies that systems with spectrality and strong
symmetry have these other properties too.  This raises the question of
whether the theorem could be proved differently, by establishing a
direct implication from spectrality and strong symmetry to combinations of these other
properties sufficient to establish the Jordan algebraic nature of the
state space.  For example, as noted in Proposition \ref{prop: consequences} it is known \cite{BMU}
that spectrality and strong symmetry imply self-duality of the cone
over the normalized state space.   In light of the Koecher-Vinberg theorem \cite{Koecher, Vinberg}, 
which states that the homogeneous self-dual cones are precisely the Jordan-algebraic ones, 
one could therefore try to show
directly that spectrality and strong symmetry imply homogeneity.  Similarly, it
is known \cite{BMU} that the conjunction of spectrality and strong
symmetry implies projectivity of the state space in the sense that
each face is the positive part of the image of a projection that is
the dual of (and hence, given self-duality of the cone, identical to) what Alfsen
and Shultz call a compression.  Since self-duality near-trivially
implies Alfsen and Shultz' postulate of \emph{symmetry of transition
  probabilities} (STP), and the finite-dimensional case of their theorem
characterizing the state spaces of a wide class of Euclidean Jordan
algebras has projectivity, symmetry of transition probabilities, and
one of several other properties (equivalent in the context of projectivity and STP) as premises, we 
would just need to establish a direct implication from spectrality and strong symmetry to one of these other
properties to obtain a direct proof of the Jordan structure.  The most promising choice is perhaps the preservation of
purity by the duals of compressions: that the image of a pure state under a
compression is always a multiple of a pure state.  
In fact, the result in \cite{BMU} was obtained by showing that the duals of
compressions preserve purity---but the additional assumption of
absence of higher-order interference was used in showing this.  

Alternative properties
capable of characterizing the Jordan-algebraic state spaces among those satisfying Alfsen and Shultz's
conditions of projectivity and symmetry of transition probabilities are (1) the 
Hilbert ball property, or (2) the satisfaction 
of the \emph{atomic covering law} by the lattice of faces of the state space.  
A finite-dimensional convex set has Alfsen and Shultz' \emph{Hilbert ball property} if and only if for every
pair of extreme points of $\Omega$, the face they generate is affinely isomorphic to a Euclidean ball.  (See \cite{ASBook2}, Def. 9.9, for additional technical conditions relevant in infinite dimension.) The atomic covering
law for a lower-bounded lattice states that if $a$ is an atom in the lattice, and $b$ any element of the lattice, then either
$a \join b = b$, or $a \join b$ covers $b$.  Here ``$x$ \emph{covers} $y$" means $x > y$ and there exists no $w$
such that $x > w  > y$, i.e. $x$ is above $y$ and there is nothing between them, and an \emph{atom} 
is an element that covers $0$.   So an alternative proof of
our result could also aim at establishing either one of these properties directly.  By the main result of 
\cite{BMU} a direct proof of the absence
of higher-order interference (see the next subsection for a rough definition) from spectrality and strong symmetry would also do the job.

\subsection{Implications for other results on general probabilistic theories having spectrality and strong symmetry}

\subsubsection{Higher order interference}
\label{subsubsection: higher order interference}

In \cite{Sorkin} Rafael Sorkin introduced a notion of a hierarchy of orders or degrees of probabilistic interference, one for each positive 
integer, for physical theories modeled in a ``histories" framework.  In \cite{3slit, UdudecBarnumEmersonUnpublished, CozThesis} Sorkin's notion was adapted to the GPT framework.  Roughly speaking, interference of order $k$ represents the idea that there are processes in which $k$ or more mutually exclusive ``paths" or ``histories" are available to the system, such that there is a measurement whose probabilities, when measured on a system that has undergone such a process, cannot be determined from the probabilities in all situations in which $k-1$ or fewer of the paths are available.  Quantum theory has interference of 
order at most $2$ (the lowest order that intuitively represents interference, since theories with maximal interference order of $1$ are classical).   If a theory fails to have interference of order $k$, then it can have no interference of 
order $k+1$ or higher.  We call interference of order $3$ or above \emph{higher-order interference}.   As mentioned in the introduction, Theorem \ref{theorem: main} improves on the main 
result of \cite{BMU}, which characterized the same class of theories, but in addition to spectrality and strong symmetry, used the additional assumption of no higher-order interference.  

\subsubsection{Query computation and query complexity}
\label{subsubsection: queries}

In \emph{query problem} in the theory of computation, it is assumed that 
one has access to a sequence, parametrized by ``instance size" 
$n$, of ``black boxes" capable
of computing a function $f_n$ between finite sets (the usual case is $f: \{0,1\}^n \rightarrow \{0,1\}$).  $f_n$ is 
not known, but a finite family $F_n$ of possible functions is specified; sometimes a prior distribution over $F_n$ is 
specified as well.   One may then ask about
the ``query complexity" of computing some function $g_n$ whose domain is the family of functions $F_n$;
roughly speaking, this is the number of times the black box appears in a circuit  capable of computing, 
with high probability of success (either in the worst case or, in case a prior distribution on $F_n$ has been 
specified, in expectation with respect to that distribution), the function $g_n$.  
The circuit is realized in some background
circuit model of computation, such as a classical or quantum circuit model, and the behavior of the black boxes may also be
classical, quantum, or more general, giving rise to a variety of possible query models of computation.
Typically one is concerned
with how the number of queries scales with the input size $n$.  For example, in Grover's ``search problem" of ``identifying 
a marked state", for which the instance size parameter is usually written as $N$, $F_N$, of cardinality $N$, is the set of functions $\{1,...,N\} \rightarrow \{0,1\}$ that take the value $1$ on exactly one element of its domain $\{1,...,N\}$, usually termed the ``marked state".  In Grover's problem $g_N: F_N \rightarrow \{1,...,N\}$ is the function that identifies which of these $N$ possibilities for $f_N$ is computed by the black box, i.e. which state is marked.   Grover's celebrated quantum query algorithm \cite{Grover96a, Grover97a} computes $g_N$ with a number of queries of order $\sqrt{N}$.  In the closely related query problem of computing OR, $F_N$ consists of 
all $N^2$ functions $\{1,...,N\} \rightarrow \{0,1\}$ and $g_N$ is OR of the values of $f_N$ on all its inputs, i.e. $g_N(f_N)$ is $1$ if $f_N$ is zero on all $N$ inputs, otherwise it is $1$.  A variant of Grover's algorithm computes OR, again with a number of queries of order $\sqrt{N}$.

In \cite{LeeSelbyGrover} it was shown, in a reasonable query model
generalizing the quantum query model, that five principles, of which
the fifth is strong symmetry, imply that to have probability $1/2$ or
greater of correctly identifying the marked state in Grover's
search problem, the number of queries must be at least $(3/2 -
\sqrt{2}) \sqrt{N/h}$, where $h$ is the maximal ``order of
interference'' of the GPT theory.   A lower bound of $\Omega(\sqrt{N})$ was
established in the quantum case in \cite{Bennett97b}; it is achieved by 
Grover's algorithm.  
The bound in
\cite{LeeSelbyGrover} is also $\Omega(\sqrt{N})$.   This limits the potential gain from
higher-order interference of degree $h$ to at most a constant factor, $c/\sqrt{h}$ compared to quantum.
The authors of \cite{LeeSelbyGrover} 
argued that their result is
``somewhat surprising as one might expect more interference to imply
more computational power''.  
   
However,  it can be shown 
(cf. \cite{ChiribellaScandoloEntanglementAxiomatic}) 
that the conjunction of the
principles used in \cite{LeeSelbyGrover} implies spectrality.  Together with Theorem \ref{theorem: main} this implies that the
GPT systems  considered in \cite{LeeSelbyGrover} are Jordan-algebraic, and hence that the
order $h$ of interference is at most 2
\cite{UdudecBarnumEmersonUnpublished, Niestegge}.   The 
question of whether or not higher-order interference of some fixed maximal degree can give a
non-constant asymptotic speedup (in terms of number of queries) over Grover's quantum algorithm for the Grover search problem in some GPT with a reasonable query model remains open, but our results imply that if 
such a speedup is possible it will be in a setting not 
allowed by the postulates in \cite{LeeSelbyGrover}.  

Besides strong symmetry, the assumptions in \cite{LeeSelbyGrover} are
  causality (roughly, no signaling from the future, which is implicit in the setting of this paper, and most work on GPTs, because of the uniqueness
of the order unit $u$), purification (every state on a system $A$ arises as the marginal of a pure state on a possibly larger composite system $AB$, and any two such purifications are related by an automorphism of 
$\Omega_{B}$), purity preservation under composition, and the existence of a
  pure sharp effect.   It is not explicitly stated whether one must assume purity preservation 
under both parallel
and sequential composition, but it appears that only parallel composition is used in the proofs.   Even with strong symmetry omitted, and purity preservation required only under parallel composition, the conjunction of these 
conditions is an extremely strong assumption,\footnote{Most of the strength of this conjunction lies in purification, 
purity preservation under parallel composition, and the existence of a pure sharp effect, since causality is part of essentially every definition of composite system in the GPT framework. The existence of a pure sharp effect would follow from the other conditions and the no-restriction hypothesis, which, although also quite strong in its way, is a natural and oft-made assumption in the GPT framework.} since if they are augmented with purity preservation under pure operations,\footnote{Pure operations are linear maps that lie in extremal rays of the cone of ``allowed" positive maps on the state space.} they would already imply the Jordan-algebraic nature of the state space (though not the restriction to simple Jordan algebras or simplices), as shown in \cite{BarnumLeeScandoloSelbyJordan}.  However, since to the best of our knowledge purity preservation under sequential composition of pure operations was not previously known to follow from the other assumptions, the possibility that the assumptions of \cite{LeeSelbyGrover} still allow for higher-order interference was open until it was excluded by the main result of the present paper.

In \cite{BarnumLeeSelbyOracles}, a definition of query computation  was formulated and two results were obtained concerning query computation in general probabilistic theories under nearly the same assumptions as \cite{LeeSelbyGrover}: the ubiquitous (and innocuous) causality, purification, purity preservation under parallel composition, and strong symmetry, as well as preservation of the maximally mixed state
(centroid of the state space) under parallel composition.\footnote{This last may well follow from the others.}  
Pure sharpness was not explicitly assumed, but the results of the paper involve situations in which nontrivial sets of perfectly 
distinguishable states exist, which are needed for the function queries to be possible, and the existence of such sets
can be shown to imply pure sharpness; indeed, in \cite{BarnumLeeSelbyOracles} the cone of unnormalized states 
is shown to be not only self-dual, but perfect.  The first
main result of \cite{BarnumLeeSelbyOracles} was that if $kn$ classical queries yield no information concerning a function to be computed, then in a general probabilistic model with maximal order of interference $k$, $n$ queries yield no information.  This allows one to obtain, from classical zero-information lower bounds on the number of queries needed to compute or approximate properties of a black-box function $f$, lower bounds in more general models, but it neither rules out nor definitely establishes the possibility of speedups over quantum query computation in more general GPT models.  It generalizes a result of Meyer and Pommersheim 
\cite{MeyerPommersheim}, who studied the quantum case.  The second main result of \cite{BarnumLeeSelbyOracles} was a 
confirmation that the generalization (from the classical and quantum cases) of the notion of black-box query used there is reasonable, in the sense that if there is a polynomial-size family of GPT circuits $C_f$ for a 
family of functions $f$, one can use them to simulate the black-box queries to the functions $f$, with a
 polynomial family of circuits.   Much of the interest in query algorithms, both quantum and classical, implicitly 
relies on this type of result, since they allow one to pass with at most polynomial cost from query algorithms to concrete circuit algorithms in cases where circuits for the function $f$ exist---giving rise to efficient algorithms in cases where the amount of  resources required by the query algorithm for the computation interleaved between the queries is also polynomial.       

The absence of an explicit assumption of sequential purity preservation is the only thing 
distinguishing the assumptions (excluding strong symmetry but including pure sharpness) of \cite{BarnumLeeSelbyOracles} from those of sharp theories with purification (which are known to have Jordan-algebraic state spaces). 
But without 
sequential purity preservation, it was still not clear that the systems of \cite{BarnumLeeSelbyOracles} are Jordan algebraic.  However, 
they do have spectrality so, as in the case of  \cite{LeeSelbyGrover}, Theorem \ref{theorem: main} implies that
these theories, too, have Jordan-algebraic state spaces, and cannot exhibit higher-order interference.  This should motivate
attempts to extend these results to more general settings.  

\subsubsection{Entropic and thermodynamic aspects of probabilistic theories}
\label{sec: thermo}

In \cite{KrummEtAlThermo}, spectrality and strong symmetry were used to obtain important properties of quantum entropy and entropy-like quantities in a more general context.  Given spectrality, it is natural to investigate real-valued entropy-like functions on the space of states defined using Schur-concave functions:\footnote{A function $f: \R^n \rightarrow \R$ is called  \emph{Schur-concave} if whenever $x$ majorizes $y$, $f(y) \ge f(x)$.  $x$ is said to majorize $y$, for $x,y \in \R^n$, if 
for all $m \in \{1,..,n\}$, it holds that  
$\sum_{i=1}^m x^{\downarrow}_m \ge \sum_{i=1}^m y^{\downarrow}_m$, where   $x^{\downarrow}, y^{\downarrow}$
are the vectors whose elements are those of $x$ and $y$ respectively, arranged in decreasing order.  Sometimes Schur-concave functions are considered to have as domain $\union_{n \in \N} \R^n$, in which case majorization is defined by comparing two 
vectors with the shorter one padded out with zeros to the length of the longer one.  ``$x$ majorizes $y$" is generally interpreted as a formalization of the idea that $x$ is ``more mixed" or ``more random" than $y$, because of the Birkhoff-von Neumann theorem, which states that 
$x$ majorizing $y$ is equivalent to $y$ being a convex combination of vectors obtained from $x$ by permuting its entries.  
Schur-concave functions are often viewed as real-valued ``measures of randomness" since they are precisely the real-valued functions that can never decrease under such operations.  This accounts for the terminology ``generalized entropies" and also for a relation of majorization to microcanonical thermodynamics (cf. e.g. \cite{Alberti81a}).}  for each such function $f$, one defines a corresponding generalized entropy as the value of $f$ on the spectrum of the state.The von Neumann entropy of a quantum state, which is given by the Shannon entropy $H(p) := - \sum_i p_i \ln p_i$ of its spectrum $p= \{p_1,...,p_n\}$, is one such entropy.  In general theories, one can define the \emph{measurement entropy} and the \emph{preparation entropy} of states.  The measurement entropy of state $\sigma$ is the minimum, over finegrained measurements, of the Shannon entropy of the probabilities of the outcomes when the measurement is made on a system in state $\sigma$; the preparation entropy is the minimum entropy of probabilities $p_i$ such that 
$\sigma = \sum_i p_i \omega_i$, for pure states $\omega_i$.  Analogous definitions can also be made for 
the generalized entropies determined by Schur-concave functions other than Shannon entropy.  In quantum theory, the preparation and measurement entropies corresponding to a given $f$ are equal to each other and to the spectral entropy corresponding to $f$.  In \cite{KrummEtAlThermo}, it was shown, assuming spectrality and strong symmetry,
that the outcome probabilities of any fine-grained measurement on $\sigma$ are majorized
by those of the spectral measurement (which are equal to $\sigma$'s spectrum), and hence that the measurement entropy determined by any Schur-concave function is equal to the corresponding spectral entropy.\footnote{The theorem was stated for the Renyi $\alpha$-entropies, but the proof uses Schur concavity and applies to arbitrary Schur concave functions.}  In parallel work in \cite{ChiribellaScandoloDiagonalization} the same conclusion was obtained using causality, purification, purity preservation under both parallel and sequential composition of pure operations, and strong symmetry.  The first four of these assumptions together imply spectrality.   So in light of the present paper, the setting of 
\cite{KrummEtAlThermo, ChiribellaScandoloDiagonalization} is no more general than that of simple Jordan-algebraic state spaces, and classical ones.\footnote{It should nevertheless be noted that to the best of our knowledge, the conclusions obtained in \cite{KrummEtAlThermo} and 
\cite{ChiribellaScandoloDiagonalization} were not previously known for the non-quantum, non-classical simple Jordan algebras.}   However, the same conclusions can also be obtained from different postulates, including or implying spectrality, but not strong symmetry.  This was done, using different sets of assumptions, in 
\cite{ChiribellaScandoloEntanglementAxiomatic}, \cite{QPLthermo}, and \cite{ChiribellaScandoloMicrocanonical}.  In light of the present work, it becomes even more interesting to determine whether the assumptions used in these works imply the Jordan algebraic structure of state space (even if not the simple structure or classicality of the Jordan algebra, which is enforced by strong symmetry but which does not follow from the assumptions of \cite{QPLthermo}, \cite{ChiribellaScandoloEntanglementAxiomatic}, or \cite{ChiribellaScandoloMicrocanonical}). 
The results in \cite{ChiribellaScandoloEntanglementAxiomatic, ChiribellaScandoloMicrocanonical} 
concern a class of theories they call 
\emph{sharp theories with purification}, which satisfy the four properties of causality, purification, purity preservation (under both parallel and sequential composition), and pure sharpness.  In \cite{BarnumLeeScandoloSelbyJordan} it was shown that all systems in this class of theories are Jordan-algebraic (although this class is not precisely simple Jordan algebras and classical theory, since some nonclassical nonsimple state spaces are definitely allowed, and to the best of our knowledge it is not known whether all simple Jordan algebras are).  
However, the assumptions of \cite{QPLthermo} are just projectivity of the state space and symmetry of transition probabilities (equivalently, projectivity and self-duality of the state cone, which are in turn equivalent
(\cite{Araki80}, cf. also \cite{QPLthermo}) to its perfection together with the normalization of the orthogonal projections onto the linear spans of faces).
All Jordan algebraic state spaces have these properties, but it is an open question whether they are the only ones.  
Essentially these assumptions (in the guise of projectivity and symmetry of transition probabilities) appear in Alfsen and Shultz's derivation \cite{ASPaper, ASBook2} of Jordan-algebraic structure, but there a choice of one of several  additional assumptions, for instance that the ``filters" that project onto faces of the state space take pure states to multiples of pure states or one of the other alternatives discussed above, is used.  
It is not known whether or not this assumption can be dropped in the characterization of Jordan-algebraic state spaces.\footnote{In \cite{Araki80} H. Araki gave a derivation of Jordan-algebraic structure in finite dimension inspired by Alfsen and Shultz's, but using a notion of projection on the state space he calls ``filter", which is \emph{prima facie} weaker than the notion of dual of a compression;  he also assumed filters preserve purity, but conjectured the assumption could be dropped.  Alfsen and Shultz \cite{ASBook2}, p. 354, state that ``in the finite-dimensional context his axioms force the filters to be compressions", but it is not clear whether this is meant to apply to the axioms without the purity-preservation assumption.}  So, the results of \cite{QPLthermo} add additional interest to the question of whether there are non-Jordan-algebraic state spaces satisfying projectivity and self-duality, while the results of the present paper show that such state spaces will not be found among those satisfying strong symmetry.  

\section{Conclusion}
\label{sec: conclusion}

We have shown that the finite-dimensional compact convex sets satisfying two properties, spectrality and strong symmetry, 
are up to affine isomorphism precisely the sets of normalized states of simple finite-dimensional Euclidean Jordan algebras, and simplices, answering a question posed in \cite{BMU} of whether the additional property of no higher-order interference used
there in addition to spectrality and strong symmetry, is needed to characterize this set of state spaces.  While this can be viewed purely as a result in convex geometry, it has important implications for the research program that studies general probabilistic theories, since significant results in that program were obtained under the assumption that the normalized state spaces of systems are spectral and
strongly symmetric, or under assumptions that imply this; we described some of these implications in the preceding section.  
We also discussed the relation of our result to other characterizations of Jordan algebraic state spaces.  

These Jordan-algebraic compact convex sets, and the cones over them, figure in many other areas of mathematics and applications, including complex analysis, symmetric spaces, optimization, and statistics, and the present result may have interesting implications in some of these areas as well.  From the both the perspective of general probabilistic theories and that of the purely mathematical theory of convex sets, our results suggest exploring the consequences of assumptions weaker than spectrality and strong symmetry.

\section*{Acknowledgments}
HB thanks the Department of Mathematical Sciences and the QMATH group at the University of Copenhagen 
for hosting him as Visiting Professor while investigations preliminary to the present work were undertaken, and the Villum Foundation for making his visit possible through its support of QMATH.   
Both authors thank Henrik Schlichtkrull for putting them in contact with each other.

\begin{appendix}
\section{Proof of Theorem \ref{theorem: sss bits are balls}}
\label{sec: proof that sss bits are balls}

In this section we establish Theorem \ref{theorem: sss bits are balls},
following \cite{DakicBruknerQuantumBeyond} but bringing in the
stronger assumption of $2$-transitivity.

\begingroup
\def\thetheorem{\ref{theorem: sss bits are balls}}
\begin{theorem}
Let $\Omega$ be a strongly symmetric spectral compact convex body whose
largest frame is of cardinality $2$.  Then $\Omega$ is affinely
isomorphic to a ball.
\end{theorem}
\addtocounter{theorem}{-1}
\endgroup

\begin{proof}
We view $\Omega$ as canonically embedded in the affine space $\aff
\Omega \iso E$ that it generates, of dimension $n$, equipped with a
Euclidean inner product for which $\aut{\Omega}$ is a subgroup of
$O(E)$, and recall that because
$\aut{\Omega}$ acts transitively on
$\partial_e \Omega$, we have that $\int_{\aut{\Omega}} d\mu(k) ~ k.\omega$ is
a fixed point of the group action, and is in fact the barycenter, $0$. 
Because it is an orbit of a subgroup of $O(V)$, the set $\partial_e \Omega$ of extremal
points of $\Omega$ is contained in the sphere $S := \{x \in V:
||x||=c\}$ with respect to the associated norm.  We scale the inner
product by a positive real number so that $c=1$.

We now prove that every pair of perfectly
distinguishable points in $\Omega$ are the endpoints of some diameter
of $S$.  
We begin by showing
(following Daki{\'c} and Brukner but with a bit more detail) that for
every extremal $\omega \in \Omega$, the point $-\omega$ also belongs
to $\Omega$, and $[\omega, -\omega]$ is a $2$-frame.

A \emph{chord} of a sphere is defined to be a closed line segment
whose endpoints are two distinct points on the sphere.  We will use
the fact that the only chords of a sphere that contain its center are
the diameters, i.e. the chords from $x$ to $-x$.  Since $\Omega$ is
spectral with maximal frame size $2$, every nonextremal point in
$\Omega$, in particular its center, $0$, is a convex combination of
two perfectly distinguishable extremal points of $\Omega$.  Let
$\omega_0$ and $\omega_1$ be extremal points of $\Omega$ such that $0$
is a convex combination of them.  Since we showed above that all
extremal points of $\Omega$ lie on the sphere $S$, the set of convex
combinations of $\omega_0$ and $\omega_1$ is a chord of $S$ containing
its center, $0$.  Therefore it is a diameter, and $\omega_1 = -
\omega_0$.  Since $\Omega$ has reversible transitivity on pure states 
(i.e. transitivity of $\aut{\Omega}$ on $1$-frames, which are
precisely the extremal points), every extremal point $\omega$ of
$\Omega$ can be obtained from $\omega_0$ by acting with an element of
$O(V)$, whence by linearity of the action, $-\omega$ is also in
$\Omega$.  So we have established that $\Omega$ is symmetric under
coordinate inversion $x \mapsto -x$, and that every pair $\omega,
-\omega$ is a maximal frame.

We still need to show that there are no other maximal frames in $\Omega$, 
i.e. no 2-frames that are not the endpoints of a diameter.\footnote{This is the 
main point at which we perceive a gap in the argument in  
\cite{DakicBruknerQuantumBeyond} which we do not see how to easily bridge using
only reversible transitivity and strong symmetry.} If we have 
transitivity on $2$-frames, we get this immediately: every $2$-frame is an
automorphic image of $(\omega_0, -\omega_0)$, and therefore of the form
$(\omega, -\omega)$ for some extremal $\omega$.  

Once we have this fact, we can use it as in
\cite{DakicBruknerQuantumBeyond} to establish that $\Omega$ is a ball.
The centroid of a compact convex set, which is $0$ in the case of $\Omega$, 
 is in its relative interior.  $\Omega$ is full-dimensional, so its relative interior
is its interior, and there is an open ball around $0$ contained in $\Omega$.  
So for any $x \in S$ there is $\lambda \in (0, 1]$ small enough that
$\lambda x \in \Omega$.  By spectrality, $\lambda x$ is a convex
combination of two perfectly distinguishable extremal points of
$\Omega$.  Since we showed above, using $2$-transitivity, that all
such pairs are endpoints of diameters, $\lambda x$ must be a convex
combination of the endpoints of a diameter.  For $x \in S$ the only
diameter containing $\lambda x \ne 0$ is the one between $x$ and $-x$.
So we have shown that $x \in \Omega$; but $x$ was an arbitrary element
of $S$.  Since the entire sphere $S$ belongs to the extreme boundary
of the convex set $\Omega$, and we earlier showed that all extremal
points of $\Omega$ are in $S$, $\Omega$ is the convex hull of the
$(n-1)$-sphere $S_{n-1}$, i.e. an $n$-dimensional ball.
\end{proof}

\section{Faces of strongly symmetric spectral sets are strongly symmetric and
spectral}
\label{sec: faces of sss sets are sss}
In this section we prove the following theorem, which includes item 
\ref{item: barycenters of faces} of Proposition \ref{prop: consequences}.

\begin{theorem}
Let $\Omega$ be a strongly symmetric spectral compact convex set.
Then every face of $\Omega$ is a strongly symmetric spectral compact
convex set; moreover if $F$ is a face of $\Omega$ and $K =
\aut{\Omega}$, then $K(F) := K_F/K^F = \aut{F}$.  Here $K_F$ is the
subgroup that takes $F$ to itself; $K^F$ is the subgroup that fixes
$F$ pointwise.  Also, $K_F := K^{c(F)}$, where $c(F) =
\sum_{i=1}^{|F|} \omega_i/|F|$, for any frame $\omega_i$ for $F$, 
is the centroid of $F$.
\end{theorem}

\begin{proof} 

By Proposition \ref{prop: barycenter from group averaging} and the
fact that $u = \sum_{i = 1}^r \omega_i$ in a strongly symmetric
compact convex set, where $\omega_i$ are a maximal frame, and the fact
that $u$ is $\aut{\Omega}$-invariant in our setup, the barycenter of
$\Omega$ is easily shown to be $(\sum_{i=1}^r \omega_i)/r$, for any
maximal frame $[\omega_1,...,\omega_r]$.  It is easily shown that any
face of $F$ is spectral, since spectrality of $\Omega$ asserts, for
$\omega \in F$, that $\omega$ is a convex combination of perfectly
distinguishable states, but these states must be in $F$ by the
definition of face, and by Proposition \ref{prop: consequences} they
must be extendable to a frame for $F$.  Then one shows, using strong
symmetry, that for each face $F$ of $\Omega$, there is a subgroup of
$K:=\aut{\Omega}$, that preserves $\lin F$ and (necessarily or else it
could not consist of automorphisms of $\Omega$) induces automorphisms
of $F$, and acts transitively on the maximal frames in $F$.  This is
immediate from strong symmetry, since the maximal frames in $F$ are
$|F|$-frames in $\Omega$, and strong symmetry says $K$ can take
\emph{any} $|F|$-frame (whether in $F$ or not) to any other.  One has
to show that frames in $F$ are still frames for $F$ viewed in its
affine span, but that is so because the cone is perfect, and when the
dual cone is represented internally via the self-dualizing inner
product, the distinguishing effects are the states themselves
(cf. item \ref{item: frames orthogonal} of Proposition \ref{prop:
  consequences}).  Perfection of the cone $V_+$ implies that the cone
over $F$ is self-dual in its linear span according to the restriction
of the inner product, so these effects are still in the relative dual
cone of $F$.

In fact, an element of $K$ takes maximal
frames of $F$ to maximal frames of $F$ if, and only if, it belongs to
the subgroup $K_{\lin{F}}$, that preserves $\lin{F}$.  The action of
this group on $\lin{F}$ gives a faithful representation of the group
$K_{\lin F}/K^{\lin F}$ (equivalently $K_F/K^F$).  
Since we have shown that $F$ is spectral and strongly symmetric, it
follows from claim in the first sentence of this proof that the
barycenter of $F$ is $\sum_{i=1}^{|F|} \omega_i/|F|$ for any maximal
frame $\omega_i$ for $F$.  Any automorphism of $F$ preserves its
barycenter, so the automorphism of $F$ induced by any element of $K_F$
must do so, i.e. $K_F \subseteq K^{c(F)}$.  Furthermore, $K^{c(F)}
\subseteq K_F$.  To see this, note that $c(F)$ is in the relative interior of $F$ and
therefore $F = \face{(c(F))}$.  Hence for $\phi \in K^{c(F)}$ we have
$\phi(F) = \phi(\face{(c(f))}) = \face{\phi(c(F))} = \face{(c(F))} =
F$, i.e. $\phi \in K_F$.  Here the second equality is a general fact
about automorphisms $\phi$ (that $\face(\phi(x)) = \phi(\face(x))$),
and the third is from the assumption $\phi \in K^{c(F)}$.
\end{proof}

\section{More detail on the classification of regular polytopes via symmetric space representations}
\label{appendix: more information on classification}

\subsection{Details of the symmetric space representations containing regular convex bodies as orbits of Farran-Robertson polytopes}
In this subsection we go line by line through those entries in Tables \ref{table: classical noncompact ss reps}, \ref{table: exceptional noncompact ss reps} and
\ref{table: ss reps from complex groups} that describe symmetric spaces associated with
simplicial Farran-Robertson polytopes,  along the way 
verifying that all cases where the simplex has 3 or more vertices ($\tri_n, n \ge 2$), are ones associated
with EJAs, and for the $\tri_1$ cases (where the convex body must be a ball), noting which ones are
spin factor Jordan algebra automorphism representations 
(i.e. $\flie(K) = \flie(\aut V) \equiv \fder V$),\footnote{Here $V$ is a Jordan algebra, and $\aut V$ is the 
compact group of \emph{Jordan algebra} automorphisms, with Lie algebra $\fder V$.} and 
which involve other transitive actions on balls.

 The nomenclature for groups, symmetric spaces, and root systems is that used in
\cite{MaddenRobertson}, which is very close to that used in 
 Helgason \cite{HelgasonBook}, Ch. X.\footnote{Helgason combines 
$BI$ and DI in the class $BDI$, whereas \cite{MaddenRobertson} keep them separate.}   Tables 2 and 3
 in \cite{MaddenRobertson}, reproduced in the first columns of Tables \ref{table: classical noncompact ss reps} 
and \ref{table: exceptional noncompact ss reps} above,  include those symmetric spaces from Table V
 in Chapter X of Helgason \cite{HelgasonBook}, Table 2 being the
 series, Table 3 the exceptional cases.  These are the Type I and Type
 III noncompact symmetric spaces, in Cartan's nomenclature.  The list
 of coincidences between different classes of symmetric spaces given in
 $\S$6.4 of Ch. X of \cite{HelgasonBook} (pp. 519-520) is helpful here
 too; we refer to these below by Helgason's enumeration, e.g. as
 ``coincidence (i)", etc.

\vspace{10pt}
\noindent
{\bfseries From Table \ref{table: classical noncompact ss reps}: classical symmetric space representations.}
\vspace{6pt}

AI is the symmetric space
$SL(n,\R)/SO(n)$ with root space $A_{n-1}$, whose
polytope is the $n$-vertex simplex, $\triangle_{n-1}$.  
$\flie(K) = \fso(n)$, and the associated
polar representation $V_0 = \fp_0$, from the Cartan decomposition
$\fsl(n,\R) = \fso(n) \oplus \fp_0$, is the traceless real symmetric
matrices.  As discussed in Section \ref{sec: Euclidean Jordan
  algebras}, this is the polar representation $V_0$ associated with
the Jordan algebra $V$ consisting of the real symmetric matrices $V_0
\oplus \R I$ (where $I$ is the identity matrix).  The regular convex
body is then isomorphic to the unit-trace real symmetric positive
semidefinite matrices. 


AII is $SU^*(2n)/Sp(n)$, of rank $n-1$,
and dimension $(n-1)(2n+1)$, with polytope $\tri_{n-1}$.  $SU^*(2n)$
is also known as $SL(n,\H)$, and $Sp(n)$ as $SU(n,\H)$.  The Cartan
decomposition is $\fsl(n,\H) = \fsu(n,\H) \oplus \fp_0$ where $\fp_0$
is the traceless quaternionic-Hermitian matrices, so the associated
polar representation has $SU(n,\H)\equiv Sp(n)$ acting on the
traceless quaternionic-Hermitian matrices.  The regular convex compact
body of the Madden-Robertson construction is therefore affinely
isomorphic to the unit-trace positive semidefinite
quaternionic-Hermitian matrices.

Type AIII, $SU(p,q)/S(U_p \times U_q)$, with root
system $C_q$ (for $p=q$) or $BC_q$ (for $p > q$).  Two polytopes are
listed here because when $q >1$ the two end nodes of the Dynkin diagrams of $C_q$ are 
nonequivalent,
as are those of $BC_q$,  giving rise to two orbit types in each case.  At
any rate, only in the case $q=1$ are these simplical polytopes,
$\tri_1$.  The symmetric spaces are then $SU(p, 1)/S(U_p \times U_1)$, the 
dimension is $2p$, and we know that the associated orbitopes are balls $B_{2p}$ 
by Theorem \ref{theorem: sss bits are balls}.
For $p>1$, the root system is of type $BC_1$, while for
 $p=1$ it is of type $C_1$.  

The $p=1, q=1$ case corresponds to the spin 
factor $\R^2 \oplus \R$, with positive cone a 
 Lorentz cone with 2 space dimensions, whose base is the disc, by
 coincidence (i) of Helgason.  (Note that $S(U_1 \times U_1) \iso U(1)$, 
with Lie algebra $\fso(2)$.) 

In the $p>1$ case, $S(U_p \times U_1) \iso U(p)$\footnote{$S(U_p \times U_q)$
is defined by the (complex) linear representation on $\C^p \oplus \C^q$ consisting of block-diagonal matrices
$M$ with blocks in $U(p)$ and  $U(q)$ with $\det M = 1$, so $S(U_p \times U_1)$ is isomorphic to the matrices
of the linear representation $V \oplus \text{det}^*$ of $U(p)$, where $V$ is the defining representation and $\det^*$ the
dual of the one-dimensional determinant representation.  This is a faithful representation of $U(p)$.}    
acts transitively on $B_{2p}$.  In, for instance, \cite{MasanesEtAlEntanglementAndBlochBall} the
$U(p)$ and $SU(p)$ representations with transitive action on $B_{2p}$ are described; for 
$Q \in U(p)$ or $SU(p)$, we have the $2p \times 2p$ real matrix
\beq
\rho(Q) = \left(
\begin{matrix}
\text{Re } Q & \text{Im } Q \\
- \text{Im } Q & \text{Re } Q
\end{matrix}
\right).
\eeq

For Type BI the simplicial cases are  the $q=1$ cases of $SO(p,q)/(SO(p) \times SO(q))$ for
$p+q$ odd, $q <p$, of type $BI$, with root space $B_q$.  So we have
$p>1$ even, and the symmetric space is $SO(p,1)/(SO(p)\times SO(1))
\equiv SO(p,1)/SO(p)$.  The Cartan decomposition is $\fso(p,1) =
\fso(p) \oplus \R^p$, and we have $\fso(p)$ acting irreducibly
on $\R^p$.  The
action is equivalent to the Lie algebra representation derived from
the defining representation of $SO(p)$.  This is the polar
representation embedded in the spin factor Jordan algebra $\R^p \oplus
\R$, where the $\R$ summand is spanned by the Jordan identity $e$,
which is fixed by the full $SO(n)$ representation as explained in
Section \ref{sec: Euclidean Jordan algebras}.  Its positive cone is a
Lorentz cone in one ``time'' dimension, and even ``space'' dimension
$p$, with normalized state space $\Omega$ a $p$-ball.  If we write
$(x, t)$ for a an element of $\R^p \oplus \R$, the convex hull of 
the orbit in the polar representation is that ball, translated to the plane
$t=0$.

Type DI with the same group quotient, but $p+q$ even,
similarly gives the Lorentz cones for odd ``space'' dimension in the $q=1$
cases, again as embedded in spin factors.  

The simplicial cases $\tri_1$ in these two lines (BI and DI) account for all spin factor isotropy group representations.  
A few of these will reappear below, due to coincidences noted in \cite{HelgasonBook}.


 DIII has $q = \lfloor n/2 \rfloor$, which gives $q=1$, and hence $\tri_1$,  only for $n=2$ and $n=3$.
 For $n=2$, the dimension
 $n(n-1)$ is $2$, and the symmetric space is $SO^*(4)/U(2)$.  Coincidence (xi) in
 Helgason's list indicates that this coincides with type AI $(n=1)$,
 i.e real symmetric $2 \times 2$ matrices $Herm(2,\R) \iso \R^2 \oplus \R$, aka the ``rebit'', 
also appearing as AIII $(p=q=1)$, DI $(p=2, q=1)$ (Helgason's BDI$ (p=2, q=1)$), and CI $(n=1)$. 

For $n=3, q = 1$ we have $d=6$. 
By coincidence (vii) in Helgason, this is also AIII $(p=3, q=1)$. $\Omega$ is 
 the $6$-ball.  The DIII symmetric space is given by $SO^*(6)/U(3)$, while $p=3,q=1$
 implies the AIII symmetric space is $SU(3,1)/(S(U(3) \times U(1))$.
The isomorphisms cited in the Helgason coincidence (vii) are 
 $\fsu(3,1) \iso \fso^*(6)$ and 
 $\fsu(4)
 \iso \fso(6)$, 
 corresponding to the ``numerators" in the noncompact and compact forms of the symmetric space
respectively.  As for the denominators, $S(U_3 \times U_1) \iso U(3)$.
 A transitive  action of $U(3)$ on the $5$-sphere
 (boundary of the $6$-ball) is known (cf. Table I in \cite{MasanesEtAlEntanglementAndBlochBall}).

 Type CI, for $n=q=1$, gives $Sp(1, \R)/U(1)$, with
 dimension 2, rank 1.  This is again the $2$-ball (cf. coincidence (i)
 in Helgason again).  This is the only simplicial case.

 For type CIII, the simplicial cases have $q=1$, with polytope $\tri_1$ and
regular convex body a ball $B_{4p}$.   The symmetric space is 
 $Sp(p,1)/(Sp(p)\times Sp(1))$.  Since
 $\fsp(1) \iso \fsu(2)$, and $Sp(p) \times SU(2)$ acts transitively on
 the boundary of the $4p$-ball
 (cf. \cite{MasanesEtAlEntanglementAndBlochBall} once again) 
In dimension 4, i.e. $p=q=1$, Helgason's
 coincidence (iii) implies that we have an EJA $\R^4 \oplus \R$:  we
 have $\fsp(1) \times \fsp(1) \iso \fso(4)$.

\vspace{10pt}
\noindent
{\bfseries From Table \ref{table: exceptional noncompact ss reps}: exceptional symmetric space representations.} 
\vspace{6pt}

Table \ref{table: exceptional noncompact ss reps} concerns spaces derived from exceptional Lie groups.  Of
these, only two have simplicial polytopes.

EIV, with polytope $\tri_2$, has symmetric space determined by the pair of
Lie algebras $(\fg, \fk) = (\fe_{6(-26)}, \ff_4)$.  The dimension of the polar
representation of $\fk$, hence of $\aff{\Omega}$, is $26$, and the Cartan
decomposition and the Lie algebra $\ff_4$ is that of the automorphism
group of the convex set of unit-trace $3 \times 3$ positive definite
octonionic-Hermitian matrices, which is indeed $26$-dimensional,
corresponding to the $27$-dimensional exceptional Jordan algebra (cf. Table \ref{table: EJAs});
$\fp_0$ in the Cartan decomposition is the traceless
octonionic-Hermitian matrices.

FII has polytope $\tri_1$, and symmetric space determined by $\fg = \ff_{4(-20)}$, $\fk = \fso(9)$.  Its
dimension is $16$, so $B$ is a $16$-ball.  $Spin(9)$, a double cover of $SO(9)$, is known to act 
 transitively on the unit $15$-sphere in its ($16$-dimensional) fundamental representation (cf. \cite{MasanesEtAlEntanglementAndBlochBall} or 
\cite{EschenburgHeintzeClassification}).  Since the other exceptional case is $Herm(3,\Oct)$, 
one might be tempted to think this representation 
is the octonionic bit, $Herm(2,\Oct)$.  But $Herm(2,\Oct)$ is $10$-dimensional, with $9$-dimensional traceless part, 
so the ``octobit" is naturally identified with the spin factor $\R^9 \oplus \R$, with state space a 
$9$-ball.\footnote{See \cite{DrayManogue} for an interpretation of $Herm(2,\Oct)$ as acted on by double covers of 
$SO(9,1)$ and of $SO(9)$.}

\vspace{10pt}
\noindent
{\bfseries From Table \ref{table: ss reps from complex groups}: representations from simple complex groups}
\vspace{6pt}

In Table 4, corresponding to Helgason's Type IV symmetric spaces,
which arise from quotients of complex semisimple groups (viewed as
real groups) by their compact real forms, the only simplicial polytopes
appear in the the first line, of type and root space
$A_n$, for $n \ge 1$, and noncompact symmetric space
given by $SL(n+1, \C)/SU(n+1)$, with dimension $n(n+2)$.  
The Cartan decomposition is $\fsl(n+1,\C) = \fsu(n+1, \C) \oplus 
\fp$ where $\fp$ is the traceless complex Hermitian matrices;
the convex hull, $\Omega$, of the orbit of the highest restricted weight
state in the traceless positive semidefinite matrices
is affinely isomorphic to the unit-trace positive semidefinite matrices,
i.e. the ``density matrices'' of standard quantum theory over $\C$. 

The remaining lines of Table 4 include
 some families with polytopes $\Box_n, \Diamond_n$, but the cases $n=1$
 are not included so there are no more simplicial cases.

\subsection{Actions by polar representations other than symmetric space representations}
\label{subsec: polar representations that are not ss reps}

In Dadok's classification of polar representations of compact connected groups  $K$ 
by symmetric space representations, usually the polar representation is isomorphic to a
 symmetric space representation, but sometimes it is necessary to pass, via Theorem \ref{theorem:
   associated symmetric space representation},
to a symmetric space representation of a larger connected compact group $\tilde{K}$, acting on the same space 
and having the same orbits as the original polar representation.  The original representation is 
the restriction of the representation of $\tilde{K}$ to the subgroup $K$. 

In the representations associated with simplicial Farran-Robertson polytopes, this
only occurs in the $\tri_1$ cases, those in which the regular body $B$ is a ball.  These
must be cases in which $K$ acts transitively (and linearly) on the sphere $\partial B \equiv \partial_e B$. 
  Any representation of a compact group acting
 transitively on a sphere is polar.\footnote{This is implied at the
 end of Case I on p. 133 of \cite{Dadok1985}.  It is also fairly obvious:
the Cartan subspace is $1$-dimensional, generated by a diameter, and the sphere is
 the unique (up to dilation) orbit.}  
So all such representations of compact connected groups must occur either as symmetric space isotropy representations
in Tables \ref{table: classical noncompact ss reps}, \ref{table: exceptional noncompact ss reps} and
\ref{table: ss reps from complex groups}, or as the restrictions of such representations to subgroups
via Theorem \ref{theorem: associated symmetric space representation} of Dadok (for a more compact and slightly more
detailed presentation, see the main theorem of \cite{EschenburgHeintzeClassification} and the remarks following it).

In both cases, the subgroup necessarily acts transitively not only on the set of extremal points ($1$-frames), 
but on the set of $2$-frames as well, since
the $2$-frames in a ball are just pairs $[x, -x]$.    One motivation for interest in these possibilities
comes from general probabilistic theories: we might be interested in theories in which the state space is 
strongly symmetric and spectral, but the group of \emph{allowed} reversible transformations is a proper
 subgroup of $\aut_0 \Omega$, yet we still want it to act transitively on frames.  In \cite{BMU}, the condition of 
energy observability (formulated for this case to require an injection of the Lie algebra of the allowed subgroup into the 
observables with the properties described in \cite{BMU}, Def. 30) was shown to rule out not only all Jordan-algebraic state spaces with $\aut_0 \Omega$ as the
allowed transformations, but also the systems with $\Omega$ a ball $B_d$, but with only a proper subgroup of 
$SO(d)$ as the allowed reversible transformations.

\subsection{More detail on the Jordan-algebraic cases}
For the symmetric space isotropy representations in which the compact group acts 
on the orthocomplement of the identity in a Jordan algebra, and has the Lie algebra 
of the group of Jordan automorphisms (call these \emph{Jordan isotropy representations}), 
the main text already gave a general argument why the regular convex body is affinely 
isomorphic to the normalized Jordan algebra state space.  In this subsection, we see a little
more concretely how this works,  by looking in more detail at those representations for 
which the Farran-Robertson polytopes are  
$\tri_n$, $n \ge 2$, which is to 
say the ``matrix" cases $Herm(m, \D)$.  
We need to verify
that, possibly up to an affine isomorphism, the orbit whose convex
hull is the regular convex body $B$ in the Madden-Robertson
construction is the same orbit whose convex hull gives the normalized
state space of $V$.  
For $\tri_n, n \ge 2$, if $A = \fp \oplus \R$ is an EJA then it must
be of the form $Herm(n,\D)$, $\D \in \{\R,\C,\H,\Oct\}$.  The
Coxeter/Dynkin diagram for the restricted Weyl group of $G_{ss}$ 
is that of $A_{n-1}$.  The space $\fp$ in the Cartan decomposition
$\fg = \fk \oplus \fp$ of the Lie algebra of $\aut _0^{ss}A_+$ is the traceless symmetric
matrices over $\D$.  A maximal abelian subspace $\fa$ of $\fp$ is
given by the traceless diagonal matrices in each case; its dimension
is $n-1$.  Its Weyl group $K_{\fa}/K^{\fa}$ (also written
$N_K(\fa)/Z_K(\fa)$) is $S_n$.  It is convenient to study $K$'s action
$\fp$ by extending it to an action of $\fp \oplus \R$, where $\R$ in
this case is generated by the identity matrix, so we have an action on
the space of symmetric matrices, not just the traceless ones.  
 $K$'s
action fixes the identity matrix, and the traceless matrices are an
invariant (in fact irreducible) subspace for the $K$-action.  $W \iso
S_n$ acts as the group of all permutations of the unit diagonal
matrices $e_{ii}$, for $i \in \{1,...,n\}$ (or if we prefer to
consider its action on the traceless matrices, it permutes the $n$
traceless matrices $e_{11} - I/n \equiv \diag{((n-1)/n,
  -1/n,...,-1/n)}, e_{22} - I/n$, etc..).  These matrices form the $n$
vertices of an $(n-1)$-simplex in the unit-trace (resp. traceless)
matrices.  A set of generators for this group is the reflections in
the hyperplanes normal to the traceless matrices $e_{ii} - e_{jj}$,
but these are not all necessary; a minimal set is $e_{ii} -
e_{{i+1},{i+1}}$.  Thus $e_{ii} - e_{{i+1}, {i+1}}$ are a system of
simple roots for $A_{n-1}$.  
For $A_{n-1}$, the diagram is linear with $n-1$ simple roots, and
unmarked links because adjacent roots are all at angle $2 \pi/3$.  If
we identify the unit matrices $e_{ii}$ with unit vectors $e_i$ in
$\R^n$ (the unit-trace matrices) to make index manipulation easier,
then $e_{11}$, projected into the traceless matrices (i.e.  $e_{11} -
I/n$) is $\lambda_1$, the first fundamental weight, which is the point
in the Madden-Robertson construction that corresponds to the first
node of the Coxeter diagram.  The convex hull of its $W$-orbit is a
simplex, the translation to the traceless matrices of the simplex
$\tri(e_1,...,e_n)$ in the unit-trace matrices.  This simplex in
$\fa$ is therefore Madden and Robertson's $\pi(B)$ in this polar
representation, for the case of the first (leftmost) node of the
Coxeter diagram.  The last node gives (up to possible dilation) the
polar body,\footnote{In the particular case of $A_{n-1}$, this is
  confirmed by the fact that $\lambda_{n-1}/(n-1)$ is the barycenter
  of a maximal face of the simplex $\pi(B)$, which is a negative
  multiple of a vertex of the simplex, sox the convex hull of its
  orbit gives (up to dilation) the negative of the original simplex,
  which is an isomorphic simplex and (up to dilation) the polar of the
  original simplex.} and the polar of a simplex is again a simplex;
indeed since the diagram is symmetric the regular convex body
associated with the last node will also be affinely isomorphic to the
one associated with the first.
\footnote{The other nodes in the Coxeter diagram correspond, in
general in the Madden-Robertson construction, to the barycenters of a
maximal flag of faces of this convex body (the orbitope over
$\lambda_1$), and this is confirmed in the $A_{n-1}$ case by the
explicit formula for the fundamental weights: the weight $\lambda_k$
is the projection of $\sum_{i=1}^k e_i$ into subspace $\sum_i x_i =
0$, i.e.  writing $t_i$ for the projection to the subspace $\sum_i x_i
= 0$ (the traceless diagonal matrices in our particular
representation), \beq \lambda_k = \sum_i^k t_i.  \eeq For comparison,
the barycenters of the unit $(n-1)$-simplex in the unit-trace matrices
are $b_k := \sum_{j=1}^k e_i/k$, and their translations to the
$(n-1)$-dimensional subspace $\sum_{i=1}^n x_i = 0$, which are the
barycenters of a maximal flag of of (proper) faces of the simplex
$\conv{W.t_1}$ in that space, are just \beq c_k := \sum_{j=1}^l t_j /
k \equiv \lambda_k/k.  \eeq} 

Since we saw that in this construction, we could identify the
fundamental weight $\lambda_1 \in \fa$ whose $W$-orbitope gives
$\pi(B)$ and whose $K$-orbitope gives $B$ in the Madden-Robertson
embedding in a polar representation, as the traceless part of the unit
matrix $e_{11}$, when that polar representation is one corresponding
to a Euclidean Jordan algebra, we see that $B = \conv{(K.(e_{11} -
  I/n))} \equiv \conv{(K.e_{11})} - I/n$ is affinely isomorphic to the
normalized state space $\conv{K.e_{11}}$ of the corresponding
Euclidean Jordan algebra, the isomorphism being given by mapping the
linear hyperplane of traceless symmetric matrices to the affine 
hyperplane of unit-trace ones, by adding $I/n$.


\end{appendix}


\end{document}